\documentclass[%
  aps,%
  prx,%
letterpaper,
  superscriptaddress,%
  reprint,%
  longbibliography,%
  nofootinbib,%
  onecolumn,
  notitlepage]{revtex4-2}
\newcommand{\comment}[1]{}
\usepackage[preset=reset,pkgset=extended,hyperrefdefs={defer,noemail}]{phfnote}
\usepackage{tikz}

\usepackage{phfqit}
\usepackage{phfparen}

\usepackage{mathptmx} 
\usepackage[semibold]{sourcesanspro} 

\usepackage[thmset=rich,smallproofs=false,proofref={off}]{phfthm}

\usepackage{bm}
  \SetSymbolFont{largesymbols}{bold}{OMX}{txex}{b}{n}

\usepackage[reset]{geometry}
\makeatletter\Gm@restore@org\makeatother

\usepackage{./mydefs}

\makeatletter
\newcommand\footnoteref[1]{\protected@xdef\@thefnmark{\ref{#1}}\@footnotemark}
\makeatother

\makeatletter
\newcommand\bibalias[2]{%
  \@namedef{bibali@#1}{#2}%
}

\newcommand\biba@deblank[1]{\romannumeral\biba@deblank@#1/ /} 
\long\def\biba@deblank@#1 /{\biba@deblank@i#1/}
\long\def\biba@deblank@i#1/#2{\z@#1}

\newtoks\biba@toks
\let\bibalias@oldcite\cite
\def\cite{%
  \@ifnextchar[{%
    \biba@cite@optarg%
  }{%
    \biba@cite{}%
  }%
}
\newcommand\biba@cite@optarg[2][]{%
  \biba@cite{[#1]}{#2}%
}
\newcommand\biba@cite[2]{%
  \biba@toks{\bibalias@oldcite#1}%
  \def\biba@comma{}%
  \def\biba@all{}%
  \def\biba@argkeys{}%
  \@for\biba@one@:=#2\do{%
    \edef\biba@one{\expandafter\@firstofone\biba@one@\@empty}%
    \edef\biba@one{\expandafter\biba@deblank\expandafter{\biba@one}}
    \edef\biba@argkeys{\biba@argkeys\biba@comma\biba@one}%
    \@ifundefined{bibali@\biba@one}{%
      \edef\biba@all{\biba@all\biba@comma\biba@one}%
    }{%
      \PackageInfo{bibalias}{%
        Replacing citation `\biba@one' with `\@nameuse{bibali@\biba@one}'
      }%
      \edef\biba@all{\biba@all\biba@comma\@nameuse{bibali@\biba@one}}%
    }%
    \def\biba@comma{,}%
  }%
  \immediate\write\@auxout{\noexpand\bgroup\noexpand\renewcommand\noexpand\citation[1]{}\noexpand\citation{\biba@argkeys}\noexpand\egroup}%
  \edef\biba@tmp{\the\biba@toks{\biba@all}}%
  \biba@tmp%
}
\makeatother

\bibalias{Aberg_06_Quantifying}{Aberg2006arXiv_quantifying}
\bibalias{Agon_19_Subsystem}{Agon2019JHEP_subsystem}
\bibalias{Albert2019arXiv_robust}{Albert2020PRX_robust}
\bibalias{Alhambra2018arXiv_algorithmic}{Alhambra2019Qu_algorithmic}
\bibalias{Alhambra_19_Heat}{Alhambra2019Qu_algorithmic}
\bibalias{Anshu2017arXiv_compression}{Anshu2019IEEETIT_compression}
\bibalias{Anshu2017arXiv_oneshot}{Anshu2019IEEETIT_oneshot}
\bibalias{Anshu2018arXiv_partially}{Anshu2020ITIT_partially}
\bibalias{Bai_22_Towards}{Bai2022JHEP_nonequilibrium}
\bibalias{Balasubramanian2022arXiv_new}{Balasubramanian2022PRD_chaos}
\bibalias{Balasubramanian_20_Quantum}{Balasubramanian2020JHEP_complexity}
\bibalias{Baumgratz_14_Quantifying}{Baumgratz2014PRL_coherence}
\bibalias{Bernamonti_18_Holographic}{Bernamonti2018JHEP_holographic}
\bibalias{Bertoni2020arXiv_algebraic}{Bertoni2020NJP_algebraic}
\bibalias{Biamonte2016arXiv_machinelearning}{Biamonte2017Nat_machinelearning}
\bibalias{Bjelakovic2003arXiv_compression}{Bjelakovic2005QIP_compression}
\bibalias{Boixo}{Boixo2018NatPhys_characterizing}
\bibalias{BookHan02_InfSpecMethods}{BookHan_InfSpecMethods}
\bibalias{BookNielsenChuang2000}{BookNielsenChuang2010}
\bibalias{BookZeng2015arXiv_matter}{BookZeng2019_matter}
\bibalias{Bouland_19_ComputationalPseudorandomness}{Bouland2019arXiv_computational}
\bibalias{Brandao2012CMP_local}{Brandao2016CMP_local}
\bibalias{Brandao2017arXiv_chainAQECC}{Brandao2019PRL_chainAQECC}
\bibalias{Brandao2019arXiv_models}{Brandao2021PRXQ_models}
\bibalias{Brandao_13_Resource}{Brandao2013_resource}
\bibalias{Brandao_21_Models}{Brandao2021PRXQ_models}
\bibalias{Bravyi2017arXiv_shallow}{Bravyi2018Sci_shallow}
\bibalias{Breuckmann2021arXiv_LDPC}{Breuckmann2021PRXQ_LDPC}
\bibalias{Brown2013_decoupling}{Brown2015CMP_decoupling}
\bibalias{Brown2019arXiv_nonClifford}{Brown2020SciAdv_nonClifford}
\bibalias{Brown_16_Complexity}{Brown2016PRD_complexity}
\bibalias{Brown_16_Holographic}{Brown2016PRL_holographic}
\bibalias{Brown_17_Quantum}{Brown2017PRD_negative}
\bibalias{Brown_18_Second}{Brown2018PRD_complexity}
\bibalias{Brown_19_Complexity}{Brown2019PRD_single}
\bibalias{Buscemi2019arXiv_dichotomies}{Buscemi2019Qu_dichotomies}
\bibalias{Caceres_20_Complexity}{Caceres2020JHEP_mixed}
\bibalias{Camargo2020arXiv_entanglement}{Camargo2021PRR_entanglement}
\bibalias{Camargo_21_Entanglement}{Camargo2021PRR_entanglement}
\bibalias{Capel2018arXiv_conditional}{Capel2018JPA_conditional}
\bibalias{Caputa_22_QuantumComplexityAndTopologicalPhases}{Caputa2022PRB_complexity}
\bibalias{Chen_11_Classification}{Chen2011PRB_classification}
\bibalias{Chiribella_17_Microcanonical}{Chiribella2017NJP_gpt}
\bibalias{Chitambar2018arXiv_resource}{Chitambar2019RMP_resource}
\bibalias{Chitambar_19_Quantum}{Chitambar2019RMP_resource}
\bibalias{Chubb2017arXiv_beyond}{Chubb2018Qu_beyond}
\bibalias{Correa2016arXiv_thermometry}{Correa2017PRA_thermometry}
\bibalias{Cotler_17_Chaos}{Cotler2017JHEP_chaos}
\bibalias{Dahlsten2013Non}{Dahlsten2013_nonequilibrium}
\bibalias{Dana_17_Resource}{Dana2017PRA_beyond}
\bibalias{Datta_2013_HypothesisTesting}{Datta2013ITIT_smooth}
\bibalias{DelRio_11_Thermodynamic}{delRio2011Nature}
\bibalias{Dupuis2016arXiv_accumulation}{Dupuis2020CMP_accumulation}
\bibalias{Dupuis2018arXiv_improved}{Dupuis2019IEEETIT_improved}
\bibalias{Dupuis_13_Generalized}{Dupuis2013_DH}
\bibalias{Eisert2021arXiv_entangling}{Eisert2021PRL_entangling}
\bibalias{Eisert_21_Entangling}{Eisert2021PRL_entangling}
\bibalias{Faist2018arXiv_thcapacity}{Faist2019PRL_thcapacity}
\bibalias{Faist2019arXiv_continuous}{Faist2020PRX_continuous}
\bibalias{Faist2019arXiv_impl}{Faist2021CMP_impl}
\bibalias{Faist2022arXiv_timeenergy}{Faist2023PRXQ_timeenergy}
\bibalias{Faist_15_Gibbs}{Faist2015NJP_Gibbs}
\bibalias{Fang2018arXiv_simulation}{Fang2019IEEETIT_simulation}
\bibalias{Garcia2022arXiv_resource}{Garcia2023PNAS_resource}
\bibalias{Gorecki2019arXiv_multiparameter}{Gorecki2020Qu_multiparameter}
\bibalias{Gosset_14_Algorithm}{Gosset2014QIC_algorithm}
\bibalias{Gottesman1997PhD_stab}{PhDGottesman1997_stab}
\bibalias{Gour2017arXiv_entropic}{Gour2018NatComm_entropic}
\bibalias{Gour2018arXiv_entropychannel}{Gour2021PRR_entropychannel}
\bibalias{Gross_Eisert_paper}{Gross2009PRL_most}
\bibalias{Gschwendtner2019arXiv_lowenergies}{Gschwendtner2019JHEP_lowenergies}
\bibalias{Guryanova_16_Thermodynamics}{Guryanova2016NatComm_multiple}
\bibalias{Haferkamp2020arXiv_improved}{Haferkamp2021PRA_improved}
\bibalias{Haferkamp2021arXiv_linear}{Haferkamp2022NPhys_linear}
\bibalias{Haferkamp_21_Linear}{Haferkamp2022NPhys_linear}
\bibalias{Halpern2021arXiv_transport}{YungerHalpern2022nQI_transport}
\bibalias{Hangleiter2023arXiv_computational}{Hangleiter2023RMP_computational}
\bibalias{Harlow2018arXiv_constraints}{Harlow2019PRL_constraints}
\bibalias{Harrow_17_Quantum}{Harrow2017N_computational}
\bibalias{Hayden2017arXiv_frame}{Hayden2021PRXQ_frame}
\bibalias{Hayden_07_Black}{Hayden2007JHEP_blackholes}
\bibalias{HindsMingo2018arXiv_multiple}{HindsMingo2018_multiple}
\bibalias{HindsMingo2018arXiv_superpositions}{HindsMingo2019Qu_superpositions}
\bibalias{Holmes_19_Coherent}{Holmes2019Q_concrete}
\bibalias{Horodecki_02_Are}{Horodecki2002_arethelaws}
\bibalias{Horodecki_03_Local}{Horodecki2003PRL_distributed}
\bibalias{Horodecki_03_Reversible}{Horodecki2003PRA_NoisyOps}
\bibalias{Horodecki_09_Quantum}{Horodecki2009RMP_entanglement}
\bibalias{Horodecki_13_Fundamental}{Horodecki2013_ThermoMaj}
\bibalias{Howard_17_Application}{Howard2017PRL_magic}
\bibalias{HunterJones2017JHEP_random}{HunterJones2018JHEP_random}
\bibalias{HypothesisEntropy}{Datta2013ITIT_smooth}
\bibalias{Jennings2010PRL_arrow}{Jennings2010PRE_arrow}
\bibalias{Ji2017arXiv_pseudorandom}{Ji2018CRYPTO_pseudorandom}
\bibalias{Katariya2020arXiv_geometric}{Katariya2021QIP_geometric}
\bibalias{Kato2016arXiv_Markov}{Kato2019CMP_Markov}
\bibalias{Khandelwal2019arXiv_general}{Khandelwal2020Qu_general}
\bibalias{Knill_98_Power}{Knill1998PRL_oneclean}
\bibalias{Kubica2020arXiv_metrological}{Kubica2021PRL_metrological}
\bibalias{Landauer_61_Irreversibility}{Landauer1961_5392446Erasure}
\bibalias{Li_22_ShortProofsOfLinearGrowth}{Li2022arXiv_short}
\bibalias{Li_22_Wasserstein}{Li2022arXiv_wasserstein}
\bibalias{Lin_19_Cayley}{Lin2019JHEP_cayley}
\bibalias{Liu2021arXiv_approximate}{Liu2023nQI_approximate}
\bibalias{Liu_19_Resource}{Liu2019arXiv_channels}
\bibalias{Losert_82_Counter}{Losert_82_Counterexamples}
\bibalias{Lostaglio2018arXiv_broadcast}{Lostaglio2019PRL_broadcast}
\bibalias{Lostaglio_17_Thermodynamic}{Lostaglio2017NJP_noncommutativity}
\bibalias{Manzano2020arXiv_nonabelian}{Manzano2022PRXQ_nonabelian}
\bibalias{Marvian2018arXiv_distillation}{Marvian2020NC_distillation}
\bibalias{NYH_16_Beyond}{YungerHalpern2016PRE_beyond}
\bibalias{NYH_16_Microcanonical}{YungerHalpern2016NatComm_NATSandNATO}
\bibalias{NYH_17_Toward}{YungerHalpern2017IaI_toward}
\bibalias{NYH_18_Beyond}{YungerHalpern2018JPA_beyond2}
\bibalias{NYH_20_Fundamental}{YungerHalpern2020PRA_photoisomerization}
\bibalias{Nahum2018PRX_operator}{Nahum2018PRX_spreading}
\bibalias{Ng2018arXiv_TO}{Ng2018TiQRbook_resource}
\bibalias{NielsenC10}{BookNielsenChuang2010}
\bibalias{Nielsen_06_Quantum}{Nielsen2006SMag}
\bibalias{Nielsen_99_Conditions}{Nielsen1999PRL}
\bibalias{Ouyang2019arXiv_robust}{Ouyang2022IEEETIT_robust}
\bibalias{PhysRevLett.127.020501}{Eisert2021PRL_entangling}
\bibalias{PhysRevLett.128.081602}{Belin2022PRL_anything}
\bibalias{Piveteau2021arXiv_error}{Piveteau2021PRL_error}
\bibalias{Popescu2018arXiv_applications}{Popescu2018PTRSA_applications}
\bibalias{Popescu_97_Thermodynamics}{Popescu1996PRA_thermodynamics}
\bibalias{Roberts_17_Chaos}{Roberts2017JHEP_chaos}
\bibalias{Rojkov2021arXiv_bias}{Rojkov2022PRL_bias}
\bibalias{Ruan_20_Purification}{Ruan2020arXiv_purification}
\bibalias{Sagawa2019arXiv_asymptotic}{Sagawa2021JPA_asymptotic}
\bibalias{Saha_21_Holographic}{Saha2021PRD_holographic}
\bibalias{Shettell2021arXiv_limits}{Shettell2021NJP_limits}
\bibalias{Sparaciari2016arXiv_workheat}{Sparaciari2017PRA_workheat}
\bibalias{Sparaciari2018arXiv_first}{Sparaciari2020Qu_first}
\bibalias{Sparaciari_17_Resource}{Sparaciari2017PRA_workheat}
\bibalias{Stanford_14_Complexity}{Stanford2014PRD_complexityshock}
\bibalias{Susskind2018PiTP_complexity1}{Susskind2018PiTP_three}
\bibalias{Susskind_16_Computational}{Susskind2016FP_computational}
\bibalias{Susskind_18_Three}{Susskind2018PiTP_three}
\bibalias{Szilard_29_Uber}{Szilard1929ZeitschriftFuerPhysik}
\bibalias{Taranto2021arXiv_cooling}{Taranto2023PRXQ_landauer}
\bibalias{TownsendTeague2023arXiv_floquetifying}{TownsendTeague2023EPTCS_floquetifying}
\bibalias{Veitch2012NJP_negative}{Veitsch_12_Negative}
\bibalias{Wehner2015arXiv_reversibility}{Alhambra2018PRA_reversibility}
\bibalias{Weilenmann2018arXiv_smooth}{Weilenmann2018_smooth}
\bibalias{Winter2015arXiv_tight}{Winter2016CMP_tight}
\bibalias{Winter_16_Operational}{Winter2016PRL_coherence}
\bibalias{Witten2018arXiv_entanglement}{Witten2018RMP_entanglement}
\bibalias{Woods2015arXiv_more}{Woods2019Q_more}
\bibalias{Woods2019arXiv_continuous}{Woods2020Qu_continuous}
\bibalias{Xu2022arXiv_qubitoscillator}{Xu2023PRXQ_qubitoscillator}
\bibalias{Yang2020arXiv_covariant}{Yang2022PRR_covariant}
\bibalias{YungerHalpern2018arXiv_photoisomerization}{YungerHalpern2020PRA_photoisomerization}
\bibalias{YungerHalpern2019arXiv_equilibration}{YungerHalpern2020PRE_equilibration}
\bibalias{YungerHalpern2021arXiv_resource}{YungerHalpern2022PRA_resource}
\bibalias{YungerHalpern_2022_Uncomplexity}{YungerHalpern2022PRA_resource}
\bibalias{Zhou2019arXiv_strongly}{Zhou2021PRX_strongly}
\bibalias{Zhou2020arXiv_covariant}{Zhou2021Qu_covariant}
\bibalias{brandao2016local}{Brandao2016CMP_local}
\bibalias{haferkamp2022random}{Haferkamp2022Q_random}
\bibalias{lostaglio2015description}{Lostaglio2015NC_beyond}
\bibalias{vanderMeer2017Smoothed}{Meer2017PRX_smoothed}

\begin{document}
\title{Complexity-constrained quantum thermodynamics}

\author{Anthony~Munson} 
\affiliation{Joint Center for Quantum Information and Computer Science, NIST and University of Maryland, College Park, Maryland 20742, USA}
\affiliation{Department of Physics, University of Maryland, College Park, Maryland 20742, USA}
\author{Naga~Bhavya~Teja~Kothakonda} 
\affiliation{Grup d'Informaci\'o Qu\`antica, Departament de F\'isica, Universitat Aut\`onoma de Barcelona, 08193 Bellaterra (Barcelona), Spain}
\affiliation{Dahlem Center for Complex Quantum Systems, Freie Universit\"{a}t Berlin, 14195 Berlin, Germany}
\author{Jonas~Haferkamp} 
\affiliation{School of Engineering and Applied Sciences, Harvard University, Massachusetts 02134, USA}
\author{Nicole~Yunger~Halpern}
\affiliation{Joint Center for Quantum Information and Computer Science, NIST and University of Maryland, College Park, Maryland 20742, USA}
\affiliation{Department of Physics, University of Maryland, College Park, Maryland 20742, USA}
\affiliation{Institute for Physical Science and Technology, University of Maryland, College Park, Maryland 20742, USA}
\author{Jens~Eisert}
\affiliation{Dahlem Center for Complex Quantum Systems, Freie Universit\"{a}t Berlin, 14195 Berlin, Germany}
\author{Philippe~Faist}
\email{philippe.faist@fu-berlin.de}
\affiliation{Dahlem Center for Complex Quantum Systems, Freie Universit\"{a}t Berlin, 14195 Berlin, Germany}

\date{\today}
\begin{abstract}
  Quantum complexity measures the difficulty of realizing a quantum process, such as preparing a state or implementing a unitary.
  We present an approach to quantifying the thermodynamic resources required to implement a process if the process's complexity is restricted. 
  We focus on the prototypical task of information erasure, or Landauer erasure, wherein an $n$-qubit memory is reset to the all-zero state.
  We show that the minimum thermodynamic work required to reset an arbitrary state in our model, via a complexity-constrained process, is quantified by the state's \emph{complexity entropy}.
  The complexity entropy therefore quantifies a trade-off between the work cost and complexity cost of resetting a state.
  If the qubits have a nontrivial (but product) Hamiltonian, the optimal work cost is determined by the \emph{complexity relative entropy}.
  The complexity entropy quantifies the amount of randomness a system appears to have to a computationally limited observer.
  Similarly, the complexity relative entropy quantifies such an observer's ability to distinguish two states.
  We prove elementary properties of the complexity (relative) entropy.
  In a random circuit---a simple model for quantum chaotic dynamics---the complexity entropy transitions from zero to its maximal value around the time corresponding to the observer's computational-power limit.
  Also, we identify information-theoretic applications of the complexity entropy.
  The complexity entropy quantifies the resources required for data compression if the compression algorithm must use a restricted number of gates.
  We further introduce a \emph{complexity conditional entropy}, which arises naturally in a complexity-constrained variant of information-theoretic decoupling.
  Assuming that this entropy obeys a conjectured chain rule, we show that the entropy bounds the number of qubits that one can decouple from a reference system, as judged by a computationally bounded referee.
  Overall, our framework extends the resource-theoretic approach to thermodynamics to integrate a notion of \emph{time}, as quantified by \emph{complexity}.
\end{abstract}
\maketitle

\twocolumngrid

\section{Introduction}

Quantum complexity is drawing increasing interest across physics, from many-body physics to quantum gravity~\cite{Ahoronov_11_ComplexityCommutingLocalHamiltonians,%
  Caputa_22_QuantumComplexityAndTopologicalPhases,%
  Brown_18_Second,Susskind2018arXiv_ThreeLectures}.
For the purposes of our work, the quantum complexity of a unitary operation (respectively, a quantum state) is the minimum number of operations required to implement the unitary (respectively, to prepare the state).
Each operation is chosen from a given set of elementary operations (e.g., a universal set of two-qubit gates).
In condensed-matter physics, preparing a topologically ordered state requires a circuit of sufficient complexity to spread information throughout a system~\cite{Chen_11_Classification,%
  Ahoronov_11_ComplexityCommutingLocalHamiltonians,%
  Hastings_11_TopologicalOrder,%
  Schwarz_13_TopologicalPEPS,%
  Miller_18_Latent,%
  Ali_20_Post,%
  Liu_20_Circuit,%
  Xiong_20_Nonanalyticity,%
  Caputa_22_QuantumComplexityAndTopologicalPhases}.
Near-term quantum devices aim to prepare states of sufficient complexities to offer quantum advantages attributable, for example, to the hardness of sampling classically from such states~\cite{%
  Boixo,%
  Arute_19_Quantum,%
  Hangleiter2023RMP_computational,%
  Deshpande_18_PhaseTransitions}.
Quantum complexity recently gained significance in the context of the \emph{anti--de Sitter--space/conformal-field-theory} (AdS/CFT) correspondence in high-energy physics: In a prominent conjecture, the complexity of the field-theoretic state dual to a wormhole connecting two black holes is proportional to the wormhole’s length~\cite{Susskind2014arXiv_notenough,Brandao_21_Models,Haferkamp_21_Linear,Li_22_ShortProofsOfLinearGrowth,Brown_23_Quantum,Susskind_16_Computational,Stanford_14_Complexity,Brown_16_Complexity,Brown_18_Second,Susskind2018arXiv_ThreeLectures,Bouland_19_ComputationalPseudorandomness,Brown_16_Holographic,BigComplexity,PhysRevLett.127.020501,PhysRevLett.128.081602,Balasubramanian_21_Complexity,YungerHalpern_2022_Uncomplexity,Belin2022PRL_anything}.

The relevance of quantum complexity to AdS/CFT motivates connections to thermodynamics.
Brown and Susskind posited that the CFT state's complexity should tend to increase, formulating a ``second law of complexity''~\cite{Brown_18_Second}.
Bai \emph{et al.\@} extended the second law of complexity by proving fluctuation relations mirroring Jarzynski's equality in statistical mechanics~\cite{Bai_22_Towards}.
A resource theory of \emph{uncomplexity}---a state's closeness to a simple tensor product---was furthermore established in Ref.~\cite{YungerHalpern_2022_Uncomplexity}.
Quantum complexity also appears connected to ergodicity and quantum chaos: complexity is believed to grow linearly for long times under typical quantum chaotic dynamics;
complexity would thereby provide a universal measure for how long a chaotic system has evolved~\cite{Susskind2018arXiv_ThreeLectures,Brandao_21_Models,%
  Haferkamp_21_Linear,Li_22_ShortProofsOfLinearGrowth,%
  Brown_23_Quantum,chen2024incompressibility}.
In contrast, standard correlation functions and entanglement entropies typically reach their equilibrium values after short times~\cite{Susskind2018arXiv_ThreeLectures}.

Our main goal is to identify which state transformations in quantum thermodynamics can be effected by processes, and probed with measurement effects, utilizing limited computational resources, hence of limited complexity.
In particular, we seek to connect quantum complexity and entropy as follows.
In conventional thermodynamics, the minimum amount of work needed to transform one equilibrium state into another, via exchanges of heat at a fixed temperature, is determined by the energy and entropy differences between the initial and final states.
The complexity of a many-body system's evolution is upper-bounded by the evolution's duration.
An observer with little computational power typically cannot distinguish a highly complex pure state from a highly entropic mixed state.
A phenomenological description of possible state transformations under short-time evolutions, from the viewpoint of such an observer, should not distinguish highly entropic initial states from highly complex pure states.
The states become distinguishable if observed over sufficiently long timescales or through sufficiently complex observables.
This fact invites us to define thermodynamic potentials for determining which state transformations can be effected under complexity limitations.
The roles of these potentials mirror the role of entropy in standard thermodynamics.
Using the potentials, we address our general question: \emph{how does one formulate thermodynamics at a given complexity scale---for complexity-constrained processes and observers?}

Our analysis relies on recent information-theoretic frameworks for quantum thermodynamics.
The relevance of an observer's information, or knowledge, in thermodynamics was significantly clarified by the pioneering works of Szil\'ard~\cite{Szilard1929ZeitschriftFuerPhysik}, Landauer~\cite{Landauer1961_5392446Erasure}, and Bennett~\cite{Bennett1982IJTP_ThermodynOfComp}.
Landauer argued that erasing a bit of information dissipates an amount of heat $\geq k_{\mathrm B} T \log(2)$~\cite{Landauer1961_5392446Erasure}. The $k_{\mathrm B}$ denotes Boltzmann's constant, $T$ denotes an environment's temperature, and the logarithm is base-$e$.
This observation led to Bennett's resolution of Maxwell's demon paradox~\cite{Bennett1982IJTP_ThermodynOfComp,%
  Bennett2003_NotesLP} and helped extend quantum thermodynamics to far-from-equilibrium systems~\cite{Goold2016JPA_review,%
BookBinder2018_ThermoQuantumRegime,BookSagawa2022_entropy,%
Faist2018PRX_workcost}.
A common model for quantum thermodynamics is the resource theory of thermal operations~\cite{Horodecki2003PRA_NoisyOps,%
  Brandao2013_resource,Horodecki2013_ThermoMaj,%
  Brandao2015PNAS_secondlaws,%
  Goold2016JPA_review,BookBinder2018_ThermoQuantumRegime}:
one assumes that energy-conserving interactions with a fixed-temperature heat bath are the only operations performable without external resources, such as thermodynamic work.
Using this model, one can determine whether a state $\rho$ can transform into a state $\sigma$, for large classes of states. 
The answer can be cast in terms of a family of entropy measures termed the R\'enyi-$\alpha$ relative entropies~\cite{Brandao2015PNAS_secondlaws}.
A closely related model captures how many natural, or spontaneous, dynamics have a particular fixed point.
Using this model, one can quantify, e.g., the thermodynamic work required to implement a general quantum process~\cite{Janzing2000_cost,Faist2015NatComm,Faist2018PRX_workcost}.
A key observation is that a process's work cost is tied fundamentally to its logical irreversibility: information erasure is costly, but energy-conserving unitary operations can, in principle, be implemented at zero work cost~\cite{Landauer1961_5392446Erasure,Bennett1982IJTP_ThermodynOfComp,Bennett2003_NotesLP}.
We introduce a simple model that captures the main features of the above models and that forms the basis for our analysis. Our model centers on an $n$-qubit memory register governed by a completely degenerate Hamiltonian.
The following primitive operations are performable (\cref{fig:SimpleModelInformationThermo}): (i) the reset of one qubit from any state to a standard state $\ket0$, costing one unit of work; (ii) the preparation of one qubit in the maximally mixed state from $\ket0$, extracting one unit of work; and (iii) the implementation of one two-qubit unitary gate, costing one unit of complexity.
(Our model neglects any complexity cost of single-qubit erasure~\cite{Taranto2021arXiv_cooling}, to separate two types of costs: the thermodynamic work cost associated with reducing the entropy of a single qubit, and the computational cost of exchanging information between the qubits.
The former operations involve thermodynamic work and act on individual qubits. The latter operations involve two-qubit interactions that tend to spread quantum information throughout the system.
Our model generalizes to elementary operations that cost both work and complexity.)

\begin{figure}
  \centering
  \includegraphics{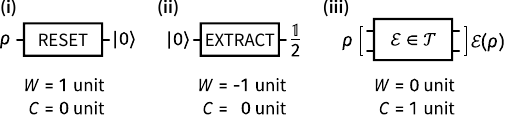}
  \caption{
    A simple model combining quantum-information thermodynamics and complexity:
    An $n$-qubit memory register is governed by a completely degenerate Hamiltonian.
    Processes consist of the primitive operations \labelcref{item:primitive-reset,item:primitive-extract,item:primitive-operation-computation}.
    Work and complexity costs are defined as the sums of the corresponding costs over the operations.
    \textbf{(a)}~A qubit can be reset from any state $\rho$ to a standard state $\ket0$.
    The resetting costs one unit of work ($W$), by Landauer's principle~\protect\cite{Landauer1961_5392446Erasure}.
    This primitive has no complexity cost ($C$).
    \textbf{(b)}~One unit of work can be extracted from a qubit in the state $\ket0$.
    The extraction leaves the qubit maximally mixed~\protect\cite{Szilard1929ZeitschriftFuerPhysik}.
    \textbf{(c)}~Operations $\mathcal{E}$ are chosen from a set $\mathcal{T}$ of elementary computations.
    Each $\mathcal{E}$ costs one unit of complexity but no work.
    A natural choice for $\mathcal{T}$ is the set of two-qubit unitary operations, potentially subject to connectivity constraints.
    Choices of $\mathcal{T}$ natural to thermodynamic applications preserve thermal states.
  }
  \label{fig:SimpleModelInformationThermo}
\end{figure}

We first revisit the standard setting of \emph{information erasure}, or \emph{Landauer erasure}.
Erasure has clarified information's role in thermodynamics; we now use erasure to clarify complexity's role in quantum thermodynamics.
We consider an $n$-qubit system with a product Hamiltonian $H$ and a zero-energy ground state.
We prove that the minimum work required to reset an arbitrary state $\rho$ to the ground state, using at most a fixed number of thermodynamic processes in our model, is given by the \emph{complexity relative entropy} of $\rho$ relative to a thermal state~\cite{YungerHalpern_2022_Uncomplexity}.
The complexity relative entropy quantifies two states' distinguishability, as judged by an observer with limited computational power.
If $H=0$, this quantity reduces to the \emph{complexity entropy}~\cite{YungerHalpern_2022_Uncomplexity} and quantifies how random a state appears to such an observer.
A variant of the complexity (relative) entropy appeared in our Ref.~\cite{YungerHalpern_2022_Uncomplexity} to quantify the number of pure qubits extractable from a state, in the resource theory of uncomplexity.
The present work further applies the complexity relative entropy to thermodynamics: we extend the pure-qubit extraction protocol of Ref.~\cite{YungerHalpern_2022_Uncomplexity} to general protocols for thermodynamic information erasure under complexity limitations.

Our result quantifies a trade-off between the work and complexity costs of erasing an $n$-qubit state $\rho$ (\cref{fig:TradeOffWorkComplexityCostErasure}).
\begin{figure}
  \centering
  \includegraphics{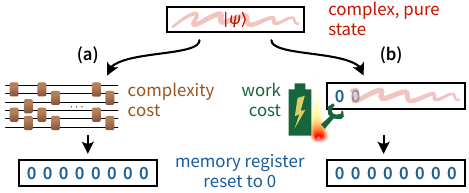}
  \caption{
    An example trade-off between the work cost and complexity cost of resetting an $n$-qubit memory register initiated in a highly complex pure state $\ket\psi$.
    \textbf{(a)}~A unitary operation resets the memory to the all-zero state. 
    In principle, the unitary costs no thermodynamic work.
    Many gates are necessary to implement the unitary, incurring a high complexity cost.
    \textbf{(b)}~To reset the memory without gates, one can erase every qubit [via primitive \ref{item:primitive-reset} in \cref{fig:SimpleModelInformationThermo}].
    This operation costs $n$ units of work total.
    The \emph{complexity entropy} quantifies the work--complexity trade-off of erasing an arbitrary state $\rho$.
    The complexity entropy was introduced in Ref.~\protect\cite{YungerHalpern_2022_Uncomplexity}; this work elucidates its properties and applications.
  }
  \label{fig:TradeOffWorkComplexityCostErasure}
\end{figure}
Suppose that $H=0$, and consider a highly complex pure state $\ket\psi$.
In standard models for information thermodynamics, one can implement unitary operations on an energy-degenerate memory register at zero work cost, since such operations are logically reversible~\cite{Bennett1973IBMJRD_LogRevComp,Egloff2015NJP_measure,Horodecki2013_ThermoMaj}.
Therefore, $\ket\psi$ can, in principle, be transformed into the all-zero state $\ket{0^n}$ at no work cost.
Yet, this transformation requires many gates, incurring a high \emph{complexity cost}.
An alternative procedure would reset each memory qubit with a thermodynamic reset operation.
This procedure costs $n\,k_{\mathrm B} T\log(2)$ units of work but no complexity.
The complexity entropy quantifies the trade-off between the work and complexity costs of erasing $\ket\psi$.

Bennett \emph{et al.\@} analyzed a trade-off between complexity and work in the context of classical bit erasure~\cite{Bennett_93_Thermodynamics}.
Kolmogorov complexity, rather than quantum complexity, is relevant to their problem.
The connection between Kolmogorov complexity and thermodynamics was further cemented in Zurek's work~\cite{Zurek1989_cost} and in Refs.~\cite{Caves1993PRE_information,Mora2006arXiv_qKolmogorov,Baez2010MSCS_algo}.
Kolmogorov complexity and quantum complexity quantify the size of a state's ``smallest description'' in different ways.
Kolmogorov complexity quantifies the size of the smallest program that generates the state (regardless of the program's runtime).
In contrast, quantum complexity measures the shortest runtime of a program that generates the state (regardless of whether the program has a compact representation).

The remainder of this work is dedicated to the analysis of more-general thermodynamic processes, including erasure in systems with nontrivial Hamiltonians, and to a deeper study of the complexity entropy's properties, uses, and extensions.
We study the work costs of general state transformations in quantum thermodynamics where complexity limitations restrict which processes and measurements one can implement.
Furthermore, we evidence the complexity entropy's broad relevance to quantum information theory.
First, we present properties of the complexity entropy and bound it using well-known complexity measures and an entanglement measure.
Then, we demonstrate the complexity entropy's relevance in random circuits, a simple model of quantum chaotic dynamics.
Last, we apply our complexity-entropy measures to information-theoretic tasks with complexity constraints.

What is the work cost of a general state transformation $\rho\to\rho'$, as in the model of \cref{fig:SimpleModelInformationThermo}?
We consider the minimum work cost of any process, of complexity at most $r \geq 0$, that maps $\rho$ to a state indistinguishable from $\rho'$.
Here, distinguishability is judged by an observer who possesses some bounded computational power $R \geq 0$.
This work cost extends the complexity relative entropy's use to general state transformations.
Quantifying the work cost addresses our main goal of establishing which state transformations are possible in quantum thermodynamics for complexity-constrained agents.
Considering regimes in which $R$ is either extremely high or extremely low, we identify cases in which bounding the work cost is sometimes possible.

Our complexity-entropy measures obey certain properties expected of entropies.
For instance, the complexity entropy achieves its maximum value on a maximally mixed state. Also, the complexity relative entropy decreases monotonically under partial traces.
Our measures lack other common properties, such as invariance under unitaries and monotonicity under completely positive, trace-preserving maps (a data-processing inequality): by construction, the complexity (relative) entropy is sensitive to a state's complexity, which can change under arbitrary quantum operations.
Moreover, the complexity entropy obeys bounds involving known complexity measures: one bound involves the \emph{strong complexity} of Ref.~\cite{Brandao2021PRXQ_models}; and the other, the \emph{approximate circuit complexity} (the minimum complexity of any unitary that approximately prepares a target state).
For a one-dimensional (1D) chain of qubits, the complexity entropy obeys also a bound involving an entanglement measure defined in terms of the quantum mutual information.

We argue that the complexity entropy can quantify chaotic behavior in a quantum many-body system.
Consider a suitably chaotic many-body system initialized in a pure product state.
The state's complexity entropy should remain low until the timescale matching an observer's computational power.
Beyond this timescale, the complexity entropy should be close to maximal.
We employ and extend the results of Refs.~\cite{Brandao2016CMP_local,Brandao2021PRXQ_models} to prove a corresponding statement about the output of a random quantum circuit applied to $\ket{0^n}$, an $n$-qubit system's all-zero state.
Suppose that an observer can measure only observables of complexities $\leq r$.
The complexity entropy, we show, transitions from zero to $n-O(1)$ when the number of gates reaches $\approx r$ (\cref{fig:ComplexityEntropyRandomCircuits}).
\begin{figure}
  \centering
  \includegraphics{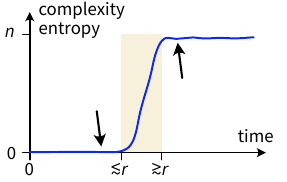}
  \caption{
    The complexity entropy of a state outputted by a random circuit, relative to an observer who can implement only measurement effects of complexities $\leq r$.
    We determine the complexity entropy's evolution using techniques from Refs.~\protect\cite{Brandao2016CMP_local,%
      Brandao2021PRXQ_models,haferkamp2022random}.
    After a time $t \lesssim r$, the state has a low enough complexity that the observer can verify the state's purity via measurement.
    After $t\gtrsim r$, the state is complex enough to be nearly indistinguishable from a highly entropic state, according to any measurement whose complexity is $\leq r$.
    }
  \label{fig:ComplexityEntropyRandomCircuits}
\end{figure}
For a random circuit of $\lesssim r$ gates, the observer can apply the inverse circuit to recover the all-zero state and so to ascertain that the random circuit's output has a low entropy.
For a random circuit of $\gtrsim r$ gates, the observer cannot distinguish the circuit's output from a maximally mixed state, using observables of complexities $\leq r$~\cite{Brandao2021PRXQ_models,chen2024incompressibility}.
Accordingly, the output has a complexity entropy of at least $n-O(1)$.
Reference~\cite{Brandao2021PRXQ_models} quantifies a random circuit's complexity by applying a powerful tool in random-circuit analysis: a \emph{unitary $k$-design}~\cite{Brandao2016CMP_local,Roberts2017JHEP_chaos}.
The lower bound's linear scaling comes from recently improved bounds for the design order achieved by random circuits~\cite{chen2024incompressibility} over existing bounds~\cite{brandao2016local,haferkamp2022random}.

We use our complexity-entropy measures to bound the optimal efficiencies of information-theoretic tasks performed under complexity limitations.
Thermodynamic erasure is related to data compression, the storage of information in the smallest possible register~\cite{BookFeynmanLecturesOnComputation1996,%
  Dahlsten2011NJP_inadequacy,delRio2011Nature}:
to erase a memory register, one can first perform data compression to reduce the number of qubits that need resetting.
The complexity entropy quantifies the resources required to compress a quantum state into the fewest qubits possible, via any limited-complexity procedure.

After addressing data compression, we study \emph{decoupling}, or ensuring that an agent's multiqubit system becomes maximally mixed and uncorrelated with a referee's system, $R$~\cite{Hayden2008OSID_decoupling,Abeyesinghe2009_MotherProtocols,%
  Dupuis2014CMP_decoupling,Szehr2011_decoupling}.
Assuming a conjecture about a chain rule for the complexity entropy, we lower-bound the number of qubits that an agent must discard to obtain a state indistinguishable, by a complexity-restricted referee, from a maximally mixed state uncorrelated with $R$.
The lower bound is given by a conditional variant of the complexity entropy. 
One can interpret the variant as a \emph{complexity-aware conditional min-entropy}.

Our paper is organized as follows.
In \cref{sec:setting}, we introduce our framework, background information, and the complexity entropy.
In \cref{sec:main-erasure}, we bound the work cost of information erasure subject to complexity restrictions.
In \cref{sec:main-general-state-transformations}, we generalize our analysis to arbitrary state transformations.
In \cref{sec:InformationTheoreticApplicationsComplexityEntropy}, we present the complexity entropy's information-theoretic properties and applications.
We conclude in \cref{sec:Discussion}.

\section{Setting}
\label{sec:setting}

In \cref{sec:setup-thermo-with-primitive-operations}, we introduce our setup and allowed operations.
In \cref{sec_One_shot}, we review \emph{one-shot information theory}, which quantifies the efficiencies of thermodynamic tasks performed on finite-size systems, with arbitrary success probabilities.
In \cref{sec_Comp_ent}, we introduce the complexity entropy and complexity relative entropy.

\subsection{Thermodynamic framework}
\label{sec:setup-thermo-with-primitive-operations}

We consider a system of $n$ noninteracting qubits.
Qubit $i$ evolves under a Hamiltonian $H_i$, and the entire system under the Hamiltonian $H = H_1 + H_2 + \cdots + H_n$.
For simplicity, we suppose that $H_i\ket{0}_i = 0$ and $H_i\ket{1}_i = E_i^{(1)} \ket{1}_i$, with $E_i^{(1)}\geq 0$.
The all-zero bit string $\ket{0^n}$ labels the zero-energy ground state.
For a fixed inverse temperature $\beta = 1/k_{\mathrm B} T$, qubit $i$ has the thermal state
\begin{align}
  \gamma_i = \Gamma_i/Z_i \ ,
  \label{eq:setting-gammai-Gammai-Zi}
\end{align}
wherein $\Gamma_i = \ee^{-\beta H_i}$ and $Z_i = \tr`*(\Gamma_i) = 1 + \ee^{-\beta E_i^{(1)}}$.
The thermal qubit's free energy is given by
\begin{align}
  F_i = -k_{\mathrm B} T \log`*(Z_i) \ .
  \label{eq:setting-Fi}
\end{align}

The \emph{thermal operations} model possible transformations in quantum thermodynamics, for a system in controlled contact with a heat bath~\cite{Brandao2013_resource,%
  Horodecki2013_ThermoMaj,%
  BookBinder2018_ThermoQuantumRegime}.
Let $S$ denote a system with a Hamiltonian $H_S$.
A thermal operation $\Phi_S$ is defined as any operation of the form
\begin{align}
 \rho \mapsto \Phi_S(\rho) = \tr_B`*( V_{SB} \, `*[ \rho \otimes \gamma_B ] \, V_{SB}^\dagger ) \ .
\end{align}
$B$ denotes any auxiliary system with a Hamiltonian $H_B$.
$V_{SB}$ denotes any unitary that strictly conserves energy: $[V_{SB}, H_S + H_B] = 0$.
For an arbitrary system $X$ governed by a Hamiltonian $H_X$, the thermal state $\gamma_X = \ee^{-\beta H_X} / \tr`*(\ee^{-\beta H_X})$ is defined similarly to~\eqref{eq:setting-gammai-Gammai-Zi}.

Another general class of operations consists of the \emph{Gibbs-preserving maps}~\cite{Janzing2000_cost}.
Let $S$ denote a system with a Hamiltonian $H_S$.
A Gibbs-preserving map $\mathcal{F}$ is any completely positive, trace-preserving map that satisfies $\mathcal{F}(\gamma_S) = \gamma_{S'}$.
The Gibbs-preserving maps form a larger class than the thermal operations~\cite{Faist2015NJP_Gibbs}.
It remains unclear what resource costs are necessary for implementing a general Gibbs-preserving map, using other standard thermodynamic operations~\cite{IQOQI_OpenProblem_GPM}.
When $H_S = 0$, $\gamma_S$ is proportional to the identity operator $\Ident_S$, so every Gibbs-preserving map $\mathcal{F}$ is a unital map $`*[\mathcal{F}(\Ident_S) = \Ident_S]$ and vice versa.

Inspired by the resource theory of thermodynamics, we use a model in which all processes consist of primitive operations.
A process is specified as a sequence $\mathcal{P} \coloneqq `*(\mathcal{E}_1, \ldots, \mathcal{E}_m)$ of primitive operations.
The sequence implements the operation $\mathcal{E}_{\mathcal{P}} \coloneqq \mathcal{E}_m \cdots \mathcal{E}_1$.
Each primitive operation $\mathcal{E}_i$ has a complexity cost $C`*(\mathcal{E}_i)$ and a work cost $W`*(\mathcal{E}_i)$.
Accordingly, each process $\mathcal{P}$ has a complexity cost $C`*(\mathcal{P})$ and a work cost $W`*(\mathcal{P})$, which are sums of the primitive operations' costs
\begin{align}
    C`*(\mathcal{P}) \coloneqq \sum_{i = 1}^m C`*(\mathcal{E}_i) \ ,
    \; \; \text{and} \; \;
    W`*(\mathcal{P}) \coloneqq \sum_{i = 1}^m W`*(\mathcal{E}_i)  \ .
\end{align}
One might refer to $C`(\mathcal{P})$ instead as the \emph{circuit size} of $\mathcal{P}$.
Suppose that $\mathcal{P}$ implements a unitary operation $\mathcal{U}$.
$C(\mathcal{P})$ might differ from the complexity $C`*(\mathcal{U})$ of $\mathcal{U}$: there might exist a process that also implements $\mathcal{U}$ but consists of fewer primitive operations.

The primitive operations on an $n$-qubit system, given a background inverse temperature $\beta=1/k_{\mathrm B} T$, are the following (\cref{fig:SimpleModelInformationThermo}):
\begin{enumerate}[label=(\roman*)]
    \item\label{item:primitive-reset} The \textsc{reset} operation: Qubit $i$ can be brought from an arbitrary state to the ground state.
    This operation has the work cost $W_{\textsc{reset},\,i} = -F_i = k_{\mathrm B} T\log(Z_i)$ and no complexity cost: $C_{\textsc{reset},\,i} = 0$.

    \item\label{item:primitive-extract} The \textsc{extract} operation: Work $\abs{F_i}$ can be extracted as qubit $i$ is brought from $\ket0_i$ to the thermal state $\gamma_i$.
    This operation has the work cost $W_{\textsc{extract},\,i} = F_i = - W_{\textsc{reset},i}$ and no complexity cost: $C_{\textsc{extract},\,i} = 0$.
    
    \item\label{item:primitive-operation-computation}
    Every operation $\mathcal{E}$ chosen from a fixed set $\mathcal{T}$ of elementary computations costs one unit of complexity, $C(\mathcal{E}) = 1$, and no work: $W(\mathcal{E}) = 0$.
    If $H_i=0$ for all $i$, then natural choices of $\mathcal{T}$ include arbitrary two-qubit unitary gates with arbitrary connectivities.
    In this case, $\mathcal{T} \simeq \bigcup_{(i,j)} \SU(4)_{i,j}$, wherein each $(i,j)$ denotes a pair of qubits.
    If $H_i \neq 0$, natural choices of $\mathcal{T}$ include the two-qubit thermal operations and the two-qubit Gibbs-preserving maps, with arbitrary connectivities.
\end{enumerate}

The work expended on the \textsc{reset}, and the work extracted via \textsc{extract}, naturally generalize Landauer's bound to noninteracting qubits.
The \textsc{reset} and \textsc{extract} operations' ideal work costs are $W_{\textsc{reset},\,i} = -W_{\textsc{extract},\,i} = -F_i$ in standard thermodynamic models, including the resource theory of thermal operations~\cite{Horodecki2013_ThermoMaj,Brandao2015PNAS_secondlaws,%
Alicki2004_hamiltonian,Frenzel2014PRE_revisited,Skrzypczyk2014NComm_individual}.
The \textsc{reset} operation can be applied to any input state, not only a thermal state.
Its deterministic work cost, $-F_i$, can be viewed as equal to the worst-case work cost of resetting any particular input state.
If a qubit has a completely degenerate Hamiltonian, $H_i = 0$, these work costs reduce to Landauer's bound: $-F_i = k_{\mathrm B} T\log(2)$.

The set $\mathcal{T}$ of elementary computations enables us to define a protocol's complexity cost.
To separate elementary thermodynamic operations from elementary computing operations, $\mathcal{T}$ should contain only operations to which one can reasonably assign no work costs.
Crucially, we choose the elementary computations to act nontrivially on just two qubits.
This property ensures that several elementary computations are needed to propagate information throughout the system and many more are needed to generate complex quantum states.

For energy-degenerate qubits ($H_i = 0$ for all $i$), a meaningful choice of $\mathcal{T}$ is a universal set of two-qubit unitary gates: the primitive operation~\ref{item:primitive-operation-computation} then enables quantum computation with unitary circuits.
This choice is undesirable if the qubits have nontrivial Hamiltonians, however: some two-qubit unitaries require work and so should be excluded from $\mathcal{T}$ (for example, a swap of two different-energy levels).
One could choose for $\mathcal{T}$ to consist
of two-qubit unitaries that commute with the total Hamiltonian.
Unfortunately, this choice typically leads to an unreasonably restricted set.
The energy-conserving unitaries on two qubits with distinct Hamiltonians is limited to unitaries that are diagonal with respect to the product energy eigenbasis. 
This choice could render impossible some natural operations that might otherwise be implementable with ancillary qubits and nondiagonal two-qubit unitaries.
One such operation is a partial thermalization of a qubit: $\rho \to \rho/2 + \gamma_i/2$.

When the qubits have nontrivial Hamiltonians, we include in $\mathcal{T}$ some nonunitary two-qubit operations to remedy the above problem.
We still ensure that the $\mathcal{T}$ elements cost no work in standard thermodynamic frameworks.
Yet $\mathcal{T}$ can include operations that are not implementable in practice.
For instance, they can require unreasonable control.
In such a case, the conclusions about the operations' work and complexity costs imply
lower bounds applicable to every setting in which
the physically implementable elementary computations form a smaller set $\mathcal{T}_0$.
(These bounds might not be tight.)
In fact, choosing a larger $\mathcal{T}$ ensures that any such lower bounds apply to a broader range of possibilities for $\mathcal{T}_0$.
In this spirit, we identify two $\mathcal{T}$s that are large and likely include all natural possibilities for $\mathcal{T}_0$:
two-qubit thermal operations and two-qubit Gibbs-preserving maps.

For degenerate Hamiltonians, the choice of $\mathcal{T}$ as the set of two-qubit unitaries does not suffer the same problems as those that apply to sets of unitary operations for nondegenerate Hamiltonians.
Alternative $\mathcal{T}$ choices include the sets of all the two-qubit noisy operations and of all the two-qubit unital maps. 
These operations are the two-qubit thermal operations and Gibbs-preserving maps, respectively, if the qubits have degenerate Hamiltonians.
However, the two-qubit noisy operations and unital maps likely include operations that demand unreasonable control requirements, making the resource-cost lower bounds loose.

\subsection{One-shot entropy measures in quantum thermodynamics}
\label{sec_One_shot}

A fundamental connection between thermodynamics and statistical mechanics is the identification, under suitable conditions, of Clausius' thermodynamic entropy with the von Neumann entropy $\HH{\rho}$ of a quantum state $\rho$, 
\begin{align}
    \HH{\rho} \coloneqq -\tr`(\rho\log\rho)\ .
\end{align}
In the information-theoretic approach to thermodynamics, beyond the traditional regime concerning many copies of a system, many thermodynamic tasks have work costs inaccurately represented by the von Neumann entropy and the standard free energy~\cite{Dahlsten2011NJP_inadequacy,delRio2011Nature,Dahlsten2013Non,Horodecki2013_ThermoMaj,lostaglio2015description,NYH2018Maximum,Wei_2017,YungerHalpern2016PRE_beyond,vanderMeer2017Smoothed,Ng_2017,Guarnieri2019Quantum}.  
Instead, these work costs are quantified with \emph{one-shot entropy measures}, such as the relative entropies~\cite{PhDRenner2005_SQKD,Datta2009IEEE_minmax,NYH2015Introducing} and R\'enyi-$\alpha$ entropies~\cite{Brandao2015PNAS_secondlaws}, including the min- and max-entropies.

We focus on the \emph{hypothesis-testing relative entropy}, which interpolates between the min- and max-relative entropies~\cite{Hiai_91_Proper,Dupuis_13_Generalized,HypothesisEntropy,Hayashi,WatrousScripts,Wang2012PRL_oneshot,Tomamichel2013_hierarchy}.
The hypothesis-testing relative entropy is defined for a quantum state $\rho$, a positive-semidefinite operator $\Gamma$, and an intolerance parameter $\eta \in (0,1]$:
\begin{align}
  \DHyp[\eta]{\rho}{\Gamma} \coloneqq
  -\log `*( \min_{\substack{0\leq Q\leq\Ident\\ \tr`(Q\rho) \geq \eta }}
  `*{ \frac{ \tr`*(Q\Gamma) }{\eta} } ) \ .
  \label{eq:setting-defn-DHyp}
\end{align}
The \emph{hypothesis-testing entropy} of a quantum state $\rho$, for $\eta \in (0,1]$, is defined as
\begin{align}
  \HHyp[\eta]{\rho} \coloneqq -\DHyp[\eta]{\rho}{\Ident} \ .
  \label{eq:setting-defn-HHyp}
\end{align}

The entropy measures here have units of nats, rather than bits, because our definitions have base-$e$, rather than base-2, logarithms.
One can convert between bits and nats via $\text{(no. of bits)} = \text{(no. of nats)} / \log(2)$.
Our convention introduces factors of $\log(2)$ in our results for qubit systems.
However, the convention yields familiar forms for thermodynamic relations between entropy and quantities such as the free energy.

The hypothesis-testing entropy quantifies how well one can distinguish between quantum states $\rho$ and $\sigma$ via a hypothesis test.
Suppose that we receive either $\rho$ or $\sigma$.
We must guess which state we obtained, based on the outcome of one measurement.
We may choose the measurement, specified by a two-outcome positive-operator-valued measure (POVM) $`{ Q, \Ident-Q }$.
The outcome $Q$ implies that we should guess ``$\rho$''; and $\Ident-Q$, that we should guess ``$\sigma$.''
Suppose that, when $\rho$ is provided, the measurement must identify $\rho$ correctly with a probability $\geq \eta$: $\tr`(Q\rho) \geq \eta$.
When $\sigma$ is provided, the optimal probability $\tr`(Q\sigma)$ of failing to identify $\sigma$ is $\eta\ee^{-\DHyp[\eta]{\rho}{\sigma}}$.
The interpretation of~\eqref{eq:setting-defn-DHyp} as a relative-entropy measure arises from certain properties.
For example,~\eqref{eq:setting-defn-DHyp} is non-negative and obeys a data-processing inequality~\cite{Dupuis2013_DH};~\eqref{eq:setting-defn-DHyp} approximates the min- and max-relative entropies when $\eta\approx 0$ and $\eta\approx 1$~\cite{Dupuis2013_DH}, respectively; and~\eqref{eq:setting-defn-DHyp} quantifies the resource costs of communication and thermodynamic tasks~\cite{BookKhatri_communication}.  
\cref{appx:hypothesis-testing-entropy} contains a more detailed discussion of the hypothesis-testing relative entropy and its properties.

The hypothesis-testing entropy and relative entropy have operational significances in quantum thermodynamics~\cite{YungerHalpern2016PRE_beyond,Buscemi2017PRA_relative,YungerHalpern2018JPA_beyond2}.
These quantities unify quantum thermodynamic results based on the min- and max- (relative) entropies~\cite{Dahlsten2011NJP_inadequacy,%
  Aberg2013_worklike,Horodecki2013_ThermoMaj,%
  Faist2015NatComm,Faist2018PRX_workcost}:
Consider resetting a state $\rho$ on a memory register $S$ to a fixed state $\ket{0}_S$.
Suppose that $\rho$ evolves under a completely degenerate Hamiltonian $H_S = 0$ and that the resetting must fail with a probability at most $\epsilon \geq 0$.
Absent restrictions on the resetting's complexity, the resetting has a work cost~\cite{Horodecki2013_ThermoMaj,Dupuis2013_DH}
\begin{align}
  W^{\epsilon}`*(\rho_S \to \ket{0}_S) \approx \HHyp[1-\epsilon]{\rho} \times k_{\mathrm B} T \ .
  \label{eq:WorkCostResetHypoTestingEntropy}
\end{align}
The approximation conceals technical details about how the failure probability is measured.
If $\rho_S = \Ident_2/2$ is the single-qubit maximally mixed state, then $\HHyp[1-\epsilon]{\rho} = \log(2)$, and we obtain the well-known formula for Landauer erasure~\cite{Landauer1961_5392446Erasure}
\begin{align}
  W^{\epsilon}`*(\Ident_2/2 \to \ket{0}_S) \approx k_{\mathrm B} T\log(2) \ .
\end{align}

\subsection{Complexity (relative) entropy}
\label{sec_Comp_ent}

The complexity entropy was introduced in the context of the resource theory of uncomplexity~\cite{YungerHalpern2022PRA_resource}.
There, a state's complexity entropy quantifies the qubits in the state $\ket{0}$ extractable via a limited number of unitary gates.
In this work, we apply the complexity entropy to quantum-information thermodynamics and detail the entropy's properties.
We present a version of the complexity entropy that is tailored to $n$-qubit quantum circuits wherein each gate is in the set $\mathcal{T}$ of elementary computations.
We introduce the complexity entropy's general form in \cref{appx:GeneralConstructionComplexityEntropy}, where we prove general bounds and monotonicity results.

A key motivation for defining the complexity entropy is the following.
Consider a hypothesis test between states $\rho$ and $\sigma$.
The hypothesis-testing relative entropy quantifies the optimal probability of wrongly rejecting $\sigma$ with any strategy that correctly accepts $\rho$ with a probability $\geq \eta$.
However, an optimal measurement $`{Q, \Ident-Q}$ might be too complex to be executed in a reasonable time.
Suppose, for instance, that $\rho = \proj{\psi}$ is a highly complex, pure $n$-qubit state and that $\sigma = \Ident_2^{\otimes n}/2^n$ is maximally mixed.
A measurement that distinguishes $\rho$ from $\sigma$ may require a complex circuit implementable only in an exponentially long time~\cite{Gross2009PRL_most,Bremner2009PRL_toorandom}.
It is natural to restrict $Q$ to be implementable in a reasonable time, with a circuit composed of $\leq r$ gates.
The \emph{complexity relative entropy} is defined similarly to the hypothesis-testing relative entropy. The former, however, has a complexity restriction in the optimization over measurement operators.

To define the complexity entropy, we must specify the set of POVM effects that a computationally limited observer can render.
Define $\Mr[r]$ as the set of POVM effects that one can implement by performing $\leq r$ gates and then applying certain single-qubit projectors.
All tensor products of those projectors constitute the set of complexity-zero POVM effects
\begin{align}
  \Mr[0] \coloneqq `*{ \bigotimes_{i=1}^n Q_i \,:\; Q_i \in `*{ \proj0, \Ident_2 } } \ .
  \label{eq:setting-defn-Mrzero}
\end{align}
For each $i$, an effect in $\Mr[0]$ projects qubit $i$ onto $\ket0$ or does nothing.
Fix any set $\mathcal{G}$ of elementary quantum operations (completely positive, trace-nonincreasing maps).
(We define complexity-restricted POVM elements in terms of $\mathcal{G}$, as an information-theoretic quantity independent of the thermodynamic framework in \cref{sec:setup-thermo-with-primitive-operations}.
Hence $\mathcal{G}$ is independent of the operations introduced there. 
If applying the POVM elements in that framework, however, one can choose for $\mathcal{G}$ to consist of operations introduced in \cref{sec:setup-thermo-with-primitive-operations}.)
For $r>0$, we define
\begin{align}
  \Mr[r] \coloneqq `*{
  \mathcal{E}_1^\dagger \cdots \mathcal{E}_r^\dagger ( P )
  \,:\; P\in\Mr[0]\,,\; \mathcal{E}_i\in\mathcal{G}
  }\ .
  \label{eq:setting-defn-Mr}
\end{align}

One natural choice for $\mathcal{G}$ might be the set of two-qubit unitary gates.
Another choice, in the context of complexity and thermodynamics, is $\mathcal{G} = \mathcal{T}$, wherein $\mathcal{T}$ denotes the set of elementary computations defined in \cref{sec:setup-thermo-with-primitive-operations}.

Even if a $Q$ belongs to $\Mr[r]$, the complementary effect $\Ident-Q$ might not.
Indeed, the definition of $\Mr[r]$ applies in settings where only one POVM effect is relevant in a hypothetical measurement.
For instance, suppose we wish to certify that some process outputs a state close to some pure state $\ket\psi$.
One may consider a hypothetical measurement of the output state with a POVM containing the effect $\proj\psi$, and ascertain that its outcome probability is close to unity.
In such a scenario, the POVM effect need not be implemented in practice, and other effects that complete the POVM may be ignored.

Having introduced complexity-restricted POVM effects, we can define complexity-restricted entropic quantities.
Let $\rho$ denote any quantum state; and $\Gamma$, any positive-semidefinite operator.
The \emph{complexity relative entropy} of $\rho$ relative to $\Gamma$, at a complexity scale $r\geq 0$ with respect to $\Mr[r]$, and for $\eta \in (0,1]$, is
\begin{align}
  \DHypr[r][\eta]{\rho}{\Gamma}
  \coloneqq -\log `*( \inf_{\substack{ Q \in \Mr[r] \\ \tr(Q\rho) \geq \eta }}
  `*{ \frac{\tr(Q\Gamma)}{\tr(Q\rho)} } ) \ .
  \label{eq:setting-defn-DHypr}
\end{align}
This definition extends the hypothesis-testing relative entropy~\eqref{eq:setting-defn-DHyp} by restricting the optimization to POVM effects implementable with $\leq r$ gates.
(See \cref{appx-topic:defn-complexity-relative-entropy} for a more general definition and for details.)
The definition~\eqref{eq:setting-defn-DHypr} mirrors a construction, based on the trace norm, in Ref.~\cite{Brandao2021PRXQ_models}.

The complexity relative entropy enjoys an operational interpretation similar to that of the hypothesis-testing relative entropy~\eqref{eq:setting-defn-DHyp}.
Consider a hypothesis test between states $\rho$ and $\sigma$.
We identify $\rho$ as the \emph{null hypothesis} and $\sigma$ as the \emph{alternative hypothesis}. 
Suppose that one can measure only a POVM $`*{Q, \Ident - Q}$ for which $Q\in \Mr[r]$.
It is useful to allow, apart from a measurement, one toss of a biased classical coin.
(We explain why below.)
One can freely choose the coin's probability $q \in (0,1]$ of landing heads up.
Consider guessing ``$\rho$'' if the coin shows ``heads'' and the POVM outcome is $Q$, guessing ``$\sigma$'' otherwise.
This strategy is equivalent to measuring the POVM $`*{ qQ , \Ident - qQ }$.
A type~I error---incorrectly rejecting the null hypothesis---occurs with a probability $1-\tr`(qQ\rho)$.
A type~II error---incorrectly rejecting the alternative hypothesis---occurs with a probability $\tr`(qQ\sigma)$. 
The following proposition shows how the complexity relative entropy characterizes hypothesis testing with complexity limitations:
\begin{mainproposition}[Hypothesis testing with complexity limitations]
    \label{mainthm:DHypr-interpretation-hypo-test}
    Let $\rho$ and $\sigma$ denote quantum states.
    Let $\eta \in (0,1]$ and $\delta \in (0,1]$.
    There exists a $Q\in \Mr[r]$ and $q \in (0,1]$ such that $\tr`*[(qQ)\rho] = \eta$ and $\tr`*[(qQ)\sigma] \leq \delta$ if and only if
    \begin{align}
        \DHypr[r][\eta]{\rho}{\sigma} \geq -\log`*(\delta/\eta) \ .
        \label{eq:DHypr-interpretation-hypo-test}
    \end{align}
\end{mainproposition}
The proposition guarantees that, if~\eqref{eq:DHypr-interpretation-hypo-test} holds, then, for any parameters $\eta, \delta \in(0,1]$, some hypothesis test of the sort just described has two properties:
a type~I error occurs with a probability $\leq 1-\eta$, and a type~II error occurs with a probability $\leq \delta$.
One can satisfy~\eqref{eq:DHypr-interpretation-hypo-test} only if $\delta \leq \eta$, since the complexity relative entropy is non-negative.
The coin toss's inclusion in the hypothesis test enables a precise relationship between the complexity relative entropy and hypothesis testing, for $\eta$ values far from 1.
If $\eta\approx 1$, then $q\approx 1$, and the agent need not toss the coin.
We prove the proposition in \cref{appx-topic:Dhypr-hypo-test}.

Having introduced the complexity relative entropy, we now define the complexity entropy.
Let $\rho$ denote any quantum state.
The \emph{complexity entropy} of $\rho$, at a complexity scale $r \geq 0$ with respect to $\Mr[r]$, and for $\eta \in (0,1]$, is
\begin{align}
  \HHypr[r][\eta]{\rho} \coloneqq - \DHypr[r][\eta]{\rho}{\Ident} \ .
  \label{eq:setting-defn-HHypr}
\end{align}
This definition mirrors~\eqref{eq:setting-defn-HHyp}.
We follow a standard procedure for defining an entropy using a relative entropy~\cite{BookTomamichel2016_Finite}.
We show that the complexity entropy is always non-negative $[\HHypr[r][\eta]{\rho}\geq0]$;
see~\cref{appx-topic:Dhypr-basic-properties} for a proof.

The normalization factor $\tr(Q\rho)$ in~\eqref{eq:setting-defn-DHypr} guarantees elementary properties of the complexity entropy, such as its having the range $[0, n\log(2)]$ for all $\eta$.
In \cref{appx:GeneralConstructionComplexityEntropy}, we define an alternative version of the complexity (relative) entropy without the normalization factor.
This alternative is technically more convenient, and we bound one version in terms of the other.

\section{Thermodynamic information erasure with complexity constraints}
\label{sec:main-erasure}

We now turn to the erasure of the $n$-qubit memory register introduced in \cref{sec:setting}.
What is the optimal work cost of resetting a state $\rho$ to $\ket{0^n}$, using the primitive operations \labelcref{item:primitive-reset,item:primitive-extract,item:primitive-operation-computation}?

\subsection{Erasure of qubits with a completely degenerate Hamiltonian}
\label{sec:main-erasure-Wcost-degenerate-H}

We first consider a simpler case: each memory qubit has a degenerate Hamiltonian ($H_i = 0$ for all $i$).
Let the elementary computations $\mathcal{T}$ be the two-qubit unitary gates with arbitrary connectivities: $\mathcal{T} \simeq \bigcup_{(i,j)} \SU(4)_{i,j}$.
The mixed-state fidelity between states $\sigma$ and $\tau$ is defined as~\cite{BookNielsenChuang2000}
\begin{equation}
    F`*(\sigma,\tau) \coloneqq \tr`*(\sqrt{\sigma^{1/2}\tau\sigma^{1/2}}) \ .
\end{equation}
We define an erasure of $\rho$ as any composition of the operations \labelcref{item:primitive-reset,item:primitive-extract,item:primitive-operation-computation} that transforms $\rho$ into a state $\rho'$ satisfying $F^2`*(\rho', \proj{0^n}) = \dmatrixel{0^n}{ \rho' } \geq \eta$, for some fixed $\eta \in (0,1]$ [\cref{fig:ThermodynamicErasureComplexityConstraints}(a)].
\begin{figure}
  \centering
  \includegraphics{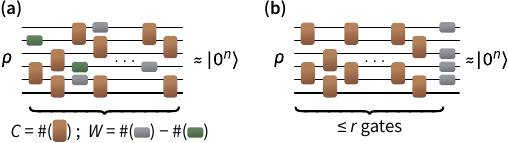}
  \caption{
    Protocols for thermodynamic information erasure with complexity constraints.
    The process consists of \textsc{reset} gates (gray circuit elements), \textsc{extract} gates (green), and computational gates (brown) [cf.\@ primitives~\labelcref{item:primitive-reset,%
      item:primitive-extract,item:primitive-operation-computation} of \cref{fig:SimpleModelInformationThermo}].
    \textbf{(a)}~A general protocol interleaves the primitive operations \labelcref{item:primitive-reset,item:primitive-extract,item:primitive-operation-computation}, preparing a state close to $\ket{0^n}$.
    The protocol's complexity cost $C$ and work cost $W$ follow from summing the primitive operations' costs (\cref{fig:SimpleModelInformationThermo}).
    \textbf{(b)}~In the simplified protocol, one applies $\leq r$ computational gates [primitive~\ref{item:primitive-operation-computation}].
    Then, one erases any qubits far from $\ket0$, using \textsc{reset} operations [primitive~\ref{item:primitive-reset}].
    The circuit compresses $\rho$ into fewer qubits, to be erased by \textsc{reset} operations.
    Such protocols extract pure $\ket{0}$ qubits in the resource theory of uncomplexity~\protect\cite{YungerHalpern2022PRA_resource}.}
  \label{fig:ThermodynamicErasureComplexityConstraints}
\end{figure}
One can interpret $F^2`*(\rho', \proj{0^n})$ as the probability of preparing $\ket{0^n}$.

We now quantify the trade-off between the work cost and complexity cost of an erasure protocol $\mathcal{E}$.
Being a process, $\mathcal{E}$ is a sequence of primitive operations.
Yet, we denote by $\mathcal{E}$ also the operation implemented by the sequence, in a slight abuse of notation.
Suppose that $\mathcal{E}$ has a complexity at most $r\geq 0$: $C(\mathcal{E}) \leq r$.
The optimal work cost for such a protocol is
\begin{align}
  W_r \coloneqq \min_{\mathcal{E}} `*{ W(\mathcal{E}) \,:\ 
  C(\mathcal{E})\leq r\,,\; F^2 \bm{(} \mathcal{E}(\rho),\proj{0^n} \bm{)} \geq\eta } \ .
  \label{eq:main-defn-optimal-work-cost-erasure-fixed-complexity}
\end{align}

\subsubsection{Protocols whose \textsc{reset} operations happen at the end}
\label{sec:main-erasure-Wcost-degenerate-H-simpleprotocol}

As an initial step, we consider an erasure protocol divided into two parts [\cref{fig:ThermodynamicErasureComplexityConstraints}(b)].
First, the state undergoes $\leq r$ computational gates [primitive~\ref{item:primitive-operation-computation}].
Then, \textsc{reset} operations [primitive~\ref{item:primitive-reset}] are applied.
Such a protocol's optimal work cost is
\begin{align}
    W_r^* = 
    & \min_{\mathcal{E}}  \bigl\{
    W(\mathcal{E}) \,:\ 
    \mathcal{E} = `*(\textstyle\prod_{i\in\mathcal{W}} \mathcal{E}_{\textsc{RESET},i}) \mathcal{E}_r \cdots \mathcal{E}_1 \,,\;
    \nonumber\\
    &\hspace{0.4cm} \mathcal{W} \subset `{1, 2, \ldots, n} \,,\; 
    F^2 \bm{(} \mathcal{E}(\rho),\proj{0^n} \bm{)} \geq\eta
    \bigr\} \ .
  \label{eq:main-defn-optimal-work-cost-erasure-fixed-complexity-simple-protocol}
\end{align}
Here $\mathcal{W}$ denotes the subset of the qubits that undergo \textsc{reset} operations.

We now identify how to erase $\rho$ using the least amount of work. We should compress $\rho$ into as few qubits as possible, using $\leq r$ gates, to apply as few \textsc{reset} operations as possible.
In other words, we should perform data compression with $\leq r$ gates.
This information-theoretic interpretation of erasure mirrors two earlier interpretations of erasures: the erasure of $n$ qubits without complexity restrictions~\cite{Dahlsten2011NJP_inadequacy} and the erasure of a quantum system using a quantum memory~\cite{delRio2011Nature}.

We analyzed a version of this compression task in the context of extracting pure qubits in the resource theory of uncomplexity~\cite{YungerHalpern2022PRA_resource}.
A variant of the complexity entropy, we showed, quantifies the qubits that cannot be reset to $\ket{0}$ states with $\leq r$ gates.
In this subsection, we adapt that argument to thermodynamic erasure.
(We discuss further applications of the complexity entropy to information-theoretic tasks in \cref{sec:InformationTheoreticApplicationsComplexityEntropy}.)

Let $\mathcal{E}$ denote any protocol that achieves the minimum in~\eqref{eq:main-defn-optimal-work-cost-erasure-fixed-complexity-simple-protocol}: $W(\mathcal{E}) = W_r^*$.
Assume that the work cost $W(\mathcal{E})$ corresponds to $\beta W(\mathcal{E})/\log(2) = \abs{\mathcal{W}} \eqqcolon w$ bits.
Define the projector $P \coloneqq \Ident_{\mathcal{W}} \otimes \proj{0^{n-w}}_{\mathcal{W}^{\rm c} }$, wherein $\mathcal{W}^{\rm c}$ denotes the complement of $\mathcal{W}$.
Define the POVM effect
\begin{align}
  Q \coloneqq \mathcal{E}_1^\dagger\cdots\mathcal{E}_r^\dagger`(P) \ ,
  \label{eq:main-erasure-degenerateH-candidate-operator-Q}
\end{align}
wherein each $\mathcal{E}_i^\dagger$ denotes the adjoint of the operation $\mathcal{E}_i$ $\big[$defined via $\tr \bm{(} \mathcal{E}_i^\dagger(A)\,B \bm{)} = \tr \bm{(} A\,\mathcal{E}_i(B) \bm{)}$ for all operators $A$ and $B$$\big]$.
$P$ projects onto $\ket0$ each qubit not subject to any \textsc{reset} operation, leaving all other qubits untouched.
We choose for the set $\mathcal{G}$ of elementary operations, used to define $\Mr[r]$ in~\eqref{eq:setting-defn-Mr}, to be the set $\mathcal{T}$ of elementary computations, used to define our model in \cref{sec:setup-thermo-with-primitive-operations}.
$Q$ belongs to $\Mr[r]$.
Furthermore, 
\begin{align}
    \tr`(Q\rho)
    &= \tr \left( P\, \mathcal{E}_r\cdots\mathcal{E}_1(\rho) \right)
    \nonumber\\
    &= \tr`*{\proj{0^{n-w}}\, \tr_{\mathcal{W}} \bm{(} \mathcal{E}_r\cdots\mathcal{E}_1(\rho) \bm{)} }
    \nonumber\\
    &= \tr \bm{(} \proj{0^n} \mathcal{E}(\rho) \bm{)}
    = F^2 \bm{(} \mathcal{E}(\rho),\proj{0^n} \bm{)}
    \geq \eta \ ,
\end{align}
so $Q$ is a candidate for the optimization~\eqref{eq:setting-defn-DHypr} defining the complexity entropy $\HHypr[r][\eta]{\rho} = -\DHypr[r][\eta]{\rho}{\Ident}$.
Moreover, since each $\mathcal{E}_i$ is unitary, $\tr`(Q) = \tr \bm{(} P\,\mathcal{E}_r\cdots\mathcal{E}_1(\Ident) \bm{)} = \tr`(P) = 2^w$.
Hence,
\begin{align}
  &\HHypr[r][\eta]{\rho} - \log`*( 1/\eta )  \leq \log \bm{(} \tr(Q) /\tr(Q\rho) \bm{)} - \log`*( 1/\eta )
  \nonumber\\
  &\leq \log \bm{(} \tr`(Q) \bm{)}
  = w\log(2)
  = \beta W(\mathcal{E})
  = \beta W_r^* \ .
\end{align}

A similar inequality points in the opposite direction, as we show by reversing the steps above.
Consider any candidate $Q\in\Mr[r]$ for the optimization~\eqref{eq:setting-defn-DHypr} defining $\HHypr[r][\eta]{\rho} = -\DHypr[r][\eta]{\rho}{\Ident}$.
Let $Q$ be optimal (or arbitrarily close to optimal).
By the definition~\eqref{eq:setting-defn-Mr}, $Q = \mathcal{E}_1^\dagger\cdots\mathcal{E}_r^\dagger`(P)$, wherein each $\mathcal{E}_i \in \mathcal{G}$ and wherein $P$ projects onto $\ket0$ each qubit in some subset $\mathcal{W}^{\rm c}$, leaving all other qubits untouched.
Let $\mathcal{E} \coloneqq `*(\textstyle\prod_{i\in\mathcal{W}}\mathcal{E}_{\textsc{RESET},i}) \mathcal{E}_r\cdots\mathcal{E}_1$.
$\mathcal{E}$ erases $\rho$, since
\begin{align}
  F^2 \bm{(} \mathcal{E}(\rho), \proj{0^n} \bm{)}
  = \tr \bm{(} \proj{0^n} \mathcal{E}(\rho) \bm{)}
  = \tr`*( Q \rho )
  \geq \eta \ .
\end{align}
Hence,
\begin{align}
  &\beta W_r^*
  \leq \beta W(\mathcal{E})
  = \abs{\mathcal{W}}\,\log(2)
  \nonumber\\
  &= \log \bm{(} \tr(Q) \bm{)}
  \leq \log \bm{(} \tr(Q) /\tr(Q\rho) \bm{)}
  = \HHypr[r][\eta]{\rho} \ .
\end{align}
The last equality holds if $Q$ is optimal.
(If $Q$ is arbitrarily close to optimal, then the last two quantities are arbitrarily close to each other.)
We have therefore proved the following theorem:
\begin{maintheorem}[Complexity-limited erasure with restricted protocols]
  \noproofref
  \label{mainthm:erasure-work-cost-complexity-entropy-simple}
  Consider erasing an $n$-qubit system governed by a fully degenerate Hamiltonian.
  Every optimal complexity-limited protocol that uses only primitive~\labelcref{item:primitive-operation-computation} (wherein $\mathcal{T}=\mathcal{G}$ consists of unitary operations) followed by primitive~\labelcref{item:primitive-reset}, expends an amount $W_r^*$ of work that obeys
  \begin{align}
    \HHypr[r][\eta]{\rho} - \log`*( 1/\eta ) \leq \beta W_r^* \leq \HHypr[r][\eta]{\rho} \ .
    \label{eq:main-optimal-Wrstar}
  \end{align}
\end{maintheorem}

Two points merit mentioning.
First, one might wonder whether the work cost $\beta^{-1} `*[\HHypr[r][\eta]{\rho} - \log(1/\eta) ]$ is achievable.
The answer depends on whether one can implement, with gates in $\mathcal{G}$, a POVM effect $Q\in \Mr[r]$ that satisfies $\tr(Q\rho) = \eta$ and is arbitrarily close to optimal for $\HHypr[r][\eta]{\rho}$.
In such a case, the first inequality in~\eqref{eq:main-optimal-Wrstar} saturates: $\HHypr[r][\eta]{\rho} - \log(1/\eta)$ equals a variant of the complexity entropy, and that variant, we show, equals $\beta W_r^*$.
(See \cref{appx-topic:reduced-complexity-entropy,%
appx-topic:erasure-exact-expression-reduced-Hhypr} for details.)
Considering the variant's definition, one can interpret the $\log(1/\eta)$ in~\eqref{eq:main-optimal-Wrstar} as a particularity of how we defined the complexity entropy.

Second, the $\log(1/\eta)$ is proportional to the work won in a successful bet on an event that occurs with a probability $\eta$~\cite{Alhambra2016PRX_probability}.
We can understand this point through a simple example.
Consider resetting the single-qubit maximally mixed state $\rho = \Ident_2/2$ to $\ket0$.
Suppose that our target success probability $\eta$ satisfies $\eta \leq 1/2$.
One successful erasure protocol does nothing: with the probability $1/2$, $\ket0$ is prepared.
Yet standard measures of entropy (including the R\'enyi entropies of one-shot information theory~\cite{PhDRenner2005_SQKD}) attribute to $\rho$ one bit of entropy.
The zero work cost of erasing a maximally mixed qubit with $\eta = 1/2$ can be understood as a sum of (i)~one unit of work expended to reduce the entropy of $\rho$ and (ii)~one unit of work extracted by a successful bet with the success probability $1/2$.

We now present a few textbook quantum states and the work costs of erasing them, in various regimes of $r$ and $\eta$.

\paragraph{Computational-basis states.} 
If $\rho = \proj{0^n}$, then $\HHypr[r][\eta]{\rho} = 0$ for all $r$ and $\eta$.
An optimal protocol is to do nothing---apply no \textsc{reset} operations and no work: $W_r^* = 0$.

If $\rho = \proj{1^n}$, then one resets $\rho$ to $\ket{0^n}$ by flipping each qubit to $\ket{0}$, incurring a minor complexity cost (minor compared to the exponential complexities expected of most $n$-qubit unitaries~\cite{Susskind2018arXiv_ThreeLectures}).
Assume that $n$ is even, for simplicity.
One can flip all $n$ qubits with one layer of $n/2$ two-qubit gates.
Therefore, whenever $r\geq n/2$, $\HHypr[r][\eta]{\rho} = 0$, and $W_r^* = 0$, for all $\eta$.
If, however, $r < n/2$, then one can flip only $2r < n$ qubits.
One must apply \textsc{reset} operations to the other $n - 2r$ qubits.
In this case, $\HHypr[r][\eta]{\rho} = (n - 2r) \log(2)$, and $W_r^* = `(n - 2r) \, k_{\mathrm B}T \log(2)$, for all $\eta$.

\paragraph{The maximally mixed state.}
If $\rho$ is maximally mixed, then $\HHypr[r][\eta]{\rho} = n \log(2)$ for all $r$ and $\eta$.
Suppose that the error tolerance is insignificant: $\eta \approx 1$.
An optimal protocol requires $n$ \textsc{reset} operations, costing the maximum amount of work: $W_r^* = n \, k_{\mathrm B} T \log(2)$.

\paragraph{Greenberger-Horne-Zeilinger (GHZ) state.}
Let
\begin{align}
    \ket{\mathrm{GHZ}} &\coloneqq \frac1{\sqrt 2}`*( \ket{0^n} + \ket{1^n} ) \ ,
\end{align}
and let $\rho = \proj{\mathrm{GHZ}}$.
One can prepare the GHZ state with a Hadamard gate followed by $n-1$ \textsc{CNOT} gates.  
Therefore, $\HHypr[r][\eta]{\rho} = 0$ if $r \geq n$.
If $r < n$, then one can apply $r$ \textsc{CNOT} gates to disentangle $r$ qubits from the other $n-r$ qubits.
The $n-r$ qubits require \textsc{reset} operations.
Absent any significant error tolerance (if $\eta \approx 1$), $\HHypr[r][\eta]{\rho} = (n-r) \log(2)$, and $W_r^* = `(n-r)\,k_{\mathrm B} T \log(2)$.

\paragraph{A Haar-random state.}
Let $\ket\psi$ denote a state chosen randomly according to the Haar measure on the pure $n$-qubit states.
The state $\rho = \proj{\psi}$ is indistinguishable from the maximally mixed state, according to $\poly(n)$-complex observables~\cite{Gross2009PRL_most,Bremner2009PRL_toorandom}.
Therefore, if $r \lesssim \poly(n)$, one cannot distinguish $\rho$ from the maximally mixed state using $\leq r$ gates.
All $n$ qubits must undergo \textsc{reset} operations.
Absent any significant error tolerance (if $\eta \approx 1$), $\HHypr[r][\eta]{\rho} = n \log(2)$, and $W_r^* = n\,k_{\mathrm B}T \log(2)$.

This example constitutes a special case of \cref{mainthm:ComplexityEntropyInRandomCircuits} in \cref{sec-topic:main-random-circuits}.
There, we prove a lower bound on the complexity entropy of a state generated by a random circuit.
Random circuits effect Haar-random unitaries in the large-circuit-depth limit.

\paragraph{Mixture of different-complexity states.}
Let $\rho$ denote a convex mixture of $\proj{0^n}$ and a high-complexity state $\proj{\psi}$: $\rho = (1-\epsilon) \proj{0^n} + \epsilon \proj{\psi}$, wherein $\epsilon \in [0, 1-\eta]$.
The POVM effect $Q = \proj{0^n}$ is a candidate for the optimization~\eqref{eq:setting-defn-DHypr} defining $\HHypr[r][\eta]{\rho} = -\DHypr[r][\eta]{\rho}{\Ident}$. Therefore, $\HHypr[r][\eta]{\rho} \leq -\log(1-\epsilon)$, and $W_r^* \leq -k_{\mathrm B} T \log(1-\epsilon) \approx \epsilon \, k_{\mathrm B} T$ (the approximation holds for small $\epsilon$).

\subsubsection{General protocols with midcircuit \textsc{reset} and \textsc{extract} operations}

Suppose that each qubit's Hamiltonian is degenerate.
Suppose further that a process can be any sequence of the primitives~\labelcref{item:primitive-reset,%
  item:primitive-extract,item:primitive-operation-computation} [\cref{fig:ThermodynamicErasureComplexityConstraints}(a)].
A process may interleave elementary computations with \textsc{reset} and \textsc{extract} operations on the same qubit.
These midcircuit nonunitary operations can reduce the complexity of information erasure.
The significance of midcircuit measurements in \emph{monitored circuits} has only recently been appreciated, both in quantum information and in condensed-matter theory. 
Experimentalists have recently used midcircuit measurements in quantum error correction and in \emph{monitored circuits}~\cite{Graham2023_midcircuit,Koh2023NP_measurementinduced,Zhu2023NP_interactive}.
Along similar lines, quantum phases of matter driven by measurements have been explored~\cite{PhysRevX.9.031009, PhysRevB.99.224307, PhysRevLett.131.200201,Yoshida2021arXiv_decoding, Suzuki2025_complexity}.
We seek to bound $W_r$, defined in~\eqref{eq:main-defn-optimal-work-cost-erasure-fixed-complexity}.
Our strategy is to map a general erasure protocol to a different protocol involving additional auxiliary systems and whose reset operations all happen at the end.
The work cost of the general protocol is then bounded using an argument adapted from the one in \cref{sec:main-erasure-Wcost-degenerate-H-simpleprotocol}.

We must transform a protocol $\mathcal{E}$ into a POVM effect.
Our strategy involves ancillary qubits, each initialized to $\ket0$.
Every midcircuit \textsc{reset} operation on qubit $i$ is performable with a final \textsc{reset} operation: we swap qubit $i$ with an ancilla, then \textsc{reset} the ancilla.
Similarly, every midcircuit \textsc{extract} operation on qubit $i$ is performable with an initial \textsc{extract} operation: we perform an \textsc{extract} operation on an ancilla, then swap the ancilla with qubit $i$.
In both processes, every ancilla begins and ends in $\ket0$ (recall that the \textsc{extract} operation is performable only on a qubit in $\ket0$).
We assume that the \textsc{swap} gate belongs to $\mathcal{T}$.
\begin{center}
  \includegraphics{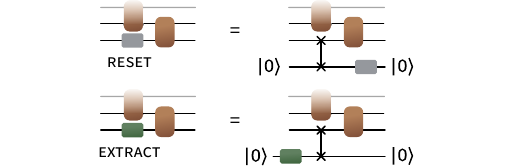}
\end{center}
Consider any candidate protocol $\mathcal{E}$ for the optimization~\eqref{eq:main-defn-optimal-work-cost-erasure-fixed-complexity}.
Suppose that $\mathcal{E}$ consists of $m_1$ \textsc{reset} operations, $m_2$ \textsc{extract} operations, and $\leq r$ gates.
$\mathcal{E}$ has the complexity cost $C(\mathcal{E}) \leq r$ and the work cost $W(\mathcal{E}) = (m_1 - m_2) \, k_{\mathrm B} T \log(2)$.
Let us transfer all the \textsc{reset} and \textsc{extract} operations in $\mathcal{E}$ to ancillas.
We obtain a protocol $\mathcal{E}'$ of $\leq r + m_1 + m_2$ gates (each \textsc{swap} contributes one gate), preceded by \textsc{extract} operations and followed by \textsc{reset} operations.
$\mathcal{E}'$ acts on $n + m_1 + m_2$ qubits and has the complexity cost $C(\mathcal{E}') \leq r + m_1 + m_2$.
Immediately after the initial \textsc{extract} operations, the input state is
\begin{align}
  \tilde{\rho}_{m_1,m_2} \coloneqq \rho \otimes \proj{0^{m_1}} \otimes \Ident_2^{\otimes m_2}/2^{m_2} \ .
\end{align}
Subsequently, $\mathcal{E}'$ outputs a state $\mathcal{E}'(\tilde{\rho}) \eqqcolon \tilde{\rho}'$ satisfying
\begin{align}
    F^2 \bm{(} \tr_{m_1}`( \tilde{\rho}' ), \proj{0^{n+m_2}} \bm{)} \geq \eta \ .
\end{align}

We now adapt the argument used for protocols whose \textsc{reset} operations happen at the end.
Suppose the set $\mathcal{G}$ of elementary operations, used to define $\Mr[r]$ in~\eqref{eq:setting-defn-Mr}, equals the set of elementary computations: $\mathcal{G} = \mathcal{T}$.
Define the POVM effect
\begin{align}
  Q \coloneqq \mathcal{E}'^\dagger `*( \proj{0^n} \otimes \Ident_2^{\otimes m_1} \otimes \proj{0^{m_2}} ) \ .
\end{align}
Then 
\begin{align}
  \tr`*(Q\tilde{\rho})
  &= \tr \bm{(} \proj{0^{n+m_2}} \, \tr_{m_1}`(\tilde\rho') \bm{)}
  \geq \eta \ .
\end{align}
Hence, $Q$ is a candidate for the optimization~\eqref{eq:setting-defn-DHypr} defining $\HHypr[r+m_1+m_2][\eta]{\tilde{\rho}} = - \DHypr[r+m_1+m_2][\eta]{\tilde{\rho}}{\Ident}$.
Therefore,
\begin{align}
  \HHypr[r+m_1+m_2][\eta]{\tilde{\rho}}
  &\leq \log \bm{(} \tr(Q) /\tr(Q\tilde{\rho}) \bm{)}
  \nonumber\\
  &\leq m_1\log(2) + \log `*( 1/\eta ) \ ,
\end{align}
so
\begin{align}
  \beta W`*(\mathcal{E}) \geq \HHypr[r+m_1+m_2][\eta]{\tilde{\rho}} - m_2\log(2) - \log`*( 1/\eta ) \ .
  \label{eq:work-cost-midcircuit-operations}
\end{align}
We use the shorthand notation
\begin{align}
  g^{r,\eta}(\rho) \coloneqq  \inf_{m_1,m_2\geq 0} `*{
  \HHypr[r+m_1+m_2][\eta]{\tilde{\rho}_{m_1,m_2}} - m_2\log(2) } 
\end{align}
and take the infimum of~\eqref{eq:work-cost-midcircuit-operations} over all protocols $\mathcal{E}$.
Combining the previous two equations yields
\begin{align}
  \beta W_r \geq g^{r,\eta}(\rho) - \log`*(1/\eta) \ .
\end{align}

The reverse direction might not hold, generally.
For any $m_1,m_2\geq 0$,~\eqref{eq:main-optimal-Wrstar} guarantees the existence of a protocol $\mathcal{E}$ that has two properties.
First, $\mathcal{E}$ consists of $\leq r+m_1+m_2$ computational gates followed by $w$ \textsc{reset} operations.
Second, $\mathcal{E}$ maps $\tilde{\rho}$ approximately to $\ket{0^{n+m_1+m_2}}$, at the work cost $w\,\log(2) \leq \HHypr[r+m_1+m_2][\eta]{\tilde{\rho}}$.
Yet, such a protocol may involve the $m_1 + m_2$ ancillas in a computation inequivalent to any computation on $n$ qubits, even if the latter computation includes midcircuit \textsc{reset} and \textsc{extract} operations.

Overall, we have bounded the optimal work cost $W_r$ of erasing $\rho$ using $\leq r$ computational gates.
\begin{maintheorem}[Bound of optimal work cost]
  \noproofref
  $W_r$ obeys
  \begin{align}
    g^{r,\eta}(\rho) - \log`*( 1/\eta ) \leq \beta W_r \leq \beta W_r^* \leq \HHypr[r][\eta]{\rho} \ .
  \end{align}
\end{maintheorem}

A large gap may separate $W_r$ and $W_r^*$ if, with few pure ancillas, one can uncompute a pure state $\ket\chi$ to $\ket{0^n}$, using a circuit much shorter than is possible without ancillas.
In other words, an erasure protocol may benefit from early \textsc{reset} operations: the resulting pure qubits may be used as ancillas for uncomputing the remaining state.
For instance, there might exist an $n$-qubit state $\ket\chi$, and an $\ell\ll n$, with the following two properties.
First, $\ket{\chi}\otimes\ket{0^\ell}$ is much less complex than 
$\ket\chi$: $C`*(\ket\chi\otimes\ket{0^\ell}) \ll C`*(\ket\chi)$.
Second, in the absence of ancillas, $r$ gates fail to extract any $\ket{0}$'s from $\ket\chi$, if $r = C`*(\ket\chi\otimes\ket{0^\ell})$.
Midcircuit \textsc{reset} operations would be able to lower the work cost of erasing $\proj\chi \otimes `*(\Ident_2^{\otimes \ell}/2^\ell)$ with at most $r=C`*(\ket\chi\otimes\ket{0^\ell})$ gates, we now show.
The following protocol would employ midcircuit \textsc{reset} operations and cost work $k_{\rm B} T \ell\log(2)$: \textsc{reset} the $\ell$ mixed qubits, paying $\ell$ units of work.
Then, uncompute the remaining state, $\ket\chi\otimes\ket{0^\ell}$, to $\ket{0^{n+\ell}}$, using $r$ gates.
The total work cost would be $\beta^{-1} \ell\log(2)$.
In contrast, consider a protocol whose \textsc{reset} operations happen at the end.
$r$ operations would not extract any $\ket{0}$'s. 
Such a protocol would require $n+\ell$ units of work: $\beta W_r^* = (n + \ell)\log(2) \gg \beta W_r$.

\subsection{Erasure of qubits with a general product Hamiltonian}
\label{sec:erasure-nontrivial-H}

We now consider a more general setup: each qubit has a not-necessarily-degenerate Hamiltonian $H_i$, as per \cref{sec:setting}.
The set $\mathcal{T}$ of elementary computations can no longer be the set of all two-qubit unitary gates: some gates would require work.

We require that $\mathcal{T}$ contain only completely positive, trace-preserving maps $\mathcal{E}$ that satisfy the 
technical property
\begin{align}
  \mathcal{E}(\Gamma)
  = \Gamma\ , \; \; \text{wherein} \; \;
    \Gamma = \ee^{-\beta \sum_i H_i} = \bigotimes_{i=1}^n \Gamma_i \ .
  \label{eq:main-erasure-nontrivialHi-condition-Gibbssubpres}
\end{align}
The operators $\Gamma_i$ are defined in~\eqref{eq:setting-gammai-Gammai-Zi}.
Examples of such operations include the two-qubit thermal operations and the two-qubit Gibbs-preserving maps (\cref{sec:setting}).
We assume~\eqref{eq:main-erasure-nontrivialHi-condition-Gibbssubpres} to ensure a thermodynamically consistent accounting of work costs: elementary computations should cost no work.

For concreteness, we describe another class of operations that satisfy condition~\eqref{eq:main-erasure-nontrivialHi-condition-Gibbssubpres}.
These operations can model computations on systems governed by product Hamiltonians, as well as crude control over interactions with a heat bath.
Consider an operation, on qubits $i$ and $j$, of the form
\begin{align}
\rho\mapsto
  \mathcal{E}(\rho)
  &= (1-q) U_{i,j} \, \rho \, U_{i,j}^\dagger + q \, \gamma_i \otimes \gamma_j \ .
    \label{eq:main-erasure-general-example-comp-operation}
\end{align}
$q$ denotes a probability.
$U_{i,j}$ denotes an energy-conserving unitary: $`*[U_{i,j}, H_i + H_j] = 0$.
$\gamma_i$ $`*(\gamma_j)$ denotes the thermal state of qubit $i$ $`*(j)$, defined in~\eqref{eq:setting-gammai-Gammai-Zi}.
With a probability $q$, $\mathcal{E}$ replaces its input with a thermal state. Otherwise, $\mathcal{E}$ effects an energy-conserving unitary.
$\mathcal{E}$ can model an interaction between qubits and a heat bath over short times, during which the bath partially thermalizes the qubits.
This interaction may be exploited to implement operations that do not strictly conserve energy, i.e., that do not commute with the total Hamiltonian.
Subjecting $\gamma_i \otimes \gamma_j$ to any operation~\eqref{eq:main-erasure-general-example-comp-operation} shows directly that the operation obeys the condition~\eqref{eq:main-erasure-nontrivialHi-condition-Gibbssubpres}.

Consider any protocol $\mathcal{E}$ in $\mathcal{T}$ that has two properties: first, $\mathcal{E}$ consists of $\leq r$ computational gates followed by \textsc{reset} operations.
Second, $\mathcal{E}$ satisfies $F^2 \bm{(} \mathcal{E}(\rho), \proj{0^n} \bm{)} \geq \eta$.
What is the optimal work cost $W_r^*$ of such a protocol?
The following theorem bounds $W_r^*$ in terms of the complexity relative entropy, generalizing~\eqref{eq:main-optimal-Wrstar} to product Hamiltonians.
\begin{maintheorem}[Optimal work cost for product Hamiltonians]
  \label{thm:not-necessarily-degenerate-Hamiltonians}
  The optimal work cost $W_r^*$ is quantified by the complexity relative entropy as
  \begin{align}
    - \DHypr[r][\eta][\mathcal{T}]{\rho}{\Gamma} -\log`*(1/\eta)
    \leq W_r^*
    \leq -\DHypr[r][\eta][\mathcal{T}]{\rho}{\Gamma} \ .
  \end{align}
  The subscript $\mathcal{T}$ signifies that the set $\Mr[r]$ in~\eqref{eq:setting-defn-Mr} is defined in terms of the gate set $\mathcal{G} = \mathcal{T}$.
\end{maintheorem}
\noindent The proof, provided in \cref{appx-topic:erasure-nontrivial-Ham}, resembles the proof for degenerate Hamiltonians.

\section{Work costs of general state transformations}
\label{sec:main-general-state-transformations}

We now turn to our paper's motivating problem.
Consider a system that begins in a state $\rho$ and then undergoes any evolution of complexity at most $r \geq 0$.
$\rho$ will transform into a possibly different state. 
Which transformations appear possible, according to an observer who can distinguish states using only some computational power $R \geq 0$?
Alternatively, how much work must one invest to effect an evolution whose complexity is $\leq r$ and whose output the observer cannot distinguish from a given state $\rho'$?

\paragraph{Setting.}

To formalize the questions above, we consider an $n$-qubit system initialized in a state $\rho$. The state undergoes an evolution $\mathcal{E}$ composed of primitive operations (\cref{sec:setting}).
$C(\mathcal{E})$ and $W(\mathcal{E})$ quantify the complexity cost and work cost of $\mathcal{E}$, respectively.
To establish whether a transformation $\rho \to \rho'$ appears to result from $\mathcal{E}$, imagine a referee who must distinguish the output $\mathcal{E}(\rho)$ from the reference $\rho'$.
The referee can effect only operations of complexities $\leq R$.
Using a measurement operator of complexity $\leq R$, the referee performs a hypothesis test of the form described in \cref{mainthm:DHypr-interpretation-hypo-test}, subject to the type~I and type~II error constraints specified with $\eta\in (0,1]$ and $\delta\in (0,1]$, respectively.
According to the referee, $\rho \to \rho'$ may result from $\mathcal{E}$ precisely if there exists no hypothesis test by which the referee can correctly accept $\mathcal{E}(\rho)$ with a probability $\geq \eta$ and incorrectly accept $\rho'$ with a probability $\leq \delta$.
Equivalently, by \cref{mainthm:DHypr-interpretation-hypo-test}, $\rho \to \rho'$ may result from $\mathcal{E}$ precisely if
\begin{align}
  \DHypr[R][\eta]`*{\mathcal{E}(\rho)}{\rho'} \leq -\log`*(\delta/\eta) \ .
  \label{eq:general-work-cost-thermo-process-condition-success}
\end{align}

Condition~\eqref{eq:general-work-cost-thermo-process-condition-success} can simplify, evoking the previous section's fidelity condition, if the referee lacks computational restrictions.
Depending on the POVMs performable, the referee could perform the optimal measurement for distinguishing between $\rho$ and $\rho'$.
The measurement's success probability equals the trace distance $\frac{1}{2} \onenorm{\rho - \rho'}$, which is small when the fidelity is large~\cite{BookNielsenChuang2010}.
Hence condition~\eqref{eq:general-work-cost-thermo-process-condition-success} effectively reduces to a constraint, as in the previous section, on $F$.

Let $\mathcal{E}$ denote any evolution, of a complexity $\leq r$, that transforms $\rho$ into any state $\mathcal{E}(\rho)$ indistinguishable from $\rho'$ by the referee.
The minimum work cost of any such evolution is
\begin{multline}
  \mathcal{W}^{r, \, R, \, \eta, \, \delta}`*[ \rho \to \rho' ] \coloneqq
  \\
 \inf`*{ W`*(\mathcal{E})\,:\;
   C`*(\mathcal{E}) \leq r \,,\;
   \DHypr[R][\eta]`*{\mathcal{E}(\rho)}{\rho'}
   \leq -\log`*(\delta/\eta) } \ .
 \label{eq:defn-general-work-cost-thermodynamic-process}
\end{multline}
This quantity, in its general form, does not appear amenable to simple analysis.
In the following, we specialize to situations where the referee's computational power is extremely high or low.
In these situations, we find examples where one of the conditions in~\eqref{eq:defn-general-work-cost-thermodynamic-process} is expected to be tractable, as is needed to upper-bound the work cost.

\paragraph{Referee with extremely high computational power.}
First, suppose that $R$ is extremely large.
The first condition in~\eqref{eq:defn-general-work-cost-thermodynamic-process} can be tractable: 
under certain conditions---see \cref{appx-sec:Dhypr-asymptotic-properties}---$\DHypr[R][\eta]`*{\mathcal{E}(\rho)}{\rho'}$ can approximate the hypothesis-testing relative entropy $\DHyp[\eta]`*{\mathcal{E}(\rho)}{\rho'}$ [defined in~\eqref{eq:setting-defn-DHyp}].
The latter can be evaluated via a semidefinite program (SDP), in principle~\cite{Dupuis2013_DH}.
Hence one can, in principle, evaluate each side of condition~\eqref{eq:general-work-cost-thermo-process-condition-success}, as is necessary to calculate the work cost~\eqref{eq:defn-general-work-cost-thermodynamic-process}.
(As a caveat, the SDP calculation's time and memory grow exponentially in $n$.)

\paragraph{Referee with extremely low computational power.}
Now, suppose that $R$ is extremely small.
The referee may have great difficulty distinguishing the output $\mathcal{E}(\rho)$ from the reference $\rho'$.
Consider the extreme case where $R=0$: the referee may perform only the local projective measurements in~\eqref{eq:setting-defn-Mrzero}.
We can apply a general upper bound: denote by $\sigma$ and $\tau$ any $n$-qubit states; and, by $\sigma_j$ and $\tau_j$, the respective reduced states of qubit $j$.
The complexity relative entropy obeys
\begin{align}
    \DHypr[R=0][\eta]`*{\sigma}{\tau} \lesssim \sum_{j=1}^n \DHyp[\eta]`*{\sigma_j}{\tau_j} \ .
    \label{eq:main-bound-cplx-rel-entr-r-eq-zero}
\end{align}
The approximation hides error terms that vanish in the limit as $\eta\to1$.
One can apply this inequality to evaluate~\eqref{eq:general-work-cost-thermo-process-condition-success} and thereby certify that $\mathcal{E}(\rho)$ is indistinguishable from $\rho'$ to a referee who lacks computational power and tolerates little error.
[Inequality~\eqref{eq:main-bound-cplx-rel-entr-r-eq-zero}~follows from \cref{thm:bound-DHypr-DHyp-tensor-products} in \cref{appx:GeneralConstructionComplexityEntropy}.]

\paragraph{Simple transformation.}
Suppose that the transformation $\mathscr{E}$ involves few primitive operations: $r < O(n)$.
Also, suppose that the system is 1D.
$\mathscr{E}$ can change the system's entanglement only by an amount proportional to $r$.
This entanglement is related to the complexity entropy at small $r$ (see \cref{sec:quantum-info--entgl-bounds} for details).
Hence we expect the first condition in~\eqref{eq:defn-general-work-cost-thermodynamic-process} to relax to a constraint on entanglement, which we expect to be more tractable.

\section{Information-theoretic features of the complexity entropy}
\label{sec:InformationTheoreticApplicationsComplexityEntropy}

The complexity relative entropy~\eqref{eq:setting-defn-DHypr} and the complexity entropy~\eqref{eq:setting-defn-HHypr} are additions to a large cast of information-theoretic entropy measures~\cite{BookTomamichel2016_Finite,PhfEntropyZoo}.
Intuitively, the complexity entropy quantifies a state's apparent randomness to an agent who can implement only limited-complexity measurement effects.
Unlike common entropies, the complexity entropy lacks unitary invariance.
Indeed, a computationally limited agent may need more work to erase a state after it has evolved under a complex unitary (cf.\@ examples in \cref{sec:main-erasure-Wcost-degenerate-H}).

In this section, we present the complexity entropy's properties, as well as its applications to information-theoretic tasks.
Most generally, the complexity (relative) entropy is defined for an abstract family $\{\Mr[r]\}$ of sets.
Each $\Mr[r]$ denotes the set of POVM effects of complexities at most $r \geq 0$.
This formalism generalizes the above, $n$-qubit constructions to arbitrary discrete quantum systems.
In \cref{appx:GeneralConstructionComplexityEntropy}, we present the general definition of the complexity (relative) entropy and prove the properties stated below.
Here, to achieve a simpler and more concrete presentation, we focus on an $n$-qubit system and on the family $`{\Mr[r]}$ defined by~\eqref{eq:setting-defn-Mr}.

For information-theoretic applications, it is convenient to measure entropy in units of bits, rather than in nats $[ 1\ \text{bit} = \log(2)\ \text{nats} ]$.
In this section, we denote with an overbar entropies expressed in units of bits; these entropies are related to their counterparts by the factor $\log(2)$:
\begin{align}
    \begin{split}
    \overline{H}(\rho) &= H(\rho)/\log(2)\ , \\
    \bDHyp[\eta]{\rho}{\sigma} &= \DHyp[\eta]{\rho}{\sigma} / \log(2)\ , \\
    \bHHyp[\eta]{\rho} &= \HHyp[\eta]{\rho} / \log(2)\ , \\
    \bDHypr[r][\eta]{\rho}{\sigma} &= \DHypr[r][\eta]{\rho}{\sigma} / \log(2)\ , \; \;
    \text{etc.}
    \end{split}
\end{align}
The factor $[\log(2)]^{-1}$ changes natural logarithms to base-2 logarithms in the entropies' definitions.
For instance, the von Neumann entropy $H(\rho) = -\tr`*(\rho\log\rho)$, expressed in units of bits, is $\overline{H}(\rho) = -\tr`*(\rho\log_2\rho)$.

\subsection{Overview of elementary properties}

In this subsection, $\rho$ denotes any quantum state, and $\Gamma$ any positive-semidefinite operator, defined on the $n$-qubit Hilbert space.
Let $r \geq 0$ and $\eta \in (0,1]$.
The complexity relative entropy inherits some properties from the hypothesis-testing relative entropy~\eqref{eq:setting-defn-DHyp}.
For example, $\bDHypr[r][\eta]{\rho}{\Gamma}$, like $\bDHyp[\eta]{\rho}{\Gamma}$, monotonically decreases as $\eta$ increases.
The greater an agent's error intolerance (the greater the $\eta$), the more mixed $\rho$ appears $[$the greater $\HHypr[r][\eta]{\rho}$ is$]$.
The complexity (relative) entropy also has properties that reflect its sensitivity to state complexity.
For example, $\bDHypr[r][\eta]{\rho}{\Gamma}$ monotonically increases as $r$ increases.
The greater an agent's computational power (the greater the $r$), the less mixed $\rho$ appears $[$the less $\HHypr[r][\eta]{\rho}$ is$]$.
These monotonicity properties imply that $\rho$ is more distinguishable from $\Gamma$ to an observer who has greater computational power and has a higher tolerance for guessing $\Gamma$ when given $\rho$.

Furthermore, the complexity (relative) entropy has the same range as the hypothesis-testing (relative) entropy:  
 \begin{subequations}
  \begin{align}
    \bDHypr[r][\eta]{\rho}{\Gamma}
    &\in `*[ -\log_2 \bm{(} \tr(\Gamma) \bm{)} \,,\; 
   - \log_2`*( \norm{\Gamma} ) ]
    \ , \; \; \text{and}
    \\
    \bHHypr[r][\eta]{\rho}
    &\in [0 \,,\; n] \ .
  \end{align}
 \end{subequations}
The complexity relative entropy, like standard relative entropies, enjoys a scaling property in its second argument: for all $\lambda>0$, $\bDHypr[r][\eta]{\rho}{\lambda\Gamma} = \bDHypr[r][\eta]{\rho}{\Gamma} - \log_2(\lambda)$.

By construction, the complexity (relative) entropy is not unitarily invariant.
On similar grounds, it does not satisfy a data-processing inequality.
Indeed, unitary evolution is a special case of data processing. Furthermore, a unitary can increase or decrease a pure state's complexity.
Therefore, a unitary can increase or decrease a state's complexity entropy.
Nevertheless, the complexity (relative) entropy monotonically increases (decreases) under partial traces.
Tracing out $k$ qubits from $\rho$ yields
\begin{subequations}
  \begin{align}
    \bDHypr[r][\eta]{\rho}{\Gamma}
    & \geq \bDHypr[r][\eta]`\big{\tr_k`(\rho)}{\tr_k`(\Gamma)} \ , \\
    \bHHypr[r][\eta]{\rho}
    &\leq \bHHypr[r][\eta]`\big{\tr_k`(\rho)} + k \ .
  \end{align}
\end{subequations}
The complexity (relative) entropy is greater (less) than or equal to the hypothesis-testing (relative) entropy: $\bHHypr[r][\eta]{\rho} \geq \bHHyp[\eta]{\rho}$ and $\bDHypr[r][\eta]{\rho}{\Gamma} \leq \bDHyp[\eta]{\rho}{\Gamma}$.
In the large-$r$ limit, one would expect the complexity (relative) entropy to converge to the hypothesis-testing (relative) entropy.
Yet, an exact convergence might not occur: some POVM effects might not be well approximated by any effect in $\Mr[r]$.
This is due to our choice of $\Mr[r]$ and to the lack of ancillas in our computational model. 
Nevertheless, we show in \cref{appx-sec:Dhypr-asymptotic-properties} that, under certain conditions, the complexity (relative) entropy approximates the hypothesis-testing (relative) entropy up to constant error terms.
Alternatively, one might hope to ensure convergence to the hypothesis-testing (relative) entropy by allowing the use of one ancillary qubit in implementations of measurement effects, as introduced in Ref.~\cite{Brandao_21_Models}.

\subsection{Bounds from well-known complexity measures}

To what extent does the complexity entropy $\bHHypr[r][\eta]{\psi}$ quantify the complexity of a pure state $\psi \equiv \proj{\psi}$?
We bound the complexity entropy in terms of two complexity measures: the \emph{strong complexity} of Ref.~\cite{Brandao2021PRXQ_models} and the \emph{approximate circuit complexity} (\cref{defn:approx-state-complexity} in \cref{appx:sec:complexity-of-measurement-effects}).
We describe the bounds qualitatively here, deferring rigorous statements, and proofs thereof, to \cref{appx-topic:complexity-entropy-relationship-to-state-complexity}.

The complexity entropy reflects approximate circuit complexity as follows.
Assume that the set $\mathcal{G}$, used in the definition~\eqref{eq:setting-defn-Mr}, contains only unitary gates.
Consider beginning with some reference state, then applying gates in $\mathcal{G}$ to prepare any state $\epsilon$-close to $\ket{\psi}$ in trace distance.
Let $C^{\epsilon}`*(\ket\psi)$ denote the least number of gates in any such process.
$C^{\epsilon}`*(\ket\psi)$ is the $\epsilon$-approximate circuit complexity of $\ket\psi$.
For all $r\geq C^{\epsilon}`*(\ket\psi)$, $\bHHypr[r][1-\epsilon^2]{\psi} \approx 0$.
Furthermore, $\bHHypr[r][1-\epsilon^2]{\psi} \not\approx 0$ for all $r < C^{2\epsilon}`*(\ket\psi)$ (see \cref{thm:complexity-entropy-pure-state-complexity}).
The complexity relative entropy also obeys an upper bound involving the \emph{strong complexity} of Ref.~\cite{Brandao2021PRXQ_models} (see \cref{thm:cplxrelentr-bound-strongcomplexity}).

\subsection{Bounds from entanglement}
\label{sec:quantum-info--entgl-bounds}

We lower-bound an $n$-qubit state's complexity entropy using a measure of the state's entanglement.
Consider a 1D chain $S$ of $n \geq 2$ qubits: $S = S_1 S_2 \ldots S_n$.
Let $\rho$ denote any state of $S$.
We define the following entanglement measure:
\begin{align}
  E(\rho) \coloneqq \frac1{n-1} \sum_{j=1}^{n-1} \II[\rho]{S_1 \ldots S_j}{S_{j+1} \ldots S_n} \ .
  \label{eq:entanglement-measure-1d-chain-mutual-infos}
\end{align}
$\II[\rho]{A}{B} \coloneqq \HH[\rho]{A} + \HH[\rho]{B} - \HH[\rho]{AB}$ denotes the quantum mutual information defined in terms of the von Neumann entropy $\HH[\rho]{X} \coloneqq -\tr`*(\rho\log\rho)$.
The mutual information quantifies all correlations---classical and quantum---including correlations due to entanglement.

Assume that the operations $\mathcal{G}$ in~\eqref{eq:setting-defn-Mr} are unitary gates that can act nontrivially only on two neighboring qubits (the gates are \emph{geometrically local}).
We upper-bound the change in $E(\rho)$ under one gate in $\mathcal{G}$.
\begin{mainproposition}[Change in $E(\rho)$ under a two-qubit unitary]
    \label{mainprop:change-in-entanglement-measure}
    Let $\rho$ denote any quantum state of a 1D chain $S$ of $n \geq 2$ qubits.
    Let $U$ denote any unitary that can act nontrivially only on two neighboring qubits in $S$.
    Then
    \begin{align}
        \abs`*{ E`*(U \rho U^\dagger) - E`*(\rho) } \leq \frac{8\log(2)}{n-1} \ .
        \label{eq:mainprop-change-in-entgl-measure---bound}
    \end{align}
\end{mainproposition}
To prove \cref{mainprop:change-in-entanglement-measure}, one observes that $U$ can alter only one mutual information in~\eqref{eq:entanglement-measure-1d-chain-mutual-infos}---the mutual information associated with (the bipartition that separates) the qubits on which $U$ can act nontrivially.
In \cref{appx-topic:entanglement-bounds}, we generalize~\eqref{eq:mainprop-change-in-entgl-measure---bound} to account for the \emph{potential entangling power}~\cite{Eisert2021PRL_entangling} of $\mathcal{G}$.
This power quantifies the entanglement generable by one gate in $\mathcal{G}$.

By repeatedly applying~\eqref{eq:mainprop-change-in-entgl-measure---bound}, we upper-bound the change in $E(\rho)$ under $\leq r$ gates in $\mathcal{G}$.
We then lower-bound the complexity entropy $\HHypr[r][\eta]{\rho}$ in terms of $E(\rho)$ as
\begin{align}
  \HHypr[r][\eta]{\rho} \geq \frac1{\eta} `*[ E(\rho) - r\,\frac{8\log(2)}{n-1} + H(\rho) + \text{error terms} ] \ .
  \label{eq:maintext-entanglement-bound-on-cplx-entropy}
\end{align}
The error terms depend on $\eta$ and $n$ and vanish in the limit as $\eta\to1$.
See \cref{appx-topic:entanglement-bounds} for details.
In \cref{appx-sec:Hamiltonian-dynamics}, we also investigate similar bounds that arise from natural dynamics under local Hamiltonians.
This setting is particularly relevant to locally interacting systems.

\subsection{Evolution of the complexity entropy under random circuits}
\label{sec-topic:main-random-circuits}

Consider an $n$-qubit circuit generated with gates sampled Haar-randomly from all two-qubit gates.
The circuit effects a unitary.
As the circuit depth grows, the unitary's complexity grows, until saturating at a value exponential in $n$.
The complexity growth is at least sublinear~\cite{Brandao2021PRXQ_models} and, for some nonrobust complexity measures~\cite{Haferkamp2022NPhys_linear,Li2022arXiv_short,Haferkamp2023arXiv_moments}, is linear.
We ask how well the complexity entropy tracks the complexity of a pure state evolving under a random circuit.

We consider \emph{brickwork circuits}~\cite{Haferkamp_21_Linear,Haferkamp2022Q_random}.
A brickwork circuit consists of staggered layers of nearest-neighbor two-qubit gates.

The following proposition reveals the behaviors of the complexity entropy at a fixed complexity scale $r$: under short circuits, the complexity entropy vanishes.
Around a circuit depth $t$ commensurate with $r$, the complexity entropy transitions to a near-maximal value (\cref{fig:ComplexityEntropyRandomCircuits}).
We bound the interval of $t$ values in which the transition happens to values satisfying $t\sim r$.
We prove the proposition using a strategy borrowed from Ref.~\cite{Brandao2021PRXQ_models}, which contains a similar proposition about the strong complexity.
\begin{mainproposition}[Transition of the complexity entropy under random circuits]
  \label{mainthm:ComplexityEntropyInRandomCircuits}
  Consider a depth-$t$, random $n$-qubit circuit $V$ whose gates act in a brickwork layout.
  Let $\ket\psi \coloneqq V\ket{0^n}$, $r\geq 0$, and $\eta \in (0,1]$.
  If the circuit is short, such that $t\leq O(r)$, then
  \begin{align}
    \bHHypr[r][\eta]`*{ \proj{\psi} } = 0 \ .
    \label{eq:mainthm-complexity-entropy-random-circuits-equals-zero}
  \end{align}
  If the circuit is long, such that $t \geq \Omega(r n^3 )$, then
  \begin{align}
    \bHHypr[r][\eta]`*{ \proj{\psi} } \geq n - \log_2`*(2/\eta) \ ,
    \label{eq:mainthm-complexity-entropy-random-circuits-lower-bound}
  \end{align}
  except with a probability $e^{-\Omega(n)}$ over the sampling of $V$.
\end{mainproposition}

The proof of~\eqref{eq:mainthm-complexity-entropy-random-circuits-equals-zero} is immediate: one can use the circuit $V$ to construct a candidate $Q = V^\dagger\proj{0^n} V \in\Mr[r]$ for the optimization~\eqref{eq:setting-defn-DHypr} that defines $\bHHypr[r][\eta]`*{ \proj{\psi} }$.
The proof of~\eqref{eq:mainthm-complexity-entropy-random-circuits-lower-bound} is presented in \cref{appx-topic:complexity-entropy-random-circuits}.

\subsection{Data compression under complexity limitations}
\label{sec:data-compression-under-computational-limitations}

The complexity entropy naturally quantifies the optimal efficiency of data compression under computational limitations.
Let $\rho$ denote any $n$-qubit state; $m$, a number of qubits; and $\epsilon \in [0,1]$, an error parameter.
To perform \emph{data compression under complexity limitations}, one seeks a unitary $U$, composed from $\leq r$ two-qubit gates, such that some subset $\mathcal{W}\subset `{1, 2, \ldots, n}$ of $\abs{\mathcal{W}} = m$ qubits satisfies
\begin{align}
  F^2 \bm{(} \tr_{\mathcal{W}}( U \rho U^\dagger ), \proj{0^{n-m}} \bm{)} \geq 1 - \epsilon \ .
\end{align}
One \emph{compresses} $\rho$ onto the qubits of $\mathcal{W}$, with an accuracy $1-\epsilon$ and using $\leq r$ gates.
For background, consider an agent who lacks complexity limitations (who operates in the limit as $r\to\infty$).
The hypothesis-testing entropy quantifies the optimal one-shot data-compression size of $\rho$~\cite{Dupuis2013_DH}
\begin{align}
  m_{\mathrm{optimal}} \approx \bHHyp[1-\epsilon]{\rho} \ .
\end{align}
Every protocol for data compression under complexity limitations can be mapped to a ``simple'' protocol for erasure.
The erasure protocol consists of computational gates followed by \textsc{reset} operations [see \cref{fig:ThermodynamicErasureComplexityConstraints}(b)].
The least number $m_{\mathrm{optimal}, r}$ of qubits onto which $\rho$ can be compressed, with an accuracy $1-\epsilon$ and using $\leq r$ gates, satisfies
\begin{align}
  \bHHypr[r][1-\epsilon]{\rho} - \log_2`*( 1/ `*[ 1-\epsilon ] )
  &\leq m_{\mathrm{optimal}, r}
  \leq \bHHypr[r][1-\epsilon]{\rho} \ .
  \label{eq:optimal-qubit-compression-number}
\end{align}
The complexity entropy $\bHHypr[r][1-\epsilon]{\rho}$ is defined with respect to the set $\mathcal{G}$ of all two-qubit gates.
In \cref{appx-topic:data-compression-complexity-limitations}, we prove~\eqref{eq:optimal-qubit-compression-number}.

\subsection{Complexity conditional entropy}

In information theory, conditional-entropy measures quantify the randomness of a system $A$, as apparent to an observer given access to a system $B$ that may share correlations with $A$.
The conditional entropy appears throughout classical and quantum information theory.
Applications include communication~\cite{Slepian1973TIT_noiseless,Abeyesinghe2009_MotherProtocols,Dupuis2014CMP_decoupling}, thermodynamic erasure with side information~\cite{delRio2011Nature}, and quantum entropic uncertainty relations~\cite{Coles2017RMP_entropic}.
Consider defining a conditional entropy $\Hbase{\bHSym}{*}[\rho]{A}[B]$ of a state $\rho_{AB}$.
A standard technique is to identify an appropriate relative entropy $\bDD_*{}{}$ and to set $\Hbase{\bHSym}{*}[\rho]{A}[B] \coloneqq -\bDD_*{\rho_{AB}}{\Ident_A\otimes\rho_B}$.
Here, $\rho_B \coloneqq \tr_A`(\rho_{AB})$ is the reduced state of $\rho$ on $B$~\cite{BookTomamichel2016_Finite}.

$\Hbase{\bHSym}{*}[\rho]{A}[B]$ reflects one's ability to distinguish $\rho_{AB}$ from a state that provides the same knowledge about $B$ ($B$ is in state $\rho_B$) but minimal knowledge about $A$ ($A$ is maximally mixed and uncorrelated with $B$).
In this interpretation, $\bDD_*{}{}$ quantifies the notion of distinguishability.
Conditional entropies defined as above include the conditional min- and max-entropies and the conditional R\'enyi-$\alpha$ entropies~\cite{BookTomamichel2016_Finite}.

We employ the complexity relative entropy~\eqref{eq:setting-defn-DHypr} to define the complexity conditional entropy.
Denote any arbitrary bipartite state by $\rho_{AB}$.
For all $r\geq 0$ and $\eta \in (0,1]$, the \emph{complexity conditional entropy} of $\rho_{AB}$ is
\begin{align}
  \bHHyprc[r][\eta][\rho]{A}[B] \coloneqq -\bDHypr[r][\eta]{\rho_{AB}}{\Ident_A\otimes\rho_B} \ .
  \label{eq:main-defn-conditional-complexity-entropy}
\end{align}

We prove the following properties of the complexity conditional entropy.
It, like the complexity entropy, decreases monotonically as $r$ increases and increases monotonically as $\eta$ increases.
Furthermore, the complexity conditional entropy is bounded in terms of the Hilbert-space dimensionality of $A$, $d_A$, as
\begin{align}
    -\log_2(d_A) \leq \bHHyprc[r][\eta][\rho]{A}[B]
    \leq \log_2(d_A) \ .
    \label{eq:cndtl-cplx-entropy-general-bounds}
\end{align}
The complexity conditional entropy also exhibits strong subadditivity: for every tripartite state $\rho_{ABC}$,
\begin{align}
    \bHHyprc[r][\eta][\rho]{A}[BC] \leq \bHHyprc[r][\eta][\rho]{A}[B] \ .
    \label{eq:cndtl-cplx-entropy-subadditivity}
\end{align}
The strong subadditivity follows from the complexity relative entropy's monotonicity under partial traces.
We prove~\eqref{eq:cndtl-cplx-entropy-general-bounds} and~\eqref{eq:cndtl-cplx-entropy-subadditivity} in \cref{appx-topic:complexity-conditional-entropy}.
The complexity conditional entropy has operational significance in a variant of the information-theoretic task of \emph{decoupling}, we show next.

\subsection{Decoupling from a reference system}
\label{sec:decoupling}

First, we define the decoupling of a system $A$ from a system $R$.
Define $d_A$ as the Hilbert-space dimensionality of $A$ and $\pi_A \coloneqq \Ident_A/d_A$ as a maximally mixed state.
We say that $A$ is maximally mixed and decoupled from $R$ if the joint state $\rho_{AR}$ is the product $\pi_A\otimes\rho_R$~\cite{BookWilde2013QIT}.
\emph{Decoupling $A$ from $R$} means transforming a general $\rho_{AR}$ into a state $\rho'_{AR}$ close to a state in which $A$ is maximally mixed and decoupled from $R$~\cite{Abeyesinghe2009_MotherProtocols,Dupuis2014CMP_decoupling}.
Decoupling is an information-theoretic primitive applied in randomness extraction~\cite{PhDRenner2005_SQKD,Tomamichel2011TIT_LeftoverHashing}, quantum communication~\cite{Abeyesinghe2009_MotherProtocols}, and thermodynamic erasure with side information~\cite{delRio2011Nature}.

The standard information-theoretic task of decoupling can be defined as follows.
Alice possesses an $n$-qubit system $A$ that might be correlated with a reference $R$.
The joint system begins in a state $\rho_{AR}$.
First, Alice applies to $A$ a unitary $U_0$, preparing $`*(U_0 \otimes \Ident_R) \rho_{AR} `*(U_0 \otimes \Ident_R)^\dagger \eqqcolon \rho_{AR}'$.
Then, Alice discards a subsystem $A_1$ formed from $k\geq 0$ qubits of $A$. A subsystem $A_2$, formed from $n-k$ qubits, remains.
Alice sends $A_2$ to a referee, who possesses $R$.
The referee attempts to distinguish the reduced state $\rho'_{A_2R}$ from $\pi_{A_2} \otimes \rho_R$.
Alice succeeds if the referee cannot distinguish the states to within some error tolerance.
Alice can ensure that $\rho_{A_2R}'$ approximates $\rho_{A_2}' \otimes \rho_R$ if~\cite{Dupuis2014CMP_decoupling}
\begin{align}
  k \gtrsim \frac12`*[ n - \bHmin[\rho][\epsilon]{A}[R] ] \ .
  \label{eq:decoupling-without-complexity-limitations}
\end{align}
$\epsilon > 0$ denotes a tolerance parameter.
$\bHmin[\rho][\epsilon]{A}[R]$ denotes the \emph{smooth conditional min-entropy} of $\rho$~\cite{Dupuis2014CMP_decoupling,BookTomamichel2016_Finite}.

Here, we introduce a complexity-constrained variant of decoupling.
Assume that $R$ consists of qubits.
Assume further that both Alice and the referee can implement two-qubit gates that form a set universal for quantum computation.
We define our decoupling task as the task above, except we impose two complexity constraints.
First, we constrain Alice's computational power, requiring that $U_0$ be implementable with at most $r_0 \geq 0$ gates.
Second, we constrain the referee's computational power.
Using at most $r_1\geq 0$ gates, the referee performs a hypothesis test of the form described in \cref{mainthm:DHypr-interpretation-hypo-test}, subject to type~I and type~II error constraints specified with $\eta\in (0,1]$ and $\delta\in (0,1]$, respectively.
Alice succeeds in our decoupling task if there exists no hypothesis test by which the referee can correctly accept $\rho_{A_2R}'$ with a probability $\geq \eta$ and incorrectly accept $\pi_{A_2} \otimes \rho_R$ with a probability $\leq \delta$.
Equivalently, by \cref{mainthm:DHypr-interpretation-hypo-test}, Alice is successful if
\begin{align}
  \bDHypr[r_1][\eta]`*{\rho'_{A_2 R}}{\pi_{A_2} \otimes \rho_R} \leq -\log_2`*( \delta/\eta ) \ .
  \label{eq:decoupling-condition-success-DHypr}
\end{align}
In terms of the complexity conditional entropy $\bHHyprc[r_1][\eta][\rho']{A_2}[R]$, the condition~\eqref{eq:decoupling-condition-success-DHypr} is
\begin{align}
    n - k - \bHHyprc[r_1][\eta][\rho']{A_2}[R] \leq - \log_2`*(\delta/\eta) \ .
\end{align}
Thus, Alice succeeds if $\bHHyprc[r_1][\eta][\rho']{A_2}[R]$ is sufficiently close to its maximum value, $n-k$.

\begin{figure}
  \centering
  \includegraphics{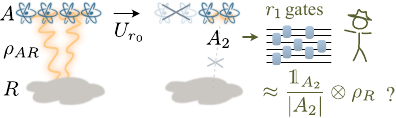}
  \caption{Decoupling with respect to a computationally bounded referee.  
    Alice's system $A$ might be correlated with a reference system $R$.
    Alice applies to $A$ a unitary $U_0$, then discards $k$ qubits of $A$.
    Alice retains an $(n-k)$-qubit system $A_2$.
    Alice succeeds in our decoupling task if a referee cannot, using a circuit of $\leq r_1$ gates, distinguish the final state of $A_2 R$ from $\pi_{A_2} \otimes\rho_R$.
    The success condition depends on the conditional complexity entropy $\bHHyprc[r_1][\eta][\rho']{A_2}[R]$.
    Conjecturing a property of that entropy, we bound, in terms of $\bHHyprc[r][\eta][\rho]{A}[R]$, the number of qubits that Alice can decouple from $R$.
    Our bound mirrors the bound~\eqref{eq:decoupling-without-complexity-limitations}, which holds in the absence of complexity constraints.
    }
  \label{fig:SetupTaskDecoupling}
\end{figure}

Before quantifying Alice's decoupling capabilities, we posit two expectations.
First, we do not expect Alice's complexity constraint to meaningfully change the number of qubits that she can decouple from $R$ (unless $r_0$ is very small).
Indeed, one may achieve near-optimal decoupling by choosing for $U_0$ to be a random circuit of only $O \bm{(} n \log^2(n) \bm{)}$ gates, in an all-to-all-coupling model~\cite{Brown2015CMP_decoupling}.

In contrast, we expect the referee's complexity constraint to substantially increase the number of qubits that Alice can decouple from $R$, by enabling Alice to succeed with strategies that would fail at standard decoupling.
For instance, consider a highly complex, maximally entangled state $\rho_{AR}$.
In the standard scenario, Alice cannot decouple any qubits from $R$.
[We can check this claim by substituting $\bHmin[\rho][\epsilon]{A}[R] = -n$ in~\eqref{eq:decoupling-without-complexity-limitations}.]
Now, suppose that the referee is computationally limited.
Alice's system is apparently decoupled from $R$ already: a complex state can be indistinguishable from a maximally mixed state, so the referee cannot distinguish $\rho_{AB}$ from a decoupled state.
More precisely, suppose that $\rho_{AR}$ is maximally entangled, that $\eta = 1$, and that the referee has unlimited computational power.
Alice cannot decouple $n-k$ qubits from $R$, if $\delta \leq 2^{-2(n-k)}$.
$[$If Alice discards $k$ qubits, then $\rho_{A_2R}'$ is a uniform mixture of $2^k$ pure states.
In an optimal strategy, the referee transforms $\rho_{A_2R}'$ into $\rho_{A_2R}'' \coloneqq \proj{0^{2(n-k)}} \otimes `*( \Ident_2^{\otimes k}/2^k )$, before implementing the measurement effect 
\begin{equation}
  Q \coloneqq \proj{0^{2(n-k)}} \otimes \Ident_2^{\otimes k} \ .
\end{equation}
Consequently, the referee can distinguish $\rho_{A_2R}'$ from $\pi_{A_2} \otimes \rho_R = \pi_{A_2 R}$: $\bDHypr[r_1][\eta]{\rho'_{A_2 R}}{\pi_{A_2} \otimes \rho_R} = -\log_2 \bm{(} \tr`(Q \pi_{A_2R}) / \tr`(Q\rho_{A_2R}'') \bm{)} = 2(n-k) \leq -\log_2`*( \delta/\eta )$.$]$
Now, however, suppose that the referee's computational power is substantially limited and that $\rho_{AR}$ is highly complex.
Alice may be able to decouple $n-k$ (or more) qubits even if $\delta \leq 2^{-2(n-k)}$: the referee may require more computational power to perform a unitary necessary for distinguishing $\rho'_{A_2R}$ from $\pi_{A_2} \otimes \rho_R$.

How many qubits $k$ must Alice discard to convince the referee that a maximally mixed, decoupled state was prepared?
We present a bound that generalizes~\eqref{eq:decoupling-without-complexity-limitations} to our decoupling task.
The bound holds if one assumes that the complexity conditional entropy obeys an inequality reminiscent of the R\'enyi-entropy chain rule~\cite{Vitanov2013_chainrules,Dupuis2015JMP_chain}.
\begin{mainconjecture}[Chain rule for the complexity conditional entropy]
  \label{mainconjecture:HHypr-chain-rule-like}
  Let $\rho_{ABR}$ denote a quantum state of systems $A$, $B$, and $R$ of $n_A$, $n_B$, and $n_R$ qubits.
  Let $r\geq 0$ and $\eta \in (0,1]$.
  We conjecture that
  \begin{align}
    \bHHyprc[r][\eta][\rho]{B}[R] \leq \bHHyprc[r][\eta][\rho]{AB}[R] + n_A \ .
    \label{eq:main-conjecture-for-decoupling-bound-HHypr-chain-rule-like}
  \end{align}
\end{mainconjecture}

In words, introducing a system $A$ can decrease the complexity conditional entropy by, at most, the size $n_A$ of $A$.
The von Neumann conditional entropy $\bHH{B}[R]$ obeys an analogous bound, due to a chain rule and the lower bound $\bHH{A}[BR] \geq - n_A$: $\bHH{B}[R] = \bHH{AB}[R] - \bHH{A}[BR] \leq \bHH{AB}[R] + n_A$.
This $\bHH{B}[R]$ bound reflects how introducing $A$ can resolve only as much randomness as $A$ can store.
Indeed, the bound saturates whenever $A$ is maximally entangled with $B$ ($R$). 
In such a case, $B$ and $R$ are uncorrelated, and $A$ purifies a maximally mixed state of $B$ ($R$). Hence $A$ resolves $n_A$ maximally mixed qubits' randomness.
Similar bounds follow from chain rules for R\'enyi entropies~\cite{Vitanov2013_chainrules,Dupuis2015JMP_chain}.

Using \cref{mainconjecture:HHypr-chain-rule-like}, we prove a bound on the number of qubits that Alice can decouple from $R$.
\begin{maintheorem}[Upper bound on the number of qubits that Alice can decouple]
    \label{mainthm:DecouplingLowerBound}
    Consider the complexity-restricted decoupling described above and depicted in \cref{fig:SetupTaskDecoupling}.
    Suppose $r_1 \geq r_0$.
    Assume \cref{mainconjecture:HHypr-chain-rule-like}.
    Under the condition~\eqref{eq:decoupling-condition-success-DHypr} (if Alice is successful), then
    \begin{align}
        k \geq \frac12`*[ n - \bHHyprc[r_1-r_0][\eta][\rho]{A}[R] + \log_2`*( \delta/\eta ) ] \ .
        \label{eq:main-thm-DHypr-randomness-extraction--lower-bound}
    \end{align}
\end{maintheorem}
We prove \cref{mainthm:DecouplingLowerBound} in \cref{sec:appendix-DHypr-randomness-extraction}.
The bound~\eqref{eq:main-thm-DHypr-randomness-extraction--lower-bound} mirrors the bound~\eqref{eq:decoupling-without-complexity-limitations}, which governs decoupling without complexity restrictions.
Indeed, if $\eta\approx 0$, the conditional hypothesis-testing entropy approximates the smooth conditional min-entropy~\cite{Dupuis2013_DH}.
As such, one can interpret the conditional complexity entropy, when $\eta\approx 0$, as a \emph{complexity-aware version of the conditional min-entropy}.

The general achievability of the
bound~\eqref{eq:main-thm-DHypr-randomness-extraction--lower-bound} is unknown.
A standard technique for proving standard decoupling's achievability involves
bounding the trace distance between a protocol's output and a decoupled target, using the Hilbert--Schmidt norm~\cite{Abeyesinghe2009_MotherProtocols,%
  Yard2009IEEETIT_source,Dupuis2014CMP_decoupling}.
One cannot readily extend this technique to \cref{mainthm:DecouplingLowerBound}: while the trace distance's operational definition naturally extends to complexity-restricted POVM effects~\cite{Brandao2021PRXQ_models}, such an extension is unknown for the Hilbert--Schmidt norm.

\cref{mainconjecture:HHypr-chain-rule-like} appears to preclude a kind of pseudorandom quantum state that we call a \emph{pseudomixed state}.
Let $A$ and $B$ denote systems of $n_A$ and $n_B$ qubits, respectively.
Let $\ket\chi_{AB}$ denote a pure state of $AB$.
Let $r \geq 0$ and $\eta \in (0,1]$.
$\ket\chi_{AB}$ is an \emph{efficiently preparable pseudomixed state} on $B$ if the complexity entropy of $\ket\chi$ is low, $\bHHypr[r][\eta]{\proj\chi} \approx 0$, while the complexity entropy of the reduced state $\chi_B \coloneqq \tr_A`*(\proj\chi)$ is high:
\begin{align}
    \bHHypr[r][\eta]{\chi_B} \geq \Omega \bm{(} \exp(n_A) \bm{)} \ .
    \label{eq:complexity-pseudorandom-quantum-state}
\end{align}
In other words, $B$ appears exponentially more mixed without the purifying $A$.
The condition~\eqref{eq:complexity-pseudorandom-quantum-state} directly contradicts the conjecture~\eqref{eq:main-conjecture-for-decoupling-bound-HHypr-chain-rule-like}.
Therefore, any efficiently preparable pseudomixed state would disprove \cref{mainconjecture:HHypr-chain-rule-like}.
We can therefore interpret \cref{mainconjecture:HHypr-chain-rule-like} as a \emph{no-efficient-pseudomixedness conjecture}.

A na\"ive approach to efficiently constructing a pseudomixed state from
pseudorandom
states~\cite{Ji2018CRYPTO_pseudorandom,Brakerski2019arXiv_pseudorandom} may
fail.
Standard constructions of pseudorandom states take as inputs
short, random bit strings $\kappa$.
The constructions yield efficient circuits for states $\ket{\phi_\kappa}$ that have the following property: no efficient quantum algorithm can distinguish $\poly(n)$ copies of $\ket{\phi_\kappa}$ from the same number of copies of a Haar-random state.
One could na\"ively expect to construct a pseudomixed state efficiently in the following way.
Consider an $n$-qubit system formed from two subsystems: an $O \bm{(} \log(n) \bm{)}$-qubit subsystem $A$ and an $O(n)$-qubit subsystem $B$.
One can efficiently prepare the state
\begin{equation}
  \ket\psi_{AB} \coloneqq \sum_\kappa \ket{\kappa}_A\otimes\ket{\phi_\kappa}_B 
\end{equation}
from $\ket{0^n} \, ,$
because pseudorandom states are efficiently preparable.
Let $r_{\mathrm{gen}}\leq\poly(n)$ denote the number of gates required to prepare $\ket\psi_{AB}$.
(One can generate $\ket\psi_{AB}$ by first preparing $A$ in $\sum_\kappa \ket{\kappa}_A$, by applying Hadamard gates.
Then, one implements a controlled unitary $\sum_\kappa \proj{\kappa}_A \otimes U_\kappa$.
$U_\kappa$ denotes an efficient unitary that prepares $\ket{\phi_\kappa}$.)
For all $\eta$, $\bHHypr[r_{\mathrm{gen}}][\eta]{\proj\psi} = 0$. Denote by $\psi_B \coloneqq \tr_A( \ketbra{\psi}{\psi})$ the reduced state of $B$. By discarding $A$, one forgets the seed $\kappa$ used to construct the pseudorandom state $\ket{\phi_\kappa}_B$. Therefore, no efficient algorithm can distinguish $\psi_B$ from a Haar-random state.
$\ket\psi_{AB}$ may appear to violate \cref{mainconjecture:HHypr-chain-rule-like}, looking highly mixed to a computationally limited observer.
Nevertheless, $\psi_B$ may have a low complexity entropy, since the pseudorandom states have low complexities.
($\psi_B$ appears mixed to a computationally limited observer because the state's preparation circuit, represented by the seed, is unknown.
Evaluating the complexity entropy, we imagine an observer who is given black-box access to a circuit that prepares the state.
Under such circumstances, pseudorandom states do not look random.
To appear highly mixed with respect to the complexity entropy, a pure state must have high complexity.)
Thus, $\ket\psi_{AB}$ does not necessarily constitute a counterexample to \cref{mainconjecture:HHypr-chain-rule-like}.

\section{Discussion}
\label{sec:Discussion}

Our framework highlights quantum complexity in thermodynamics.
Thermodynamics is operational: it concerns how efficiently an agent, given certain resources, can accomplish certain tasks.
We incorporate into a version of the theory complexity restrictions on (i) the agent's operations and (ii) the evaluation of the agent's output.
This ``agent'' could be human, could consist of a system's natural dynamics, etc. 

Our framework introduces a new role for time in thermodynamics: a complex process requires a long time.
Conventional thermodynamics distinguishes only what is possible and impossible.
 Augmented with complexity, our model of thermodynamics dictates what is practical and impractical. 
A spin glass illustrates the need to incorporate complexity restrictions into thermodynamics: according to the conventional thermodynamic laws, a spin glass can cool to low-temperature states.
However, this cooling process would require a very long time, or a highly complex process.

We illustrated the interplay between complexity and conventional thermodynamics through Landauer erasure.
Consider aiming to erase a highly complex pure state.
One could ``uncompute'' the state, paying in complexity.
Alternatively, one could erase every qubit, paying work.
More generally, consider resetting an arbitrary state $\rho$ with a success probability $\eta$.
The complexity entropy $\HHypr[r][\eta]{\rho}$ quantifies the trade-off between the complexity and work costs. 

Landauer erasure is one of many thermodynamic and information-processing tasks.
We analyzed data compression and decoupling, as well.
Many tasks merit analysis in the presence of complexity constraints.
We expect complexity-restricted entropies to quantify these tasks' optimal efficiencies. 

Our entropies quantify how mixed a state appears to a computationally limited agent.
For instance, even a pure state can have a high complexity entropy, if uncomputable by an agent.
Hence the complexity entropy quantifies a variation on pseudorandomness---more specifically, pseudomixedness: randomness apparent in determinstic phenomena, due to the observer's computational limitations. 

Unlike the standard hypothesis-testing entropy, the complexity entropy cannot be approximated via convex optimization.
We expect that computing the complexity entropy is typically hard, given strong evidence that computing the state complexity is hard~\cite{Gosset2014QIC_algorithm,VanDeWetering2023arXiv_tcount}.
Yet it might be possible to derive more-tractable bounds for the complexity entropy in specific settings.
Examples include low-complexity regimes as well as settings featuring random dynamics.
We derived such a bound for pure states evolving under random circuits (\cref{sec-topic:main-random-circuits}).

Other measures have been defined to quantify complexity or apparent randomness: the computational min-entropy~\cite{Chen2017arXiv_computational}, the coarse-grained entropy~\cite{GellMann2007PRA_quasiclassical}, the observational entropy \cite{Safranek2019PRA_observational}, the logical depth~\cite{Bennett1988UTMHS_logical}, and the quantum Kolmogorov complexity~\cite{Berthiaume2001JCSS_qKolmogorov,Mora2006arXiv_qKolmogorov}.
These measures reflect different approaches to quantifying complexity, as illustrated in the introduction.
Rigorous comparisons between these quantities and the complexity entropies merit future work.

The main text illustrates the complexity entropies' properties and usefulness; but generalizations are possible, and many are provided in \cref{appx:GeneralConstructionComplexityEntropy}.
For example, systems can be generalized beyond qubits to qudits.
We expect extensions to continuous-variable systems to be achievable.
Also, the complexity measure used can be replaced, as with Nielsen's complexity~\cite{Nielsen_05_Geometric,Nielsen_06_Quantum,Nielsen_06_Optimal,Dowling_06_Geometry}.
One could even replace $r$ with a matrix-product-state bond dimension---anything that constrains the POVM effects implementable.
One need only specify those POVM effects to construct general complexity entropies.

Another opportunity is to literally analyze---break apart---our \textsc{reset} and \textsc{extract} operations.
We have presented these operations as a computation's basic units.
Yet each operation may consist of multiple steps, entailing extra costs.
To incorporate these costs into our proofs, one could apply results from Ref.~\cite{Taranto2023PRXQ_landauer}.
Furthermore, we attributed to the \textsc{reset} the ideal work cost of $k_{\rm B} T \ln(2)$.
Any realistic \textsc{reset} will cost more work.
Such extra costs can be absorbed straightforwardly into our assumptions.
Alternatively, we could consider a different set of primitives $\mathcal{E}$ that have fixed work costs $W(\mathcal{E})$ and complexity costs $C(\mathcal{E})$.
(Any primitive operation could cost a nonzero amount of work and a nonzero amount of complexity.)
One possibility is to follow the approach in Ref.~\cite{Faist2018PRX_workcost}.
That is, allow as a primitive operation every completely positive, trace-preserving two-qubit map $\mathcal{E}$.
Then assign to $\mathcal{E}$ a complexity cost of 1 and a work cost quantified by
$\log \opnorm{ e^{\beta H/2} \mathcal{E}(e^{-\beta H})e^{\beta H/2} }$, wherein $H$ denotes the operated-on system's Hamiltonian.
Using this approach, one could lower-bound a process's thermodynamic costs~\cite{Faist2018PRX_workcost}.

Complexity was recently incorporated into thermodynamics alternatively, in a resource theory~\cite{YungerHalpern2022PRA_resource}.
In the resource theory of uncomplexity, an agent can perform arbitrarily many free operations.
Each operation is slightly noisy, and the agent will have a natural noise tolerance.
Hence the agent will naturally limit the number of operations they perform.
In the present paper, the agent's complexity restriction is a hard, external constraint.
Just as complexity has recently emerged in resource theories, pseudorandomness has emerged in studies of quantum gravity~\cite{Kim2020JHEP_ghost,Bouland2019arXiv_computational,Susskind2020arXiv_exptime}.
We therefore anticipate the complexity entropy's utility in understanding black holes in the context of the AdS/CFT correspondence, uniting tools of quantum information theory, high-energy physics, and statistical physics.

\begin{acknowledgments}
  The authors thank Matthew Coudron, Richard Kueng, Lorenzo Leone, Yi-Kai Liu, Carl Miller, Ralph Silva, and Twesh Upadhyaya for valuable discussions.
  This research has been supported by the National Institute of Standards and Technology and the Joint Center for Quantum Information and Computer Science under NIST Grant No. 70NANB21H055\textunderscore0, as well as by the Ford Foundation and by the National Science Foundation (QLCI Grant No. OMA-2120757).
  A.~M.~thanks Freie Universität Berlin for its hospitality and acknowledges The Ford Foundation for its support.
  The Berlin team has been funded by the DFG (FOR 2724 and CRC 183), the FQXi, and the ERC (DebuQC).
  N.B.T.K. acknowledges support by the European Space Agency (EISI Project No. 2021-01250-ESA), and by the Spanish MICIN (Project No. PID2022-141283NB-I00) with the support of FEDER funds.
  J.~H.~was funded by the Harvard Quantum Initiative.
\end{acknowledgments}







\clearpage
\newgeometry{hmargin=1.25in,vmargin=1in,marginparsep=12pt,marginparwidth=50pt}%
\onecolumngrid
\setstretch{1.12}%
\appendix

\section*{APPENDICES}

\section{Preliminaries, notation, and useful lemmas}

Throughout these Appendices, we use the following definitions and conventions.
We denote the set of non-negative reals by $\mathbb{R}_+ \coloneqq `{ x \in \mathbb{R}: x \geq 0}$.
We consider a system $S$ equipped with a Hilbert space $\Hs_S$ of a finite dimensionality $d_S$.
For any system $R$ distinct from $S$, we denote the composition of $S$ and $R$ by $SR$.
$\Hs_S$, as well as spaces of operators and of superoperators on $\Hs_S$,
are equipped with the standard topology of $\mathbb{C}^{d_S}$.
We denote the identity operator on $\Hs_S$ by $\Ident_S$.
We denote the identity superoperator on $\Hs_S$ by $\IdentProc[S]{}$.
We denote the single-qubit identity operator by $\Ident_2$.
For any Hermitian operators $A$ and $B$, we write $A \leq B$ if $B-A$ is positive-semidefinite.
A \textit{subnormalized quantum state} $\rho_S$ is a positive-semidefinite operator that satisfies $\tr(\rho_S) \leq 1$.
If $\tr(\rho_S) = 1$, then $\rho_S$ is normalized and represents a physical quantum state.
In the absence of ambiguity, we omit system subscripts from operators.
We denote the projector onto a pure state $\ket\psi$ by $\psi \coloneqq \proj{\psi}$.
We denote by $\pi_S \coloneqq \Ident_S/d_S$ the maximally mixed state of $S$.
For any sets $X$ and $Y$ of operators, we define $X \otimes Y \coloneqq \{ A \otimes B :\, A \in X,\, B \in Y \}$.
An operator $A$ has the operator norm $\opnorm{A}$, defined as the greatest singular value of $A$.
The trace norm $\onenorm{A}$ is the sum of the singular values of $A$.
We denote by $\ket{0}$ the eigenvalue-1 eigenvector of the Pauli-$z$ operator.
Tensor products are notated as $\ket{0^k} \coloneqq \ket{0}^{\otimes k}$.
A \emph{POVM effect} $Q$ on a system $S$ is a (Hermitian) positive-semidefinite operator satisfying $0 \leq Q \leq \Ident_S$.
Finally, all logarithms are base-$e$, unless otherwise indicated.

The \emph{diamond norm} of a superoperator $\mathcal{X}$ on $\Hs_S$ is defined as
\begin{align}
  \dianorm{\mathcal{X}} \coloneqq \max_{\rho_{SR}} `*{ \onenorm{ (\mathcal{X}\otimes\IdentProc[R]{})(\rho_{SR}) } } \ .
  \label{eq:def-dianorm}
\end{align}
$R$ denotes any system isomorphic to $S$, and the maximization is defined with respect to all subnormalized states $\rho_{SR}$.
For any subnormalized state $\rho$ on $S$, the diamond norm $\dianorm{\mathcal{X}}$ upper-bounds the trace norm $\onenorm{\mathcal{X}(\rho)}$ as
\begin{align}
    \onenorm{ \mathcal{X}(\rho) } = \onenorm{ (\mathcal{X}\otimes\IdentProc[R]{})(\rho_S \otimes \rho_R) } \leq \dianorm{\mathcal{X}} \ .
    \label{eq:dianorm-upper-bounds-onenorm}
\end{align}
For any completely positive, trace-preserving maps $\mathcal{E}$ and $\mathcal{F}$, the diamond distance $\frac{1}{2}\dianorm{ \mathcal{E} - \mathcal{F} }$ quantifies the one-shot distinguishability of $\mathcal{E}$ and $\mathcal{F}$~\cite{BookWilde2013QIT}.

For any state $\rho$ of $S$, the \textit{von Neumann entropy} of $\rho$ is $\HH{\rho} \equiv \HH[\rho]{S} \coloneqq -\tr \bm{(} \rho \log(\rho) \bm{)}$.
For any positive-semidefinite operator $\Gamma$ on $\Hs_S$, the (Umegaki) \textit{quantum relative entropy} of $\rho$ with respect to $\Gamma$ is $\DD{\rho}{\Gamma} \coloneqq \tr \bm{(} \rho [ \log(\rho) - \log(\Gamma) ] \bm{)}$.
If $S$ is a joint system, $S = XY$, the \textit{conditional quantum entropy} $\HH[\rho]{X}[Y]$ is $\HH[\rho]{X}[Y] \coloneqq \HH[\rho]{XY} - \HH[\rho]{Y}$.

The following inequality is used often in quantum information theory. The inequality is equivalent to a general lower bound on a version of the conditional min-entropy~\cite{PhDRenner2005_SQKD,BookTomamichel2016_Finite}.

\begin{lemma}[Partial order for joint states]
  \label{thm:state-partial-order-lemma}
  Let $X$ and $Y$ denote distinct quantum systems.
  If $\rho_{XY}$ denotes any subnormalized state on $XY$, then
  \begin{align}
      \frac{1}{d_X} \, \rho_{XY} \leq \Ident_X \otimes \rho_Y \ .
  \end{align}
\end{lemma}
\begin{proof}[**thm:state-partial-order-lemma]
  Let $C$ denote a system isomorphic to $X$. Let $`{ \ket{j}_X }$ and $`{ \ket{j}_C }$ denote bases for the eigenspaces of $X$ and $C$.
  If $\ket{\chi}_C \coloneqq \sum_j \ket{j}_C / \sqrt{d_X}$, then $\rho_{XY} \otimes \chi_C \leq \rho_{XY} \otimes \Ident_C$, since $\chi_C \leq \Ident_C$.
  Every positive map preserves the partial order on operators.
  [If $\mathcal{N}$ is a positive (linear) map acting on operators $A$ and $B \geq A$, then $\mathcal{N}(B) \geq \mathcal{N}(A)$: $\mathcal{N}(B-A)=\mathcal{N}(B)-\mathcal{N}(A)$ is positive-semidefinite, since $B-A$ is positive-semidefinite.]
  Hence, conjugating both sides of $\rho_{XY} \otimes \chi_C \leq \rho_{XY} \otimes \Ident_C$ by $\sum_j \proj{j}_X\otimes\bra{j}_C$ yields the pinching inequality
  \begin{align}
    \rho_{XY} =
    d_X `*(\sum_j \proj{j}_X\otimes\bra{j}_C) (\rho_{XY}\otimes\chi_C) `*(\sum_k \proj{k}_X\otimes\ket{k}_C)
    \leq d_X \sum_j \proj{j}_X \, \bra{j}_X \rho_{XY} \ket{j}_X \ .
    \label{eq:pinching-inequality-computational-basis}
  \end{align}
  In terms of the Fourier-transformed states $\ket{\tilde l}_X \coloneqq \sum_j \ee^{i2\pi lj/d_X} \ket{j}_X / \sqrt{d_X}$, the pinching inequality is
  \begin{align}
    \rho_{XY} \leq
    d_X \sum_l \proj{\tilde l}_X \, \bra{\tilde l}_X \rho_{XY} \ket{\tilde l}_X \ .
    \label{eq:pinching-inequality-fourier-basis}
  \end{align}
  Applying the positive map $d_X \sum_j \proj{j}_X \, \bra{j}_X (\cdot) \ket{j}_X$ to both sides of~\eqref{eq:pinching-inequality-fourier-basis} yields
  \begin{align}
    d_X \sum_j \proj{j}_X \, \bra{j}_X \rho_{XY} \ket{j}_X
    \leq d_X^2 \sum_{jl} \proj{j}_X \abs`*{\braket{j}{\tilde l}_X}^2 \bra{\tilde l}_X \rho_{XY} \ket{\tilde l}_X
    = d_X \Ident_X \otimes \tr_X`*(\rho_{XY}) = d_X \Ident_X \otimes \rho_Y \ .
    \label{eq:positive-map-on-pinching-inequality}
  \end{align}
  We have used $\abs`*{\braket{j}{\tilde l}_X}^2 = 1/d_X$.
  Chaining together~\eqref{eq:pinching-inequality-computational-basis} and~\eqref{eq:positive-map-on-pinching-inequality} proves the lemma.
\end{proof}

The following lemma provides a useful expression for, and bound on, the diamond distance between a unitary operation and the identity process.

\begin{lemma}[Diamond distance between a unitary operation and the identity process]
  \label{thm:diamond-norm-unitary-to-ident}
  Let $U$ denote any unitary operator, and let $\mathcal{U}(\cdot) \coloneqq U `(\cdot) U^\dagger$.
  It holds that 
  \begin{align}
    \frac12 \dianorm{ \mathcal{U} - \IdentProc{} } &= \sin `*(  \min`*{ e(U), \frac{\pi}{2} } ) ,
    \label{eq:thm-diamond-norm-unitary-to-ident-eU}
  \end{align}
  where
  \begin{align}
    e(U)
    &\coloneqq \min`*{
      \opnorm{H} \,:\; H = H^\dagger , \; \exists\chi\in\mathbb{R} \text{ such that } \ee^{-iH} = \ee^{-i\chi}U 
    } \ .
    \label{eq:thm-diamond-norm-unitary-to-ident-eU-defn-eU}
  \end{align}
  Furthermore,
  \begin{align}
    \frac12 \dianorm{ \mathcal{U} - \IdentProc{} } \leq \opnorm{ U - \Ident } \ .
    \label{eq:thm-diamond-norm-bound-eU-with-op-norm-U}
  \end{align}
\end{lemma}

The quantity $e(U)$ is called the \textit{potential entangling power} of $U$; a similar quantity is defined in Ref.~\cite{Eisert2021PRL_entangling}.
$e(U)$ is insensitive to global phases: for all $\chi\in\mathbb{R}$, $e(\ee^{i\chi} U) = e(U)$.

Inequality~\eqref{eq:thm-diamond-norm-bound-eU-with-op-norm-U} automatically yields an upper bound on the diamond distance between arbitrary unitary operations.
For all unitaries $U$ and $V$ acting on the same Hilbert space, the operations $\mathcal{U}(\cdot) \coloneqq U`(\cdot) U^\dagger$ and $\mathcal{V}(\cdot) \coloneqq V`(\cdot) V^\dagger$ obey
\begin{align}
    \frac12\dianorm{ \mathcal{U} - \mathcal{V} }
    = \frac12\dianorm{ \mathcal{V}^\dagger \mathcal{U} - \IdentProc[]{} } 
    \leq \opnorm{V^\dagger U - \Ident}
    = \opnorm{U - V}\ .
\end{align}
The first (second) equality follows from the diamond (operator) norm's invariance under unitary operations (operators).

\begin{proof}[*thm:diamond-norm-unitary-to-ident]
  By Proposition~18 of Ref.~\cite{Nechita2018JMP_almost}, $\frac12 \dianorm{ \mathcal{U} - \IdentProc{} } = \sin`*(\min`{\alpha,\pi/2})$, wherein $\alpha$ denotes half the angle of the shortest arc $A$, on the unit circle, that contains every eigenvalue (eigenphase) of $U$.
  Let $\chi_* \in \mathbb{R}$ such that $\ee^{i\chi_*}$ is the midpoint of $A$.
  (Since $U$ has only finitely many eigenphases, $A$ does not encompass the unit circle, so $A$ has endpoints and a midpoint.)
  By definition, $A = `*{ \ee^{i\chi_*} \ee^{i\varphi} \,:\; \varphi \in [-\alpha, \alpha] }$.
  $A$ has the endpoints $\theta_\pm \coloneqq \ee^{i\chi_*} \ee^{\pm i \alpha}$.
  $\theta_\pm$ are eigenphases of $U$.
  Otherwise, one could obtain an arc shorter than $A$ that contains every eigenphase of $U$.

  To prove~\eqref{eq:thm-diamond-norm-unitary-to-ident-eU}, we show that $e(U) = \alpha$.
  Let $-iH_*$ denote the unique logarithm of $\ee^{-i\chi_*} U$ such that every eigenvalue of $H_*$ lies in $(-\pi,\pi]$.
  Since $H_*$ is Hermitian and $\ee^{-iH_*} = \ee^{-i\chi_*} U$, $e(U) \leq \opnorm{H_*}$.
  Because every eigenvalue of $U$ lies in $A$, every eigenvalue of $\ee^{-iH_*}$ lies in the arc $A_* \coloneqq `*{ \ee^{i\varphi} \,:\; \varphi \in [-\alpha, \alpha] }$, which one obtains by multiplying every point in $A$ by $\ee^{-i\chi_*}$.
  Consequently, every eigenvalue of $H_*$ lies in $[-\alpha, \alpha]$, so $\opnorm{H_*} \leq \alpha$.
  Hence, $e(U) \leq \alpha$.
  
  For the sake of contradiction, suppose $e(U) < \alpha$.
  Let $H$ denote any Hermitian operator that achieves the minimum in~\eqref{eq:thm-diamond-norm-unitary-to-ident-eU-defn-eU}.
  $\ee^{-iH} = \ee^{-i\chi} U$ for some $\chi \in \mathbb{R}$, and $\opnorm{H} = e(U) < \alpha$.
  Since every eigenvalue of $H$ lies in $`*[-\opnorm{H}, \opnorm{H}]$, the arc $A_H \coloneqq `*{ \ee^{i\chi_*} \ee^{i\varphi} \,:\; \varphi \in `*[-\opnorm{H}, \opnorm{H}] }$ contains every eigenvalue of $\ee^{i\chi} \ee^{-iH} = U$.
  Having the half-angle $\opnorm{H} < \alpha$, $A_H$ is shorter than $A$, a contradiction.
  Therefore, $e(U) = \alpha$.

  We now prove~\eqref{eq:thm-diamond-norm-bound-eU-with-op-norm-U}.
  Suppose $\alpha \geq \pi/2$.
  Then $A$ contains a half-circle and thus has an endpoint $\theta \in `*{ \theta_\pm }$ whose real part is nonpositive: $\Re(\theta)\leq 0$.
  $\theta$ is an eigenphase of $U$, so $\theta-1$ is an eigenvalue of $U-\Ident$.
  Hence,
  \begin{align*}
    \frac{1}{2} \dianorm{ \mathcal{U} - \IdentProc{} }
    = \sin`*( \min`*{ \alpha, \frac\pi2 } )
    = 1
    \leq \abs`*{ \Re(\theta) - 1 }
    = \abs`*{ \Re(\theta - 1) }
    \leq \abs`*{\theta-1}
    \leq \opnorm{U-\Ident} \ .
  \end{align*}
  Now, suppose $\alpha < \pi/2$.
  Then
  \begin{align}
      \frac{1}{2} \dianorm{ \mathcal{U} - \IdentProc{} }
      &= \sin`*( \min`*{ \alpha, \frac\pi2 } )
      = \sin`*(\alpha)
      = \abs`*{ \ee^{i\chi_*} \sin`*(\alpha) } 
      \nonumber \\
      &= \frac{1}{2} \abs`*{ \ee^{i\chi_*} `*( \ee^{i\alpha} - \ee^{-i\alpha} ) }
      = \frac{1}{2} \abs`*{ \theta_+ - \theta_- }
      = \frac{1}{2} \abs`*{ (\theta_+ -1) - (\theta_- -1) }
      \nonumber \\
      &\leq \frac{1}{2} \abs`*{ \theta_+ -1 } + \frac{1}{2} \abs`*{ \theta_- -1 } 
      \nonumber \\
      &\leq \opnorm{U - \Ident} \ .
  \end{align}
  The second inequality follows because $U -\Ident$ has the eigenvalues $\theta_\pm - 1$.
\end{proof}

\section{Measures of complexity and POVM-effect complexity}
\label{appx:sec:complexity-of-measurement-effects}

We define the complexity (relative) entropy in terms of POVM-effect complexity.
This section's main purpose is to introduce general assumptions that define POVM-effect complexity.
Our setting's generality enables one to define the complexity (relative) entropy for a range of physical systems, including qubits, qudits, and continuous-variable systems.

\subsection{Complexity of superoperators and unitaries}

Using conservative assumptions, we define complexity measures for superoperators and unitaries.
We later use these measures to construct measures of POVM-effect complexity.

\begin{definition}[Superoperator-complexity measure]
  \label{defn:superoperator-complexity-general}
  Let $S$ denote a composition of $N$ quantum subsystems: $S = S_1 S_2 \ldots S_N$.
  For every subset $I\subset `{1, 2, \ldots, N}$, let $S_I$ denote the composition of the subsystems in $`{S_i}_{i\in I}$.
  A superoperator-complexity measure on $S$ is a family of functions $C_{S_I}$ that map completely positive, trace-nonincreasing maps to $\mathbb{R}_+ \cup `{ \infty }$, for all $I \subset `{ 1, 2, \ldots, N}$.
  The functions must have the following properties, for every $I \subset `{ 1, 2, \ldots, N}$:
  \begin{enumerate}[(\roman*)]
  
  \item\label{item:superoperator-complexity--identity-has-complexity-zero}
    The identity operation has zero complexity: $C_{S_I}(\IdentProc[S_I]{}) = 0$.
    
  \item\label{item:superoperator-complexity--sequential-composition}
    No sequential composition of operations has a complexity that exceeds the sum of the individual operations' complexities:
    for any operations $\mathcal{E}_1$ and $\mathcal{E}_2$ on $S_I$, $C_{S_I}(\mathcal{E}_2\mathcal{E}_1) \leq C_{S_I}(\mathcal{E}_1) + C_{S_I}(\mathcal{E}_2)$.
    
  \item\label{item:superoperator-complexity--parallel-composition}
    No parallel composition of operations has a complexity that exceeds the sum of the individual operations' complexities:
    for any $J$ disjoint from $I$, for any operation $\mathcal{E}_1$ on $S_I$, and for any operation $\mathcal{E}_2$ on $S_J$, $C_{S_I S_J}(\mathcal{E}_1\otimes\mathcal{E}_2) \leq C_{S_I}(\mathcal{E}_1) + C_{S_J}(\mathcal{E}_2)$.
  \end{enumerate}
\end{definition}
For every operation $\mathcal{E}$ on any $S_I$, we employ the shorthand notation $C_S(\mathcal{E}) \coloneqq C_{S_I}(\mathcal{E})$.
Accordingly, we denote a superoperator-complexity measure by its function $C_S$, for convenience.

The following is a natural way to construct a superoperator-complexity measure:
Choose a set $\mathcal{G}$ of elementary operations.
Define the complexity of an arbitrary operation $\mathcal{E}$ as the minimum number of operations in $\mathcal{G}$ that, under composition, effect $\mathcal{E}$.

\begin{definition}[Circuit-superoperator-complexity measure]
  \label{defn:circuit-complexity-measure}
  Let $S = S_1 S_2 \ldots S_N$ denote a quantum system.
  We employ the notation of \cref{defn:superoperator-complexity-general}.
  For every subset $I\subset `{1, 2, \ldots, N}$, let $\mathcal{G}_{S_I}$ denote an (arbitrary) set of completely positive, trace-nonincreasing maps such that $\IdentProc[S_I]{} \in \mathcal{G}_{S_I}$ and $\mathcal{G}_{S_I} \subset \mathcal{G}_{S_J}$ for all $J \supset I$.
  Denote the set $\mathcal{G}_S$ by $\mathcal{G}$.
  A \emph{circuit-superoperator-complexity measure} on $S$ associated with $\mathcal{G}$ is a family of functions $C_{\mathcal{G},S_I}$, for any $I \subset `{1, 2, \ldots, N}$, defined as
  \begin{align}
    C_{\mathcal{G},S_I}(\mathcal{E}) \coloneqq \min `*{ r\in \mathbb{N} :
    \mathcal{E}_r \cdots \mathcal{E}_2 \mathcal{E}_1 = \mathcal{E}\,,\;
    \mathcal{E}_i \in \mathcal{G}_{S_I} \: 
    \, \forall i \in `*{1, 2, \ldots, r}
    }\ .
  \end{align}
\end{definition}

We introduce three conventions.
First, we denote a circuit-superoperator-complexity measure by its function $C_{\mathcal{G},S}$.
Second, we designate the composition of a zero number of operations on $S_I$ as the identity operation $\IdentProc[S_I]{}$.
Third, we set $C_{S_I}(\mathcal{E}) = \infty$ whenever no finite sequence of operations from $\mathcal{G}_{S_I}$ implements an operation $\mathcal{E}$ exactly.

The following proposition confirms that every circuit-complexity measure (\cref{defn:circuit-complexity-measure}) is a superoperator-complexity measure (\cref{defn:superoperator-complexity-general}).

\begin{proposition}[Compatibility of circuit-superoperator complexity and general-superoperator complexity]
  \label{thm:check-circuit-complexity-is-also-superoperator-complexity}
  Define $S$ and the gate set $\mathcal{G}$ as in \cref{defn:circuit-complexity-measure}.
  Every circuit-superoperator-complexity measure on $S$ associated with $\mathcal{G}$ is a superoperator-complexity measure.
\end{proposition}

\begin{proof}[**thm:check-circuit-complexity-is-also-superoperator-complexity]
  We verify that every circuit-superoperator-complexity measure $C_{\mathcal{G},S}$ has the three properties in \cref{defn:superoperator-complexity-general}.
  Property~\ref{item:superoperator-complexity--identity-has-complexity-zero} holds because, by definition, $\IdentProc[S_I]{} \in \mathcal{G}_{S_I}$; and, by convention, $\IdentProc[S_I]{}$ is the sequential composition of zero operations.
  Property~\ref{item:superoperator-complexity--sequential-composition} automatically characterizes every sequential composition that involves an infinite-complexity operation.
  Now, we prove that property~\ref{item:superoperator-complexity--sequential-composition} characterizes finite-complexity operations $\mathcal{E}$ and $\mathcal{E}'$ on $S_I$.
  If the operations' complexities are $C_{S_I}(\mathcal{E}) = r$ and $C_{S_I}(\mathcal{E}') = r'$, we can decompose $\mathcal{E}$ and $\mathcal{E}'$ as sequences of $r$ and $r'$ operations in $\mathcal{G}_{S_I}$, respectively.
  Therefore, we can decompose $\mathcal{E'}\mathcal{E}$ as a sequence of $r + r'$ operations in $\mathcal{G}_{S_I}$.
  Hence $C_{S_I}(\mathcal{E'} \mathcal{E}) \leq r + r' = C_{S_I}(\mathcal{E}) + C_{S_I}(\mathcal{E}')$.
  Property~\ref{item:superoperator-complexity--parallel-composition} follows as a special case of Property~\ref{item:superoperator-complexity--sequential-composition}: one can express every operation $\mathcal{E}$ on $S_I$ as the operation $\mathcal{E} \otimes \IdentProc[S_J]{}$ on $S_I S_J$ and vice versa.
\end{proof}

Many complexity measures in the literature are defined only for unitary operations.
The following definition specializes \cref{defn:superoperator-complexity-general} to such situations.

\begin{definition}[Unitary-complexity measure]
  \label{defn:unitary-complexity}
  A \emph{unitary-complexity measure} is
  a superoperator-complexity measure that assigns finite values only to unitary operations.
  A unitary-complexity measure $C_S$ is called \emph{adjoint-invariant} if $C_S(\mathcal{U}) = C_S(\mathcal{U}^\dagger)$ for every unitary operation $\mathcal{U}$ on $S$.
\end{definition}
Let $U$ denote any unitary operator on $\Hs_S$, and let $\mathcal{U}(\cdot) \coloneqq U (\cdot) U^\dagger$.
We write $C_S(U)$ as a shorthand for $C_S( \mathcal{U} )$, for convenience.
Correspondingly, a \emph{unitary-circuit-complexity measure} is a unitary-complexity measure that is also a circuit-complexity measure.
Equivalently, a unitary-circuit-complexity measure is a circuit-complexity measure associated with a set of unitary gates.

\subsection{Quantum state complexity}

Here, we introduce state-complexity measures in general, for completeness and for use in later appendices.

\begin{definition}[Exact-state-complexity measure]
  \label{defn:exact-state-complexity}
  Let $S$ denote a quantum system.
  Let $C_S$ denote any superoperator-complexity measure; and $\rho_0$, any state of $S$.
  The \emph{exact-state-complexity measure} with respect to $C_S$ and $\rho_0$ is
  \begin{align}
    C_S^{\rho_0}(\rho)
    \coloneqq \inf `*{ C_S(\mathcal{E}) : \mathcal{E}(\rho_0)  = \rho } \ .
  \end{align}
\end{definition}
\noindent 
An exact-state-complexity measure $C_S^{\rho_0}$ is called an \emph{exact-pure-state-complexity measure} if $\rho_0$ is pure and if $C_S^{\rho_0}$ assigns finite values only to pure states.
$C_S^{\rho_0}$ is such a measure if, for instance, $\rho_0$ is pure and $C_S$ is a unitary-complexity measure.

\begin{definition}[Approximate-state-complexity measure]
  \label{defn:approx-state-complexity}
  Let $S$ denote a quantum system.
  Let $C_S$ denote any superoperator-complexity measure; and $\rho_0$, any state of $S$.
  Let $\delta \geq 0$.
  The \emph{$\delta$-approximate-state-complexity measure} with respect to $C_S$ and $\rho_0$ is
  \begin{align}
    C_S^{\rho_0,\,\delta}(\rho)
    \coloneqq \lim_{\zeta\to 0^+} \inf `*{ C_S(\mathcal{E}) \,:\;
    \frac12\onenorm{\mathcal{E}(\rho_0) - \rho } \leq \delta + \zeta }\ .
    \label{eq:defn-approx-state-complexity}
  \end{align}
\end{definition}

The stabilization $\lim_{\zeta\to 0^+}$ in~\eqref{eq:defn-approx-state-complexity} helps smooth $C_S^{\rho_0,\,\delta}$ at unusual discontinuities in superoperator-complexity measures.
For instance, consider any set $\mathcal{G}_0$ of two-qubit unitary gates that is universal for quantum computation.
Let $\mathcal{G}$ denote the set of $n$-qubit unitaries realizable as finite sequences of gates in $\mathcal{G}_0$.
$\mathcal{G}$ is closed under composition.
Let $C_{\mathcal{G}}$ denote the unitary-circuit-complexity measure associated with $\mathcal{G}$ (\cref{defn:circuit-complexity-measure}).
By construction, the identity $\Ident_2^{\otimes n}$ has zero complexity, and every unitary $U \in \mathcal{G}$ has unit complexity: $C_{\mathcal{G}}(U) = 1$.
Any other unitary has infinite complexity, since it can only be approximated by a sequence of gates in $\mathcal{G}_0$---and thus by a unitary in $\mathcal{G}$~\cite{BookNielsenChuang2000}.
Let $\ket{\psi_0}$ denote any pure $n$-qubit state.
The stabilization ensures that $C_S^{\psi_0,\,0}(\psi) = 1$ for all pure states $\ket{\psi} \neq \ket{\psi_0}$.
Without the stabilization, every state $\ket\psi$ inaccessible from $\ket{\psi_0}$ by any unitary in $\mathcal{G}$ would have an infinite complexity: $C_S^{\psi_0,\,0}(\psi) = \infty$.

As an immediate consequence of~\eqref{eq:defn-approx-state-complexity}, $C_S^{\rho_0,\delta}$ is right-continuous in $\delta$:
\begin{align}
  \lim_{\zeta\to 0^+} C_S^{\rho_0,\delta + \zeta}(\rho) = C_S^{\rho_0,\delta}(\rho)\ .
  \label{eq:approx-state-complexity-right-continuity-in-delta}
\end{align}

An approximate-state-complexity measure $C_S^{\rho_0,\,\delta}$ is called an \emph{approximate-pure-state-complexity measure} if $\rho_0$ is pure and if $C_S^{\rho_0,\,0}$ assigns finite values only to pure states.
$C_S^{\rho_0,\,\delta}$ is such a measure if, for instance, $\rho_0$ is pure and $C_S$ is a unitary-complexity measure.

The following proposition provides an alternative expression for~\eqref{eq:defn-approx-state-complexity}.
From this expression, we derive a condition under which the stabilization in~\eqref{eq:defn-approx-state-complexity} becomes inconsequential.

\begin{proposition}[Alternative expression for approximate-state-complexity measure]
  \label{thm:stab-approx-state-complexity-measure}
  Let $S$ denote a quantum system.
  Let $C_S$ denote any superoperator-complexity measure; and $\rho_0$, any state of $S$.
  Let $\delta \geq 0$.
  It holds that
  \begin{align}
    C_S^{\rho_0,\delta}(\rho)
    &= \inf `*{ r \,:\;
      \dist_1(\rho, \mathfrak{G}_r[\rho_0]) \leq \delta
      } \ ,
    \label{eq:stab-approx-state-complexity-measure-expression-with-set-dist1}
  \end{align}
  where
  \begin{align}
      \mathfrak{G}_r[\rho_0]
    \coloneqq `*{ \mathcal{E}(\rho_0) \,:\; C_S(\mathcal{E}) \leq r } \
    \; \; \text{and} \; \;
      \dist_1(\rho, T) &\coloneqq \inf_{\rho'\in T} `*{ \frac12\onenorm[\big]{ \rho - \rho' } } \ ,
      \label{eq:stab-approx-state-complexity-measure-defn-dist1}
  \end{align}
  and wherein $T$ denotes any set of states.
\end{proposition}

We now show that, under simple conditions, we can forgo the stabilization in the definition of the approximate-state-complexity measure.

\begin{corollary}[No stabilization of the approximate-state-complexity measure required for compact gate sets]
  \label{thm:stab-approx-state-complexity-measure-compact-gate-set}
  Consider the setting in \cref{thm:stab-approx-state-complexity-measure}.
  Suppose that, for all $r\geq 0$, the sets $\mathfrak{G}_r[\rho_0]$ are compact.
  Then
  \begin{align}
    C_S^{\rho_0,\,\delta}(\rho)
    = 
    \inf
    `*{ C_S(\mathcal{E}) \,:\;
    \frac12\onenorm{\mathcal{E}(\rho_0) - \rho } \leq \delta }\ .
    \label{eq:approx-state-complexity-nostab-Gr-compact}
  \end{align}
  Furthermore, the sets $\mathfrak{G}_r[\rho_0]$ are compact if $C_S$ is a circuit-complexity measure (\cref{defn:circuit-complexity-measure}) associated with a compact gate set $\mathcal{G}_S$.
\end{corollary}

\begin{proof}[*thm:stab-approx-state-complexity-measure]
  Let $\rho$ denote any state of $S$.
  We show that the expressions in~\eqref{eq:defn-approx-state-complexity} and~\eqref{eq:stab-approx-state-complexity-measure-expression-with-set-dist1} equal
  \begin{align}
    r^*
    &\coloneqq \inf`*{
    r\ :\ 
    \forall\ \zeta > 0\ ,\ 
    \exists\ \mathcal{E}\ \text{ such that }\ C_S(\mathcal{E}) \leq r\ \text{ and }\ 
    \frac12\onenorm{\mathcal{E}(\rho_0) - \rho} \leq \delta + \zeta
    } \ .
    \label{eq:qiphojlkwop2iewojbnlk}
  \end{align}
  
  The infimum in~\eqref{eq:stab-approx-state-complexity-measure-expression-with-set-dist1} equals $r^*$, since
  \begin{align}
    &
    \inf`*{ r \ : \ \dist_1`\big(\rho, \mathfrak{G}_r[\rho_0]) \leq \delta }
      \nonumber\\
    &= 
    \inf`*{ r \,:\ \forall\ \zeta>0\;,\  \exists\ \rho'\in\mathfrak{G}_r[\rho_0]\; \text{ such that }\ 
    \frac12\onenorm{\rho - \rho'} \leq \delta + \zeta }
    = r^* \ .
  \end{align}

  Now, consider $C_S^{\rho_0,\delta}(\rho)$, as defined in~\eqref{eq:defn-approx-state-complexity}.
  We show that $C_S^{\rho_0,\delta}(\rho) = r^*$ by showing first that $r^* \leq C_S^{\rho_0,\delta}(\rho)$ and then that $C_S^{\rho_0,\delta}(\rho) \leq r^*$.
  Consider any $\zeta,\zeta' > 0$.
  By~\eqref{eq:defn-approx-state-complexity}, there exists an $\mathcal{E}$ such that $C_S(\mathcal{E}) \leq C_S^{\rho_0,\delta}(\rho) + \zeta'$ and $\frac12\onenorm{ \mathcal{E}(\rho_0) - \rho } \leq \delta + \zeta$.
  Hence, $r = C_S^{\rho_0,\delta}(\rho) + \zeta'$ is a candidate for the optimization in~\eqref{eq:qiphojlkwop2iewojbnlk}, so
  \begin{align}
    r^* \leq C_S^{\rho_0,\delta}(\rho) + \zeta' \ .
  \end{align}
  Thus, $r^* \leq C_S^{\rho_0,\delta}(\rho)$, since $\zeta' > 0$ is arbitrary.
  By the infimum in~\eqref{eq:qiphojlkwop2iewojbnlk}, there exists an $\mathcal{E}$ and an $r\leq r^* + \zeta'$ such that $C_S(\mathcal{E})\leq r$ and $\frac12\onenorm{\mathcal{E}(\rho_0) - \rho} \leq \delta+\zeta$.
  Hence, $r$ is lower-bounded as
  \begin{align}
    \inf`*{ C_S(\mathcal{E})\ :\ \frac12\onenorm{\mathcal{E}(\rho_0) - \rho} \leq \delta + \zeta }
    \leq r \leq r^* + \zeta' \ .
    \label{eq:random-infimum-in-proof}
  \end{align}
  Thus, the infimum in~\eqref{eq:random-infimum-in-proof} lower-bounds $r^*$, since $\zeta' > 0$ is arbitrary.
  Consequently,
  \begin{align}
    C_S^{\rho_0,\delta}(\rho)
    = \lim_{\zeta\to 0^+}
    \inf`*{ C_S(\mathcal{E})\ :\ \frac12\onenorm{\mathcal{E}(\rho_0) - \rho} \leq \delta + \zeta }
    \leq r^* \ .
  \end{align}
\end{proof}

\begin{proof}[*thm:stab-approx-state-complexity-measure-compact-gate-set]
  Suppose that $\mathfrak{G}_r[\rho_0]$ is compact for all $r\geq 0$.
  The infimum defining $\dist_1`(\rho, \mathfrak{G}_r[\rho_0])$, in~\eqref{eq:stab-approx-state-complexity-measure-defn-dist1}, is achieved by some state in $\mathfrak{G}_r[\rho_0]$.
  Thus,
  \begin{align}
      \dist_1`*(\rho, \mathfrak{G}_r[\rho_0])
      = \min`*{ \frac12\onenorm{\mathcal{E}(\rho_0) - \rho}\ :\ C_S(\mathcal{E})\leq r } \ .
      \label{eq:final-expression-for-dist-set}
  \end{align}
  Substituting~\eqref{eq:final-expression-for-dist-set} into~\eqref{eq:stab-approx-state-complexity-measure-expression-with-set-dist1} yields~\eqref{eq:approx-state-complexity-nostab-Gr-compact}.

  Now, suppose that $C_S$ is a circuit-complexity measure (\cref{defn:circuit-complexity-measure}) associated with a gate set $\mathcal{G}_S$.
  Assume that $\mathcal{G}_{S_I}$ is compact for all $I \subset `{1, 2, \ldots, N}$.
  Consider the function $F$ that maps every sequence $(\mathcal{E}_1, \ldots ,\mathcal{E}_r)$ of gates in $\mathcal{G}_S$ to an action on $\rho_0$, as follows:
  \begin{align}
      F`*( \mathcal{E}_1, \mathcal{E}_2, \ldots, \mathcal{E}_r ) \coloneqq `*( \mathcal{E}_r \cdots \mathcal{E}_2 \mathcal{E}_1 ) `*( \rho_0 ) \ .
  \end{align}
  By definition, $F$ is a map from $(\mathcal{G}_S)^{\times r}$ to $\mathfrak{G}_r[\rho_0]$.
  $F$ is continuous, since it is a polynomial map.
  $(\mathcal{G}_S)^{\times r}$ is compact, since the finite product of compact sets is compact.
  Therefore, $\mathfrak{G}_r[\rho_0]$ is compact, since continuous functions preserve compactness.
\end{proof}

\subsection{POVM-effect complexity}

A central ingredient in the complexity relative entropy's definition is a family $\{\Mr[r]\}$ of sets of POVM effects.
$\{ \Mr[r] \}$ is a formal parameter that can accommodate any notion of POVM-effect complexity.
Each set $\Mr[r]$ contains precisely the POVM effects of complexities $\leq r$.
An example of such a family of sets is defined via~\eqref{eq:setting-defn-Mr}.

Here, we present conservative assumptions about every family $\{\Mr[r]\}$ to ensure that it has properties necessary for our construction of the complexity (relative) entropy.

\begin{definition}[POVM-effect-complexity sets]
  \label{defn:Mr-sets-general}
  Let $S$ denote a composition of $N$ quantum subsystems: $S = S_1 S_2 \ldots S_N$.
  For every subset $I\subset `{1, 2, \ldots, N}$, let $S_I$ denote the composition of the subsystems in $`{S_i}_{i\in I}$.
  A \emph{family of POVM-effect-complexity sets} on $\Hs_S$ is a family $`{ \Mr[r][S_I] }_{r\in\mathbb{R}_+}$ of sets of POVM effects, for each $S_I$.
  Each $\Mr[r][S_I]$ is a \emph{POVM-effect-complexity set}.
  The family must satisfy the following axioms, for every $I\subset `{1, 2, \ldots, N}$:
  \begin{enumerate}[label={(\roman*)}]
    \item\label{item:POVM-effect-complexity--identity} 
    The identity operator has zero complexity: $\Ident_{S_I} \in \Mr[0][S_I]$.
    \item\label{item:POVM-effect-complexity--monotonous}
    For all $r'\geq r\geq 0$, $\Mr[r][S_I] \subset \Mr[r'][S_I]$.
    \item\label{item:POVM-effect-complexity--subsystems}
    For all $r,r'\geq 0$ and for all $J$ disjoint from $I$, $\Mr[r][S_I]\otimes\Mr[r'][S_J] \subset \Mr[r+r'][S_I S_J]$.
  \end{enumerate}
\end{definition}
\noindent
For convenience, we write $`{ \Mr[r][S] }$ as shorthand for a family $`{ \Mr[r][S_I] }_{r, I}$.

\subsection{POVM-effect-complexity sets from simple POVM effects and a superoperator-complexity measure}

A natural way to construct POVM-effect-complexity sets is as follows:
Designate a set of ``simple,'' zero-complexity POVM effects.
Define the complexity of an effect $Q$ as the minimum complexity of any superoperator that maps a simple effect to $Q$.
In the main text, we use this procedure to construct the sets $\Mr[r]$ in~\eqref{eq:setting-defn-Mr} from single-qubit projectors and a circuit-complexity measure.

The conditions in \cref{defn:Mr-sets-general} imply that zero-complexity effects must satisfy two conditions:
(i) the identity operator is a zero-complexity effect.
(ii) The parallel composition of two zero-complexity effects is a zero-complexity effect.
Any set of ``simple'' POVM effects must satisfy these criteria, which we formalize in the following definition.

\begin{definition}[Simple POVM effects]
  \label{defn:Psimple-general}
  Let $S$ denote a composition of $N$ quantum subsystems: $S = S_1 S_2 \ldots S_N$.
  For every subset $I\subset `{1, 2, \ldots, N}$, let $S_I$ denote the composition of the subsystems in $`{S_i}_{i\in I}$.
  Let $\{\Psimple[S_I]\}$ denote a family of sets of POVM effects.
  The $\Psimple[S_I]$ are sets of \textit{simple POVM effects} if, for all $I\subset `{1, 2, \ldots, N}$, the following conditions hold:
  \begin{enumerate}[(\roman*)]
    \item $\Psimple[S_I]$ contains the identity operator: $\Ident_{S_I} \in \Psimple[S_I]$.
    \item $\Psimple[S_I]\otimes \Psimple[S_J] \subset \Psimple[S_I S_J]$ for all $J$ disjoint from $I$.
  \end{enumerate}
\end{definition}
\noindent
For convenience, we denote a family $`{ \Psimple[S_I] }$ by its set $\Psimple[S]$.

We now construct POVM-effect-complexity sets from simple POVM effects and a superoperator-complexity measure.
\begin{definition}[POVM-effect-complexity sets from simple effects and a superoperator-complexity measure]
  \label{defn:Mr-general-from-Psimple-and-superop-cplx-measure}
  Let $S = S_1 S_2 \ldots S_N$ denote a quantum system as in \cref{defn:Psimple-general}.
  Let $\Psimple[S]$ denote any set of simple POVM effects (\cref{defn:Psimple-general}); and $C_S$, any superoperator-complexity measure (\cref{defn:superoperator-complexity-general}).
  For all $r \geq 0$ and $I\subset `{1, 2, \ldots, N}$, we define
  \begin{align}
    \Mr[r][S_I]`*(\Psimple[S_I], C_{S_I}) \coloneqq `*{
    \mathcal{E}^\dagger( P ) \,:\;
    P \in \Psimple[S_I] \,,\;
    C_{S_I}(\mathcal{E}) \leq r
    }\ .
    \label{eq:Mr-general-from-Psimple-and-superop-cplx-measure}
  \end{align}
\end{definition}

\begin{proposition}
  \label{thm:Mr-Psimple-C-define-POVM-effect-cplx-measure}
  The sets $\Mr[r][S_I]`*(\Psimple[S_I], C_{S_I})$ constructed in~\eqref{eq:Mr-general-from-Psimple-and-superop-cplx-measure} define a family of POVM-effect-complexity sets (\cref{defn:Mr-sets-general}).
\end{proposition}

\begin{proof}[**thm:Mr-Psimple-C-define-POVM-effect-cplx-measure]
  We write $\Mr[r][S_I] = \Mr[r][S_I]`*(\Psimple[S_I], C_{S_I})$ for short.
  We verify that the sets $\Mr[r][S_I]$ have the three properties in \cref{defn:Mr-sets-general}.
  Property~\ref{item:POVM-effect-complexity--identity} holds because $\Ident_{S_I} = \IdentProc[S_I]{}^\dagger`*(\Ident_{S_I})$, $\Ident_{S_I} \in \Psimple[S_I]$ (\cref{defn:Psimple-general}), and $C_{S_I}`*(\IdentProc[S_I]{}) = 0$ (\cref{defn:superoperator-complexity-general}).
  Property~\ref{item:POVM-effect-complexity--monotonous} holds because, for all $r' \geq r \geq 0$ and for all operations $\mathcal{E}$ on $S_I$, $C_{S_I}(\mathcal{E}) \leq r$ implies $C_{S_I}(\mathcal{E}) \leq r'$.
  To show that Property~\ref{item:POVM-effect-complexity--subsystems} holds, we consider any POVM effects $Q_1 \in \Mr[r][S_I]$ and $Q_2 \in \Mr[r'][S_J]$, with $I$ and $J$ disjoint.
  $Q_1 = \mathcal{E}_1^\dagger (P_1)$ for some effect $P_1 \in \Psimple[S_I]$ and for some operation $\mathcal{E}_1$ satisfying $C_{S_I}(\mathcal{E}_1) \leq r$.
  Similarly, $Q_2 = \mathcal{E}_2^\dagger (P_2)$ for some effect $P_2 \in \Psimple[S_J]$ and for some operation $\mathcal{E}_2$ satisfying $C_{S_J}(\mathcal{E}_2) \leq r'$.
  By \cref{defn:Psimple-general}, $P_1\otimes P_2 \in \Psimple[S_I S_J]$.
  By \cref{defn:superoperator-complexity-general}, $C_{S_I S_J}(\mathcal{E}_1\otimes\mathcal{E}_2) \leq C_{S_I}(\mathcal{E}_1) + C_{S_J}(\mathcal{E}_2) \leq r + r'$.
  Thus, $Q_1\otimes Q_2 = `\big(\mathcal{E}_1\otimes \mathcal{E}_2)^\dagger`\big(P_1\otimes P_2) \in \Mr[r+r'][S_I S_J]$.
\end{proof}

By~\eqref{eq:Mr-general-from-Psimple-and-superop-cplx-measure}, every simple POVM effect is a zero-complexity effect, since the identity process has zero complexity.
In general, the converse is false: not every zero-complexity effect is a simple POVM effect.
Indeed, a nontrivial zero-complexity operation applied to a simple effect may yield a zero-complexity effect that is not a simple effect.
Let $\Psimple$ denote a set of simple POVM effects (\cref{defn:Psimple-general}); and $C$, a superoperator-complexity measure (\cref{defn:superoperator-complexity-general}).
Suppose that there exists a zero-complexity operation $\mathcal{E} \neq \IdentProc[]{}$: $C(\mathcal{E}) = 0$.
Suppose, further, that $\mathcal{E}^\dagger(P) \not\in \Psimple$ for some effect $P \in \Psimple$.
Then $\Mr[r=0](\Psimple,C) \supsetneq \Psimple$, since $\mathcal{E}^\dagger(P) \in \Mr[r=0](\Psimple,C)$, by~\eqref{eq:Mr-general-from-Psimple-and-superop-cplx-measure}.
If $\IdentProc[]{}$ is the unique zero-complexity operation---or, more generally, if $\Psimple$ is closed under (the adjoint of) every zero-complexity operation---then $\Mr[r=0](\Psimple,C) = \Psimple$.
In particular, $\IdentProc[]{}$ is the unique zero-complexity operation whenever $C$ is a circuit-complexity measure (\cref{defn:circuit-complexity-measure}).

\subsection{POVM-effect-complexity sets from a POVM-effect-complexity measure}

An alternative, equivalent formulation of a POVM-effect-complexity set (\cref{defn:Mr-sets-general}) arises from a complexity measure on individual POVM effects.
This formulation has the advantage of using a primitive formally similar to a superoperator-complexity measure (\cref{defn:superoperator-complexity-general}).
Consequently, some readers may find POVM-effect-complexity sets more intuitive in this formulation.

\begin{definition}[POVM-effect-complexity measure]
  \label{defn:POVM-complexity-general}
  Let $S$ denote a composition of $N$ quantum subsystems: $S = S_1 S_2 \ldots S_N$.
  For every subset $I\subset `{1, 2, \ldots, N}$, let $S_I$ denote the composition of the subsystems in $`{S_i}_{i\in I}$.
  A \emph{POVM-effect-complexity measure} on $S$ is a family of functions $F_{S_I}$ that map POVM effects on $S_I$ to $\mathbb{R}_+ \cup `{ \infty }$, for all $I \subset `{ 1, 2, \ldots, N}$.
  The functions have the following properties, for every $I \subset `{ 1, 2, \ldots, N}$:
  \begin{enumerate}[(\roman*)]
  \item\label{item:POVM-complexity--identity-has-complexity-zero}%
    The identity effect has zero complexity: $F_{S_I}(\Ident_{S_I}) = 0$.
  \item\label{item:POVM-complexity--parallel-composition} %
    No parallel composition of POVM effects has a complexity greater than the sum of the individual effects' complexities:
    $F_{S_I S_J}(Q_1\otimes Q_2) \leq F_{S_I}(Q_1) + F_{S_J}(Q_2)$ for all $J$ disjoint from $I$.
  \end{enumerate}
\end{definition}

For convenience, we denote a POVM-effect-complexity measure by its function $F_S$.
Every POVM-effect-complexity measure induces a family of POVM-effect-complexity sets (\cref{defn:Mr-sets-general}), as per the following definition.

\begin{definition}[POVM-effect-complexity sets
from a POVM-effect-complexity measure]
  \label{defn:Mr-sets-from-POVM-complexity-measure}
  Let $S = S_1 S_2 \ldots S_N$ denote a quantum system; and $F_S$, a POVM-effect-complexity measure, as in \cref{defn:POVM-complexity-general}.
  A \emph{family of POVM-effect-complexity sets} on $\Hs_S$, associated with $F_S$, is a family $`{ \Mr[r][S_I]`*(F_{S_I}) }_{r\in\mathbb{R}_+ }$ of sets of POVM effects.
  Each $\Mr[r][S_I]`*(F_{S_I})$ denotes the set of POVM effects on $S_I$ of complexities $\leq r$:
  \begin{align}
      \Mr[r][S_I]`*(F_{S_I}) \coloneqq `*{Q \: : \: F_{S_I}(Q) \leq r } \ .
      \label{eq:POVM-effect-complexity-sets-induced-from-measure}
  \end{align}
\end{definition}
\noindent
One can straightforwardly check that the sets $\Mr[r][S_I]`*(F_{S_I})$ in~\eqref{eq:POVM-effect-complexity-sets-induced-from-measure} satisfy the conditions of \cref{defn:Mr-sets-general}.

We now construct a POVM-effect-complexity measure from simple POVM effects and a superoperator-complexity measure, providing an alternative way to define the sets $\Mr[r][S_I]`*(\Psimple[S_I],C_{S_I})$ in~\eqref{eq:Mr-general-from-Psimple-and-superop-cplx-measure}.
Let $S$ denote a quantum system; and $F_S$, a POVM-effect-complexity measure, as in \cref{defn:POVM-complexity-general}.
Let $\Psimple[S]$ denote a set of simple POVM effects (\cref{defn:Psimple-general}).
The following prescription, for all subsets $I\subset `{1, 2, \ldots, N}$, defines a POVM-effect-complexity measure:
\begin{align}
    F_{S_I}(Q) \coloneqq \inf `*{ C_{S_I}(\mathcal{E}) \, : \, P \in \Psimple[S_I], \, \mathcal{E}^\dagger( P ) = Q } \ .
    \label{eq:POVM-effect-complexity-from-simple-effects-and-superoperator-complexity}
\end{align}

\begin{proposition}
  \label{thm:FS-Psimple-C-define-POVM-effect-cplx-measure}
  The sets $F_{S_I}$ constructed in~\eqref{eq:POVM-effect-complexity-from-simple-effects-and-superoperator-complexity} define a POVM-effect-complexity measure.
\end{proposition}
\begin{proof}[**thm:FS-Psimple-C-define-POVM-effect-cplx-measure]
  We verify that the sets $F_{S_I}$ have both the properties in \cref{defn:POVM-complexity-general}.
  Property~\ref{item:POVM-complexity--identity-has-complexity-zero} holds because $C_{S_I}`*(\IdentProc[S_I]{}) = 0$ (\cref{defn:superoperator-complexity-general}), $\Ident_{S_I} \in \Psimple[S_I]$ (\cref{defn:Psimple-general}), and $\IdentProc[S_I]{}^\dagger`*(\Ident_{S_I}) = \Ident_{S_I}$.
  To show that Property~\ref{item:POVM-complexity--parallel-composition} holds, we consider any POVM effects $Q_1$ on $S_I$ and $Q_2$ on $S_J$, with $I$ and $J$ disjoint.
  Let $P_1 \in \Psimple[S_I]$, and let $\mathcal{E}_1$ denote an operation on $S_I$ such that $\mathcal{E}_1^\dagger(P_1) = Q_1$.
  Similarly, let $P_2 \in \Psimple[S_J]$, and let $\mathcal{E}_2$ denote an operation on $S_J$ such that $\mathcal{E}_2^\dagger(P_2) = Q_2$. 
  $\mathcal{E}_1 \otimes \mathcal{E}_2$ is a candidate for the $F_{S_I S_J}(Q_1 \otimes Q_2)$ optimization in~\eqref{eq:POVM-effect-complexity-from-simple-effects-and-superoperator-complexity}, since $P_1 \otimes P_2 \in \Psimple[S_I S_J]$ (\cref{defn:Psimple-general}) and $(\mathcal{E}_1 \otimes \mathcal{E}_2)^\dagger(P_1 \otimes P_2) = Q_1 \otimes Q_2$.
  Hence, $F_{S_I S_J}(Q_1 \otimes Q_2) \leq C_{S_I S_J}`*(\mathcal{E}_1 \otimes \mathcal{E}_2)$.
  Furthermore, $C_{S_I S_J}(\mathcal{E}_1 \otimes \mathcal{E}_2) \leq C_{S_I}`*(\mathcal{E}_1) + C_{S_J}`*(\mathcal{E}_2)$ (\cref{defn:superoperator-complexity-general}), so $F_{S_I S_J}(Q_1 \otimes Q_2) \leq C_{S_I}`*(\mathcal{E}_1) + C_{S_J}`*(\mathcal{E}_2)$.
  We obtain $F_{S_I S_J}(Q_1 \otimes Q_2) \leq F_{S_I}(Q_1) + F_{S_J}(Q_2)$ by taking the infimum of the last inequality with respect to all candidates $\mathcal{E}_1$ and $\mathcal{E}_2$ for the $F_{S_I}(Q_1)$ and $F_{S_J}(Q_2)$ optimizations in~\eqref{eq:POVM-effect-complexity-from-simple-effects-and-superoperator-complexity}, respectively.
\end{proof}

A POVM-effect-complexity measure constructed as in~\eqref{eq:POVM-effect-complexity-from-simple-effects-and-superoperator-complexity}, with a circuit-superoperator-complexity measure (\cref{defn:circuit-complexity-measure}), is called a \textit{circuit-POVM-effect-complexity measure}.
The POVM-effect-complexity measure in~\eqref{eq:POVM-effect-complexity-from-simple-effects-and-superoperator-complexity} induces, via~\eqref{eq:POVM-effect-complexity-sets-induced-from-measure}, the sets $\Mr[r][S_I] `*(\Psimple[S_I], C_{S_I})$ in~\eqref{eq:Mr-general-from-Psimple-and-superop-cplx-measure}.

\subsection{Examples of POVM-effect-complexity sets}

Any set of simple POVM effects (\cref{defn:Psimple-general}) and any unitary-complexity measure (\cref{defn:unitary-complexity}) determine a family of POVM-effect-complexity sets, as per~\eqref{eq:Mr-general-from-Psimple-and-superop-cplx-measure}.
In this section, we provide natural examples of simple POVM effects and unitary-complexity measures on $n$ qubits.

\subsubsection*{Examples of simple POVM effects}

Among the POVM effects that one can directly implement in a laboratory, local effects are often the most feasible.
The following sets of simple POVM effects contain only tensor products.
\begin{itemize}
    \item The set of tensor products of the single-qubit projectors $\proj{0}$ and $\Ident_2$.
    We use this choice throughout the main text [see~\eqref{eq:setting-defn-Mrzero}], as is natural in the context of erasure.
    
    \item The set of tensor products of single-qubit projectors.
    This choice is natural if one can easily perform any single-qubit gate before implementing a projector onto a single-qubit state, e.g., $\ket{0}$.
    
    \item The set of tensor-product POVM effects. 
    If performing nonlocal operations is hard, local operations are a natural choice for simple effects.
    Tensor-product POVM effects generalize tensor-product projectors.
\end{itemize}

Importantly, there exist POVM effects that one cannot implement by first applying a unitary and then implementing a tensor-product POVM effect.
For instance, the two-qubit effect $\frac{1}{2} `*(\proj{00} + \proj{01} + \proj{10})$ is a rank-3 operator, unlike any tensor-product two-qubit effect: operator rank is invariant under unitary transformation, so the given effect is not unitarily equivalent to any tensor-product effect.

\subsubsection*{Examples of unitary-complexity measures}
\label{sec:examples-unitary-complexity measures}

A natural unitary-complexity measure on $n$ qubits is the circuit complexity associated with a two-qubit gate set $\mathcal{G}$ (\cref{defn:circuit-complexity-measure}).
Depending on one's experimental capabilities, one can consider circuits with all-to-all connectivity (any gate in $\mathcal{G}$ can act on any two qubits); geometric locality (any gate in $\mathcal{G}$ can act on only two neighboring qubits); or intermediate connectivity (some gates in $\mathcal{G}$ can act on any two qubits, and other gates can act on only two neighboring qubits).
Convenient choices for $\mathcal{G}$ include the following:
\begin{itemize}
    \item The set of two-qubit gates: $\mathcal{G} = \SU(4)$.
    This choice is widespread in studies of random circuits~\cite{Brandao2016CMP_local,Dalzell2022PRXQ_random}.
    For instance, this choice has elucidated the linear growth, with time, of exact circuit complexity under random circuits \cite{Haferkamp_21_Linear, Li_22_ShortProofsOfLinearGrowth}.
    
    \item The following Clifford gates: Hadamard, $S$, and \textsc{CNOT} gates.
    This set generates quantum circuits that can be efficiently simulated classically~\cite{Gottesman1998_heisenbergrep}.
    The set is not universal, and almost all unitaries have infinite complexities with respect to the set.
    
    \item Any finite, universal set of two-qubit gates, such as the Clifford + $T$ gate set.
    This set can approximate every unitary arbitrarily well with a finite-complexity unitary~\cite{NielsenC10}.
    In this setting, the circuit complexity grows algebraically under random quantum circuits with gates drawn from a finite gate set~\cite{Brandao2021PRXQ_models,brandao2016local,haferkamp2022random}.
\end{itemize}

Our definition of unitary complexity accommodates other notions of complexity, such as the Nielsen complexity~\cite{Nielsen_05_Geometric,Nielsen_06_Quantum,Nielsen_06_Optimal,Dowling_06_Geometry}.
Another example is the minimum $T$-count of any Clifford + $T$ circuit that approximates a given unitary.
(In the latter example, different gates contribute different amounts of complexity.
Each $T$ gate contributes one unit of complexity; and each Clifford gate, zero units.

\section{Hypothesis-testing (relative) entropy}
\label{appx:hypothesis-testing-entropy}

In a hypothesis test, one receives either $\rho$ or $\sigma$, and guesses which state one received.
In the most general strategy to distinguish $\rho$ and $\sigma$, one performs a two-outcome measurement $\{Q,\Ident-Q\}$.
One guesses $\rho$ if $Q$ obtains and guesses $\sigma$ otherwise.
Suppose that one must, if the state is $\rho$, guess $\rho$ with a probability $\geq \eta \in (0, 1]$.
The minimum probability of wrongly guessing $\rho$ defines the \textit{hypothesis-testing relative entropy}.

\begin{definition}[Hypothesis-testing relative entropy]
  \label{def:defn-DHyp}
  Let $\rho$ denote any subnormalized state, and $\Gamma$ any positive-semidefinite operator, that act on the same Hilbert space.
  Let $\eta \in (0, \tr(\rho) ]$.
  The \textit{hypothesis-testing relative entropy} is defined as \cite{Datta_2013_HypothesisTesting,Wang2012PRL_oneshot,Dupuis2013_DH}
  \begin{align}
    \DHyp[\eta]{\rho}{\Gamma}
    \coloneqq -\log \; `*(
    \min_{\substack{
    0 \leq Q \leq \Ident\\
    \tr`(Q\rho) \geq \eta
    }}
    `*{ \frac{\tr`*(Q\Gamma)}{\eta} }
    )
    \ .
    \label{eq:defn-DHyp}
  \end{align}
\end{definition}
In a hypothesis test, the least probability of erroneously rejecting $\sigma$, using a strategy that successfully accepts $\rho$ with a probability $\geq \eta$, is $\eta \exp\bm{(} -\DHyp[\eta]{\rho}{\sigma} \bm{)}$.
Every optimal POVM effect $Q$ satisfies $\tr`(Q\rho) = \eta$.
If any such $Q$ violated this equation, one could achieve a better objective value with $Q' \coloneqq \eta Q / \tr`(Q\rho) \leq Q$.
Consequently, if $\Gamma = \rho$, then $\DHyp[\eta]{\rho}{\Gamma} = 0$. In this case, every $Q$ satisfying $\tr`(Q\rho) = \eta$ is optimal.

There always exists an optimal effect for the $\DHyp[\eta]{}{}$ optimization, since the domain of optimization is compact.
This fact justifies the use of a minimum, instead of an infimum, in~\eqref{eq:defn-DHyp}.

\begin{definition}[Hypothesis-testing entropy]
  \label{def:defn-HHyp}
  Let $\rho$ denote any subnormalized state.
  Let $\eta \in (0, \tr(\rho) ]$.
  The \textit{hypothesis-testing entropy} is defined as \cite{Dupuis2013_DH}
  \begin{align}
    \HHyp[\eta]{\rho}
    \coloneqq - \DHyp[\eta]{\rho}{\Ident}
    = \log \; `*(
    \min_{\substack{
    0 \leq Q \leq \Ident\\
    \tr`(Q\rho) \geq \eta} }
    `*{ \frac{\tr`*(Q)}{\eta} } ) \ .
    \label{eq:defn-HHyp}
 \end{align}
\end{definition}

The hypothesis-testing entropy and relative entropy are related to smooth entropies, including the min- and max-relative entropies~\cite{Datta2009IEEE_minmax,Dupuis2013_DH}.
The smooth entropies quantify the optimal efficiencies of operational tasks performed in the absence of complexity restrictions.
Such tasks feature finitely many copies of a quantum state or random variable, as well as finite failure probabilities.

One can express the hypothesis-testing relative entropy~\eqref{eq:defn-DHyp} as a semidefinite program.
The dual problem takes the form~\cite{Dupuis2013_DH}
\begin{equation}
  \DHyp[\eta]{\rho}{\Gamma}
  = - \log `*(
  \max_{\substack{ X\geq 0, \, \mu\geq 0\\ \mu\rho \leq \Gamma + X }}
  `*{ \mu - \frac{\tr`*(X)}{\eta} }
  ) \ .
  \label{eq:DHyp-alternative-expressions--dual-max-mu-X}\\
\end{equation}

We now reformulate the hypothesis-testing relative entropy.
The rewriting reveals the entropy's formal similarity to the complexity relative
entropy~\eqref{eq:setting-defn-DHypr}.

\begin{proposition}[Alternative expression for hypothesis-testing relative entropy]
  \label{thm:DHyp-alternative-expressions}
  Let $\rho$ denote any subnormalized state; and $\Gamma$, any positive-semidefinite operator.
  Let $\eta \in (0, \tr(\rho) ]$.
  The hypothesis-testing relative entropy obeys
  \begin{align}
      \DHyp[\eta]{\rho}{\Gamma}
      &= -\log`*(
        \inf_{ \substack{0\leq Q\leq \Ident\\ \tr`(Q\rho)\geq\eta }}
          `*{ \frac{\tr`*(Q\Gamma)}{\tr`*(Q\rho)} }
      )\ .
      \label{eq:DHyp-alternative-expressions--ratio-trQGamma-trQrho}
  \end{align}
\end{proposition}
\noindent
Without loss of generality, we can assume that every candidate effect $Q$ for the optimization satisfies $\opnorm{Q}=1$.
If a candidate $Q$ did not, we could replace $Q$ with a candidate $Q' \coloneqq Q/\opnorm{Q} \geq Q$ that achieves the same objective value as $Q$.
\begin{proof}[**thm:DHyp-alternative-expressions]
  Let $Q$ denote any POVM effect such that $\tr`(Q\rho) \geq \eta$.
  Let $Q' \coloneqq \eta Q / \tr`(Q\rho)$.
  By definition, $Q'$ satisfies $\tr`(Q'\rho) = \eta$.
  $Q$ and $Q'$ are candidates for the $\DHyp[\eta]{}{}$ optimization in~\eqref{eq:defn-DHyp} and achieve the same objective value:
  \begin{align}
      \frac{\tr`*(Q\Gamma)}{\tr`*(Q\rho)} = \frac{\tr`*(Q'\Gamma)}{\tr`*(Q'\rho)} \ .
  \end{align}
  Consequently, the set of candidate objective values in~\eqref{eq:DHyp-alternative-expressions--ratio-trQGamma-trQrho} is unaltered if we restrict to POVM effects $Q'$ satisfying $\tr`(Q'\rho) = \eta$:
  \begin{align}
        `*{ \frac{\tr`*(Q\Gamma)}{\tr`*(Q\rho)} : 0 \leq Q \leq \Ident, \,  \tr`*(Q\rho) \geq \eta } = `*{ \frac{\tr`(Q'\Gamma)}{\tr`*(Q'\rho)} : 0 \leq Q' \leq \Ident, \,  \tr`*(Q'\rho) = \eta } \ .
  \end{align}
  Therefore,
  \begin{align}
      \label{eq:InfimumLabel1}
    \inf_{ \substack{0 \leq Q \leq \Ident \\ \tr`(Q\rho) \geq \eta }}
          `*{ \frac{\tr`*(Q\Gamma)}{\tr`*(Q\rho)} } \
    = \inf_{ \substack{0 \leq Q' \leq \Ident\\ \tr`(Q'\rho) = \eta }}
          `*{ \frac{\tr`*(Q'\Gamma)}{\tr`*(Q'\rho)} } \
    = \inf_{ \substack{0 \leq Q' \leq \Ident\\ \tr`(Q'\rho) = \eta }}
          `*{ \frac{\tr`*(Q'\Gamma)}{\eta} } \ .
  \end{align}
  The second equality follows because $\tr`(Q'\rho) = \eta$ for every $Q'$ in the second infimum's domain.
  
  Consider, again, any effect $Q$ satisfying $\tr`(Q\rho) \geq \eta$ and the effect $Q' \coloneqq \eta Q / \tr`(Q\rho)$.
  Clearly $Q' \leq Q$, so $\tr`(Q'\Gamma) / \eta \leq \tr`(Q\Gamma) / \eta$.
  Hence, the final infimum in~\eqref{eq:InfimumLabel1} does not decrease if we extend its domain to all $Q$ satisfying $\tr`(Q\rho) \geq \eta$: 
  \begin{align}
      \label{eq:InfimumLabel2}
    \inf_{ \substack{0 \leq Q' \leq \Ident\\ \tr`(Q'\rho) = \eta }}
          `*{ \frac{\tr`*(Q'\Gamma)}{\eta} } \,
    =
    \inf_{ \substack{0 \leq Q \leq \Ident\\ \tr`(Q\rho) \geq \eta }}
          `*{ \frac{\tr`*(Q\Gamma)}{\eta} } \ .
  \end{align}
  Finally, the second infimum in~\eqref{eq:InfimumLabel2} is a minimum, since the infimum's domain is compact.
  Hence,
  \begin{align}
      \label{eq:InfimumLabel3}
    \inf_{ \substack{0 \leq Q \leq \Ident\\ \tr`(Q\rho) \geq \eta }}
          `*{ \frac{\tr`*(Q\Gamma)}{\eta} } \,
    = \min_{ \substack{0 \leq Q \leq \Ident\\ \tr`(Q\rho) \geq \eta }}
          `*{ \frac{\tr`*(Q\Gamma)}{\eta} } \,
    = e^{-\DHyp[\eta]{\rho}{\Gamma}} \ .
  \end{align}
  The second equality follows from \cref{def:defn-DHyp}.
  Chaining together~\eqref{eq:InfimumLabel1},~\eqref{eq:InfimumLabel2}, and~\eqref{eq:InfimumLabel3} yields an equality equivalent to~\eqref{eq:DHyp-alternative-expressions--ratio-trQGamma-trQrho}.
\end{proof}

\section{Complexity (relative) entropy and its variants}
\label{appx:GeneralConstructionComplexityEntropy}

Imagine a hypothesis test performed by an observer able to render only limited-complexity measurement effects.
The hypothesis-testing relative entropy is ill-suited for such a test: the optimization in~\eqref{eq:defn-DHyp} is appropriate only for an observer who can render all POVM effects.
In \cref{appx-topic:defn-complexity-relative-entropy}, we introduce the complexity (relative) entropy as a hypothesis-testing (relative) entropy tailored for the complexity-limited observer.
In the optimization defining the complexity (relative) entropy, candidate POVM effects cannot exceed a given complexity $r \geq 0$.
\cref{appx-topic:Dhypr-basic-properties} details the complexity (relative) entropy's elementary properties.
In \cref{appx-topic:Dhypr-hypo-test}, we apply the complexity relative entropy to hypothesis testing with complexity limitations.
In \cref{appx-topic:complexity-entropy-relationship-to-state-complexity}, we bound the complexity (relative) entropy using state-complexity measures.
\cref{appx-topic:reduced-complexity-entropy} mainly concerns a variant of the complexity (relative) entropy, the reduced complexity (relative) entropy.
The reduced complexity entropy features in our data-compression results and appears in Ref.~\cite{YungerHalpern_2022_Uncomplexity} as ``the complexity entropy.''
Last, in \cref{appx-topic:complexity-conditional-entropy}, we define the complexity conditional entropy and study some of its elementary properties.
In the following, $\{\Mr[r]\}$ denotes any family of POVM-effect-complexity sets (\cref{defn:Mr-sets-general}).
The complexity (relative) entropy and its variants are defined with respect to this set, unless further specified.

\subsection{Definition of the complexity (relative) entropy}
\label{appx-topic:defn-complexity-relative-entropy}

\begin{definition}[Complexity relative entropy]
  \label{def:defn-DHypr}
  Let $\rho$ denote any subnormalized state; and $\Gamma$, any positive-semidefinite operator.
  Let $r \geq 0$ and $\eta \in (0, \tr(\rho)]$.
  Here, $\rho$, $\Gamma$, and every $Q \in \Mr[r]$ act on the same Hilbert space.
  The \emph{complexity-restricted hypothesis-testing relative entropy}, or simply the \emph{complexity relative entropy}, is
  \begin{align}
    \DHypr[r][\eta]{\rho}{\Gamma}
    \coloneqq -\log\; `*(
    \inf_{\substack{Q \in \Mr[r] \\ \tr`(Q\rho) \geq \eta}}
    `*{ \frac{\tr`*(Q\Gamma)}{\tr`*(Q\rho)} }
    )\ .
    \label{eq:defn-DHypr}
  \end{align}
\end{definition}

The denominator $\tr`(Q\rho)$ is a normalization factor ensuring that the complexity (relative) entropy has
some desired properties.
First, the complexity relative entropy satisfies $\DHypr[r][\eta]{\rho}{\rho} = 0$ (\cref{thm:DHypr-zero-same-args}).
Second, if $\rho$ acts on a Hilbert space of dimensionality $d$, the complexity entropy assumes values in the range $[0, \log(d)]$ (\cref{thm:DHypr-trivial-bounds}).
The hypothesis-testing relative entropy is the same with a denominator $\tr`(Q\rho)$ [as in~\eqref{eq:DHyp-alternative-expressions--ratio-trQGamma-trQrho}] or $\eta$ [as in~\eqref{eq:defn-DHyp}].
In contrast, the complexity relative entropy might increase if we replace $\tr`(Q\rho)$ with $\eta$ in~\eqref{eq:defn-DHypr}: it may be impossible to find a $Q\in\Mr[r]$ such that $\tr`(Q\rho)=\eta$ [or such that $\tr`(Q\rho)$ approximates $\eta$ arbitrarily well].
This is typically the case, for instance, for a family of sets that contains only projectors; such a family is defined by~\eqref{eq:setting-defn-Mr}.

\begin{definition}[Complexity entropy]
  \label{def:defn-HHypr}
  Let $\rho$ denote any subnormalized state.
  Let $r \geq 0$ and $\eta \in (0, \tr(\rho) ]$.
  The \emph{complexity-restricted hypothesis-testing entropy}, or simply the \emph{complexity entropy}, is
  \begin{align}
    \HHypr[r][\eta]{\rho}
    \coloneqq - \DHypr[r][\eta]{\rho}{\Ident}
    = \log\; `*(
    \inf_{\substack{Q \in \Mr[r] \\ \tr`(Q\rho) \geq \eta}}
    `*{ \frac{\tr`*(Q)}{\tr`*(Q\rho)} } ) \ .
    \label{eq:defn-HHypr}
  \end{align}
\end{definition}

The complexity entropy measures how well one can distinguish $\rho$ from the maximally mixed state, using a limited-complexity measurement effect.
Indeed, if $\rho$ acts on a Hilbert space of dimensionality $d$, then $\HHypr[r][\eta]{\rho} = \log`(d) - \DHypr[r][\eta]{\rho}{\pi}$. The $\pi \coloneqq \Ident/d$ denotes the maximally mixed state.
Two special cases offer insight. First, $\HHypr[r][\eta]{\rho} = 0$ if $\rho$ is pure and $\rho \in \Mr[r]$.
Second, $\HHypr[r][\eta]{\rho} = \log(d)$ if $\rho$ is normalized and $\Ident \in \Mr[r]$ is the only candidate for the $\HHypr[r][\eta]{\rho}$ optimization.

The complexity entropy is analogous to the strong complexity of Ref.~\cite{Brandao2021PRXQ_models}.
Compared to the strong complexity, the complexity entropy quantifies distinguishability in terms of a hypothesis test, rather than a modified trace distance.
Also, the complexity entropy bears a meaningful interpretation for mixed states.

\subsection{Elementary properties of the complexity (relative) entropy}
\label{appx-topic:Dhypr-basic-properties}

The complexity (relative) entropy possesses several basic properties.

\begin{proposition}[Vanishing complexity relative entropy]
  \label{thm:DHypr-zero-same-args}
  Let $\rho$ denote any subnormalized state.
  Let $r\geq 0$ and $\eta \in (0,\tr(\rho) ]$.
  It holds that
  \begin{align}
    \DHypr[r][\eta]{\rho}{\rho} = 0 \ .
  \end{align}
\end{proposition}
\begin{proof}[**thm:DHypr-zero-same-args]
  Every candidate effect $Q\in\Mr[r]$ for the optimization in~\eqref{eq:defn-DHypr}, including $\Ident \in \Mr[r=0] \subset \Mr[r]$, achieves the objective value $\tr`(Q \rho)/\tr`(Q\rho) = 1$.
  Hence, $\DHypr[r][\eta]{\rho}{\rho} = \log(1) = 0$.
\end{proof}

In general, the converse is false: there may exist a state $\sigma\neq\rho$ such that $\DHypr[r][\eta]{\rho}{\sigma} = 0$.
This situation arises if the measurement effects in $\Mr[r]$ are too crude to distinguish $\rho$ from $\sigma$.
For instance, consider the sets $\Mr[r]$ defined in~\eqref{eq:setting-defn-Mr} for $n$ qubits, with respect to the set of two-qubit gates.
Let $\rho=\proj{1^n}$.
For any $\eta\in(0,1]$, the only effect $Q\in\Mr[{r=0}]$ satisfying $\tr`(Q\rho)\geq\eta$ is $Q=\Ident$.
Therefore, trivially, $\DHypr[{r=0}][\eta]{\rho}{\sigma} = 0$ for all $\sigma$.

The complexity relative entropy $\DHypr[r][\eta]{\rho}{\Gamma}$ lies within a fixed range dependent on only $\tr(\rho)$ and the eigenvalues of $\Gamma$.
The range takes a simple form if $\rho$ is normalized and $\Gamma$ is positive-definite.
\begin{proposition}[General bounds]
   \label{thm:DHypr-trivial-bounds}
  Let $\rho$ denote any subnormalized state, and $\Gamma$ any positive-semidefinite operator, that act on a Hilbert space of dimensionality $d$.
  Let $r \geq 0$ and $\eta \in (0, \tr(\rho) ]$.
  It holds that
  \begin{subequations}
    \begin{align}
      \DHypr[r][\eta]{\rho}{\Gamma} &\geq -\log \bm{(} \tr`(\Gamma) \bm{)} + \log \bm{(} \tr`(\rho) \bm{)} \ .
        \label{eq:thm-DHypr-trivial-bounds--lower-bound}
    \end{align}
    Furthermore, if $\Gamma$ is positive-definite (has full rank),
    \begin{align}
      \DHypr[r][\eta]{\rho}{\Gamma} &\leq \log`*( \opnorm{\Gamma^{-1}} ) + \log \bm{(} \tr`(\rho) \bm{)} \ .
        \label{eq:thm-DHypr-trivial-bounds--upper-bound}
    \end{align}
  \end{subequations}
  In particular,
  \begin{align}
      - \log \bm{(} \tr`(\rho) \bm{)}
      \leq \HHypr[r][\eta]{\rho}
      \leq \log(d) - \log \bm{(} \tr`(\rho) \bm{)} \ .
       \label{eq:thm-DHypr-trivial-bounds--HHypr}
  \end{align}
  Consequently, if $\rho$ is normalized and $\Gamma$ is positive-definite,
  \begin{subequations}
    \begin{align}
      -\log \bm{(} \tr`(\Gamma) \bm{)}
      \leq \DHypr[r][\eta]{\rho}{\Gamma}
        &\leq \log `*( \opnorm{\Gamma^{-1}} ) \ ,
        \label{eq:thm-DHypr-trivial-bounds--special-case}
      \\
      \text{and} \; \; \;
      0 \leq \HHypr[r][\eta]{\rho}
      &\leq \log(d) \ .
      \label{eq:thm-DHypr-trivial-bounds--special-case--HHypr}
    \end{align}
  \end{subequations}
\end{proposition}

\noindent
\cref{eq:thm-DHypr-trivial-bounds--lower-bound} implies that $\DHypr[r][\eta]{\rho}{\sigma} \geq 0$ for normalized states $\rho$ and $\sigma$.

\begin{proof}[**thm:DHypr-trivial-bounds]
  \eqref{eq:thm-DHypr-trivial-bounds--lower-bound} follows because $\Ident \in \Mr[r]$ is a candidate effect for the $\DHypr[r][\eta]{\rho}{\Gamma}$ optimization.
  To prove~\eqref{eq:thm-DHypr-trivial-bounds--upper-bound}, let $Q \in \Mr[r]$ denote any candidate for the $\DHypr[r][\eta]{\rho}{\Gamma}$ optimization.
  If $\Gamma$ is positive-definite, $\Gamma \geq \opnorm{\Gamma^{-1}}^{-1} \Ident$, wherein $\opnorm{\Gamma^{-1}}^{-1}$ is the smallest eigenvalue of $\Gamma$.
  Furthermore, $\rho \leq \tr`(\rho) \Ident$, since $\tr`(\rho) \leq 1$.  
  Thus, $\tr`(Q\Gamma) \geq \opnorm{\Gamma^{-1}}^{-1} \tr`(Q)$, and $\tr`(Q\rho) \leq \tr`(\rho) \, \tr`(Q)$.
  These two inequalities, together, imply that
  \begin{align}
    \frac{ \tr`*(Q\Gamma) }{ \tr`*(Q\rho) } \
    \geq \frac{ \opnorm{\Gamma^{-1}}^{-1} \tr`*(Q) }{ \tr`*(\rho) \tr`*(Q) } \
    = \bm{(} \opnorm{\Gamma^{-1}} \, \tr`(\rho) \bm{)}^{-1} \ .
  \end{align}
  Therefore, $\bm{(} \opnorm{\Gamma^{-1}} \, \tr`(\rho) \bm{)}^{-1}$ lower-bounds the objective value $\tr`(Q\Gamma) / \tr`(Q\rho)$ of every candidate effect $Q$.
  By the definition of an infimum,
  \begin{align}
    \bm{(} \opnorm{\Gamma^{-1}} \, \tr`(\rho) \bm{)}^{-1}
    \leq \inf_{\substack{Q \in \Mr[r] \\ \tr`(Q\rho) \geq \eta}}
    `*{ \frac{\tr`*(Q\Gamma)}{\tr`*(Q\rho)} } 
    = e^{ -\DHypr[r][\eta]{\rho}{\Gamma} } \ ,
  \end{align}
  which is equivalent to~\eqref{eq:thm-DHypr-trivial-bounds--upper-bound}.
  One obtains~\eqref{eq:thm-DHypr-trivial-bounds--HHypr} by setting $\Gamma = \Ident$ in~\eqref{eq:thm-DHypr-trivial-bounds--lower-bound} and~\eqref{eq:thm-DHypr-trivial-bounds--upper-bound}.
\end{proof}

The complexity relative entropy never exceeds the hypothesis-testing relative entropy.
\begin{proposition}[Upper bound by hypothesis-testing relative entropy]
  \label{thm:DHypr-bounded-by-DHyp}
  Let $\rho$ denote any subnormalized state; and $\Gamma$, any positive-semidefinite operator.
  Let $r \geq 0$ and $\eta \in (0, \tr(\rho) ]$.
  It holds that
    \begin{align}
        \DHypr[r][\eta]{\rho}{\Gamma} \leq \DHyp[\eta]{\rho}{\Gamma}
        \; \; \; \text{and}  \; \; \;
        \HHypr[r][\eta]{\rho} \geq \HHyp[\eta]{\rho} \ .
    \end{align}
\end{proposition}
\begin{proof}[**thm:DHypr-bounded-by-DHyp]
  The $\DHyp[\eta]{\rho}{\Gamma}$ optimization ranges over all POVM effects, while the $\DHypr[r][\eta]{\rho}{\Gamma}$ optimization ranges over only effects in $\Mr[r]$.
\end{proof}

The greater an agent's computational power (the greater the $r$), the less mixed $\rho$ appears $[$the less $\HHypr[r][\eta]{\rho}$ is$]$.
Likewise, the greater an agent's error intolerance (the greater the $\eta$), the more mixed $\rho$ appears $[$the greater $\HHypr[r][\eta]{\rho}$ is$]$.
Consequently, the complexity entropy monotonically decreases as $r$ increases and monotonically increases as $\eta$ increases.
The complexity entropy inherits its monotonicity in $\eta$ from the hypothesis-testing entropy.

\begin{proposition}[Monotonicity in $r$ and $\eta$]
  \label{thm:DHypr-monotonous-in-r-and-in-eta}
  Let $\rho$ denote any subnormalized state; and $\Gamma$, any positive-semidefinite operator.
  Let $r \geq 0$ and $\eta \in (0, \tr(\rho) ]$.
  \begin{enumerate}
  \item For all $r' \geq r$, $\DHypr[r'][\eta]{\rho}{\Gamma} \geq \DHypr[r][\eta]{\rho}{\Gamma}$, and $\HHypr[r'][\eta]{\rho} \leq \HHypr[r][\eta]{\rho}$.
  \item For all $\eta' \in (0, \eta ]$, $\DHypr[r][\eta']{\rho}{\Gamma} \geq \DHypr[r][\eta]{\rho}{\Gamma}$, and $\HHypr[r][\eta']{\rho} \leq \HHypr[r][\eta]{\rho}$.
  \end{enumerate}
\end{proposition}
\begin{proof}[**thm:DHypr-monotonous-in-r-and-in-eta]
  The monotonicity in $r$ follows because $\Mr[r] \subset \Mr[r']$: every candidate $Q \in \Mr[r]$ for the $\DHypr[r][\eta]{\rho}{\Gamma}$ optimization belongs to $\Mr[r']$ and is therefore a candidate for the $\DHypr[r'][\eta]{\rho}{\Gamma}$ optimization.
  The monotonicity in $\eta$ follows because every candidate $Q$ for the $\DHypr[r][\eta]{\rho}{\Gamma}$ optimization satisfies $\tr(Q \rho) \geq \eta \geq \eta'$ and is therefore a candidate for the $\DHypr[r][\eta']{\rho}{\Gamma}$ optimization.
\end{proof}

Like standard relative entropies, the complexity relative entropy enjoys a scaling property in its second argument.

\begin{proposition}[Scaling property in the second argument]
  \label{thm:DHypr-scaling-2ndarg}
  Let $\rho$ denote any subnormalized state; and $\Gamma$, any positive-semidefinite operator.
  Let $r \geq 0$, $\eta \in (0, \tr(\rho) ]$, and $a > 0$.
  It holds that
  \begin{align}
    \DHypr[r][\eta]{\rho}{a\Gamma} = \DHypr[r][\eta]{\rho}{\Gamma} - \log(a)\ .
  \end{align}
\end{proposition}
\begin{proof}[**thm:DHypr-scaling-2ndarg]
  The equality holds because $\log(a)$ factorizes out of the optimizer in~\eqref{eq:defn-DHypr}.
\end{proof}

The complexity (relative) entropy inverts (preserves) the partial order of positive-semidefinite operators.
\begin{proposition}[Monotonicity under operator ordering]
    \label{thm:DHypr-monotonic-under-ordering}
    Let $\rho$ and $\rho' \geq \rho$ denote any subnormalized states.
    Let $\Gamma$ and $\Gamma' \leq \Gamma$ denote any positive-semidefinite operators.
    Let $r \geq 0$ and $\eta \in (0, \tr(\rho) ]$.
    It holds that
    \begin{subequations}
    \begin{align}
      \label{eq:thm-DHypr-monotonic-under-ordering}
        \DHypr[r][\eta]{\rho}{\Gamma} &\leq \DHypr[r][\eta]{\rho'}{\Gamma'}
        \\
        \text{and} \; \; \;
      \label{eq:thm-DHypr-monotonic-under-ordering--HHypr}
        \HHypr[r][\eta]{\rho} &\geq \HHypr[r][\eta]{\rho'}\ .
    \end{align}
    \end{subequations}
\end{proposition}

\begin{proof}[**thm:DHypr-monotonic-under-ordering]
  For every POVM effect $Q$, $\tr`(Q\rho)\leq\tr`(Q\rho')$, and $\tr`(Q\Gamma)\geq\tr`(Q\Gamma')$.
  Therefore,
  \begin{align}
    \ee^{-\DHypr[r][\eta]{\rho}{\Gamma}}
    = \inf_{\substack{Q\in\Mr[r]\\ \tr`(Q\rho)\geq\eta}} `*{ \frac{\tr`*(Q\Gamma)}{\tr`*(Q\rho)} }
    \geq \inf_{\substack{Q\in\Mr[r]\\ \tr`(Q\rho')\geq\eta}} `*{ \frac{\tr`*(Q\Gamma')}{\tr`*(Q\rho')} }
    = \ee^{-\DHypr[r][\eta]{\rho'}{\Gamma'}}\ .
  \end{align}
  One obtains~\eqref{eq:thm-DHypr-monotonic-under-ordering--HHypr} by setting $\Gamma = \Gamma' = \Ident$ in~\eqref{eq:thm-DHypr-monotonic-under-ordering}.
\end{proof}

For tensor-product states, the complexity entropy has a property similar to the von Neumann entropy's subadditivity.

\begin{proposition}[Subadditivity for tensor-product states]
  \label{thm:DHypr-subadditivity}
  Let $S$ and $S'$ denote distinct quantum systems.
  Let $\rho_S$ and $\rho_{S'}'$ denote any subnormalized states, and let $\Gamma_S$ and $\Gamma_{S'}'$ denote any positive-semidefinite operators.
  Let $r,r' \geq 0$, $\eta \in (0, \tr(\rho) ]$, and $\eta' \in (0, \tr(\rho') ]$.
  It holds that
  \begin{subequations}
  \begin{align}
    \DHypr[r+r'][\eta\eta']{\rho \otimes \rho'}{\Gamma \otimes \Gamma'}
    &\geq \DHypr[r][\eta]{\rho}{\Gamma} + \DHypr[r'][\eta']{\rho'}{\Gamma'}
      \label{eq:thm-DHypr-subadditivity}
    \\
    \text{and} \; \; \;
    \HHypr[r+r'][\eta\eta']{\rho \otimes \rho'}
    &\leq \HHypr[r][\eta]{\rho} + \HHypr[r'][\eta']{\rho'}\ .
      \label{eq_Superadd_H}
    \end{align}
    \end{subequations}
\end{proposition}

A consequence of \cref{thm:DHypr-subadditivity} is, ancillas cannot decrease the complexity relative entropy. For all subnormalized states $\rho$ and $\sigma$ defined on the same Hilbert space, and for all subnormalized states $\tau$,
\begin{align}
    \DHypr[r][\eta]{\rho \otimes \tau}{\sigma \otimes \tau}
    \geq \DHypr[r][\eta]{\rho}{\sigma}\ .
    \label{eq:DHypr-ancillas}
\end{align}
Inequality~\eqref{eq:DHypr-ancillas} follows from setting $r' = 0$ and $\eta' =1$ in~\eqref{eq:thm-DHypr-subadditivity} and applying \cref{thm:DHypr-zero-same-args} $`\big[\DHypr[r'][\eta']{\tau}{\tau} = 0]$.
Inequality~\eqref{eq:DHypr-ancillas} is compatible with observations about the power of the one-clean-qubit computational model (DQC1)~\cite{Knill_98_Power} and about how tossing an extra pure qubit into a black hole is expected to decrease the black hole's complexity~\cite{Brown_18_Second,Susskind2018arXiv_ThreeLectures}.

\begin{proof}[*thm:DHypr-subadditivity]
  Consider any $\zeta,\zeta'>0$.
  By the definitions of $\DHypr[r][\eta]{\rho}{\Gamma}$ and $\DHypr[r'][\eta']{\rho'}{\Gamma'}$, there exist $Q \in \Mr[r][S]$, with $\tr`(Q\rho) \geq \eta$, and $Q' \in \Mr[r'][S']$, with $\tr`(Q'\rho') \geq \eta'$, such that
  \begin{align}
    \frac{\tr`*(Q\Gamma)}{\tr`*(Q\rho)}
    \leq \ee^{-\DHypr[r][\eta]{\rho}{\Gamma}} + \zeta
    \; \;  \; \; \; \; \text{and} \; \; \; \; \; \;
    \frac{\tr`*(Q'\Gamma')}{\tr`*(Q'\rho')}
    \leq \ee^{-\DHypr[r'][\eta']{\rho'}{\Gamma'}} + \zeta' \ .
  \end{align}
  By \cref{defn:Mr-sets-general}, $\Mr[r][S] \otimes \Mr[r'][S'] \subseteq \Mr[r+r'][S S']$, so $Q \otimes Q' \in \Mr[r+r'][S S']$.
  Moreover,
  \begin{align}
    \tr`*( `[Q\otimes Q']\,`[\rho\otimes\rho'] ) = \tr`(Q\rho)\tr`(Q'\rho') \geq \eta\eta' \ .
  \end{align}
  Therefore, $Q\otimes Q'$ is a candidate for the $\DHypr[r+r'][\eta\eta']{\rho\otimes\rho'}{\Gamma\otimes\Gamma'}$ optimization.
  Consequently,
  \begin{align}
    \ee^{-\DHypr[r+r'][\eta\eta']{\rho\otimes\rho'}{\Gamma\otimes\Gamma'}}
    \leq \frac{\tr`*(Q\Gamma)\tr`*(Q'\Gamma')}{\tr`*(Q\rho)\tr`*(Q'\rho')}
    \leq `\Big( \ee^{-\DHypr[r][\eta]{\rho}{\Gamma}} + \zeta )
    `\Big( \ee^{-\DHypr[r'][\eta']{\rho'}{\Gamma'}} + \zeta' )\ ,
  \end{align}
  which implies~\eqref{eq:thm-DHypr-subadditivity}, since $\zeta$ and $\zeta'$ are arbitrary.
  One obtains~\eqref{eq_Superadd_H} by setting $\Gamma = \Ident_S$ and $\Gamma' = \Ident_{S'}$ in~\eqref{eq:thm-DHypr-subadditivity}.
\end{proof}

The complexity relative entropy never increases under any partial trace.

\begin{proposition}[Monotonicity under partial traces]
  \label{thm:cplxrelentropy-monotonicity-partial-trace}
  Let $A$ and $B$ denote distinct quantum systems.
  Let $\rho_{AB}$ denote any subnormalized quantum state of $AB$.
  Let $\Gamma_{AB}$ denote any positive-semidefinite operator.
  Let $r \geq 0$ and $\eta \in (0, \tr(\rho) ]$.
  It holds that
    \begin{align}
        \DHypr[r][\eta]{\rho_{AB}}{\Gamma_{AB}} \geq \DHypr[r][\eta]{\rho_{A}}{\Gamma_{A}} \ ,
        \label{eq:thm:cplxrelentropy-monotonicity-partial-trace--DHypr}
    \end{align}
   wherein $\rho_A \coloneqq \tr_B`(\rho_{AB})$ and $\Gamma_A \coloneqq \tr_B`(\Gamma_{AB})$.
   Furthermore,
    \begin{align}
        \HHypr[r][\eta]{\rho_{AB}} \leq \HHypr[r][\eta]{\rho_A} + \log(d_B) \ .
        \label{eq:thm:cplxrelentropy-monotonicity-partial-trace--HHypr}
    \end{align}
\end{proposition}

\begin{proof}[**thm:cplxrelentropy-monotonicity-partial-trace]
  Consider any $\zeta>0$.
  There exists a $Q_A \in \Mr[r][A]$ such that $\tr`(Q_A\rho_A) \geq \eta$ and $\tr`(Q_A\Gamma_A)/\tr`(Q_A\rho_A) \leq \exp \bm{(} -\DHypr[r][\eta]{\rho_A}{\Gamma_A} \bm{)} + \zeta$.
  By Properties~\ref{item:POVM-effect-complexity--identity} and~\ref{item:POVM-effect-complexity--subsystems} of \cref{defn:Mr-sets-general}, $\tilde{Q}_{AB} \coloneqq Q_A \otimes \Ident_B \in \Mr[r][AB]$.
  Moreover, $\tr`(\tilde{Q}_{AB} \rho_{AB}) = \tr`(Q_A \rho_A) \geq \eta$.
  Therefore, $\tilde{Q}_{AB}$ is a candidate for the $\DHypr[r][\eta]{\rho_{AB}}{\Gamma_{AB}}$ optimization.
  Consequently,
  \begin{align}
    \ee^{ -\DHypr[r][\eta]{\rho_{AB}}{\Gamma_{AB}} }
    \leq \frac{ \tr`*(\tilde{Q}_{AB}  \Gamma_{AB} ) }{ \tr`*(\tilde{Q}_{AB} \rho_{AB}) }
    = \frac{ \tr`*(Q_A \Gamma_{A} ) }{ \tr`*(Q_A \rho_{A}) }
    \leq \ee^{ -\DHypr[r][\eta]{\rho_A}{\Gamma_A} } + \zeta \ ,
  \end{align}
  which implies~\eqref{eq:thm:cplxrelentropy-monotonicity-partial-trace--DHypr}, since $\zeta$ is arbitrary.
  \eqref{eq:thm:cplxrelentropy-monotonicity-partial-trace--HHypr} follows from applying first~\eqref{eq:thm:cplxrelentropy-monotonicity-partial-trace--DHypr} and then \cref{thm:DHypr-scaling-2ndarg}, to get
  \begin{align}
    \HHypr[r][\eta]{\rho_{AB}}
    = - \DHypr[r][\eta]{\rho_{AB}}{\Ident_{AB}}
    \leq - \DHypr[r][\eta]{\rho_{A}}{d_B\Ident_{A}}
    = \HHypr[r][\eta]{\rho_A} + \log(d_B) \ .
  \end{align}
\end{proof}

We can bound the complexity (relative) entropy in some cases where its argument undergoes a limited-complexity unitary.

\begin{proposition}[Unitary operations on arguments]
    \label{thm:DHypr-arg-U-rp}
    Let $\rho$ denote any subnormalized state; and $\Gamma$, any positive-semidefinite operator.
    Let $\Psimple$ denote any set of simple POVM effects (\cref{defn:Psimple-general}) and $C$ any adjoint-invariant unitary-complexity measure (\cref{defn:unitary-complexity}).
    Let $\{\Mr[r] = \Mr[r](\Psimple,C)\}$ denote the family of POVM-effect-complexity sets defined by~\eqref{eq:Mr-general-from-Psimple-and-superop-cplx-measure}.
    Let $r , r' \geq 0$ and $\eta\in(0,\tr(\rho)]$.
    Let $U$ denote any unitary satisfying $C(U) \leq r'$.
    It holds that
    \begin{subequations}
    \begin{align}
        \DHypr[r+r'][\eta]{U\rho U^\dagger}{U \Gamma U^\dagger} &\geq \DHypr[r][\eta]{\rho}{\Gamma}
        \label{eq:thm-DHypr-arg-U-rp-for-DHypr---adjoint-invariant}
        \\
        \text{and} \; \; \;
        \HHypr[r+r'][\eta]{ U \rho U^\dagger } &\leq \HHypr[r][\eta]{ \rho }\ .
        \label{eq:thm-DHypr-arg-U-rp-for-HHypr---adjoint-invariant}
    \end{align}
    \end{subequations}
\end{proposition}
\begin{proof}[**thm:DHypr-arg-U-rp]
    Consider any $\zeta>0$.
    There exists a $Q\in \Mr[r]$ such that $\tr`(Q\rho)\geq\eta$ and $\tr`(Q\Gamma)/\tr`(Q\rho) \leq \exp \bm{(} -\DHypr[r][\eta]{\rho}{\Gamma} \bm{)} + \zeta$.
    $Q = U_0^\dagger P U_0$ for some effect $P \in \Psimple$ and for some unitary $U_0$ satisfying $C(U_0) \leq r$.
    Let $Q' \coloneqq U Q U^\dagger = ( U_0 U^\dagger)^\dagger P (U_0 U^\dagger)$.
    By \cref{defn:superoperator-complexity-general} and by the adjoint invariance of $C$, $C(U_0U^\dagger) \leq C(U_0) + C(U^\dagger) = C(U_0) + C(U) \leq r + r'$, so $Q' \in \Mr[r+r']$.
    Moreover, $\tr`(Q' [U\rho U^\dagger]) = \tr`(Q\rho) \geq \eta$.
    Therefore, $Q'$ is a candidate for the $\DHypr[r+r'][\eta]{U\rho U^\dagger}{U \Gamma U^\dagger}$ optimization, so
    \begin{align}
        \ee^{ -\DHypr[r+r'][\eta]{U\rho U^\dagger}{U \Gamma U^\dagger} }
        \leq \frac{\tr`\big(Q' `\big[U\Gamma U^\dagger])}{\tr`\big(Q' `\big[U\rho U^\dagger])}
        = \frac{\tr`(Q\Gamma)}{\tr`*(Q\rho)}
        \leq \ee^{ -\DHypr[r][\eta]{\rho}{\Gamma} } + \zeta \ ,
    \end{align}
    which implies~\eqref{eq:thm-DHypr-arg-U-rp-for-DHypr---adjoint-invariant}, since $\zeta$ is arbitrary.
\end{proof}

For composite systems, one can bound the complexity (relative) entropy in terms of hypothesis-testing (relative) entropies on each subsystem, if $\Mr[r]$ contains only tensor-product POVM effects.
In most cases, $\Mr[r]$ contains nonlocal effects for all $r > 0$; in such cases, only $\Mr[r=0]$ may consist solely of tensor-product effects.
Importantly, the set $\Mr[r=0]$ defined in~\eqref{eq:setting-defn-Mrzero}, for $n$ qubits, contains only tensor-product effects.

\begin{proposition}[Upper bound by hypothesis-testing relative entropies of subsystems]
  \label{thm:bound-DHypr-DHyp-tensor-products}
  Let $S$ denote a composition of $N$ quantum subsystems: $S = S_1 S_2 \ldots S_N$.
  Let $\rho_S$ denote any subnormalized state.
  For all $i=1,2,\ldots,N$, let $\Gamma_{S_i}$ denote any positive-semidefinite operator.
  Let $\Gamma_S \coloneqq \bigotimes_{i=1}^N \Gamma_{S_i}$.
  Let $r \geq 0$ and $\eta\in(0,\tr(\rho_S)]$.
  Suppose that $\Mr[r][S]$ contains only tensor-product POVM effects: every $Q_S \in \Mr[r][S]$ is of the form $ \bigotimes_{i=1}^N Q_{S_i}$. Each $Q_{S_i}$ is an effect on $S_i$.
  It holds that
  \begin{subequations}
    \begin{align}
      \DHypr[r][\eta]{\rho_S}{\Gamma_S}
      &\leq \sum_{i=1}^N \DHyp[\eta]{\rho_{S_i}}{\Gamma_{S_i}}
      \label{eq:bound-DHypr-DHyp-tensor-products}
      \\
      \text{and} \; \; \;
      \HHypr[r][\eta]{\rho_S}
      &\geq \sum_{i=1}^N \HHyp[\eta]{\rho_{S_i}} \ .
      \label{eq:bound-DHypr-DHyp-tensor-products---Hhypr}
    \end{align}
  \end{subequations}
\end{proposition}

\begin{proof}[**thm:bound-DHypr-DHyp-tensor-products]
  Consider any $\zeta>0$.
  There exists a $Q_S = \bigotimes_{i=1}^N Q_{S_i} \in \Mr[r][S]$, with each $Q_{S_i}$ an effect on $S_i$, such that $\tr`*(Q_S\rho_S) \geq \eta$ and $\tr`*(Q_S\Gamma_S) / \tr`*(Q_S\rho_S) \leq \exp \bm{(} - \DHypr[r][\eta]{\rho_S}{\Gamma_S} \bm{)} + \zeta$.
  For all $i=1,2,\ldots,N$, let $Q_S^{(i)} \coloneqq `*( \Ident_{S_1} \otimes \cdots \otimes \Ident_{S_{i-1}} ) \otimes Q_{S_i} \otimes `*( \Ident_{S_{i+1}} \otimes \cdots \otimes \Ident_{S_N} )$.
  For all $i$, $Q_S^{(i)} \geq Q_S$, since $\Ident_{S_j} \geq Q_{S_j}$ for all $j \neq i$;
  consequently, $\tr`*(Q_{S_i} \, \rho_{S_i}) = \tr`\big( Q_S^{(i)} \, \rho_S ) \geq \tr`*( Q_S \, \rho_S ) \geq \eta$.
  Thus, for all $i$, $Q_{S_i}$ is a candidate for the $\DHyp[\eta]{\rho_{S_i}}{\Gamma_{S_i}}$ optimization, so $\exp \bm{(} -\DHyp[\eta]{\rho_{S_i}}{\Gamma_{S_i}} \bm{)} \leq \tr`\big(Q_{S_i} \Gamma_{S_i})$.
  Hence,
  \begin{align}
    \ee^{ - \sum_{i=1}^N \DHyp[\eta]{\rho_{S_i}}{\Gamma_{S_i}} }
    \leq \prod_{i=1}^N \frac{\tr`*( Q_{S_i} \Gamma_{S_i} )}{\tr`*( Q_{S_i} \rho_{S_i} )}
    \leq \prod_{i=1}^N \frac{\tr`*( Q_{S_i} \Gamma_{S_i} )}{\tr`*( Q_S \rho_S )}
    = \frac{\tr`*( Q_S \Gamma_S )}{\tr`*( Q_S \rho_S )}
    \leq \ee^{ -\DHypr[r][\eta]{\rho_S}{\Gamma_S} } + \zeta \ ,
  \end{align}
  which implies~\eqref{eq:bound-DHypr-DHyp-tensor-products}, since $\zeta$ is arbitrary.
\end{proof}

\subsection{Complexity relative entropy and hypothesis testing}
\label{appx-topic:Dhypr-hypo-test}

We prove \cref{mainthm:DHypr-interpretation-hypo-test} of the main text, quantifying the type~I and type~II errors of a hypothesis test that involves limited-complexity measurement effects.
 
\begin{proposition}[Hypothesis testing with complexity limitations]
    \label{thm:DHypr-interpretation-hypo-test}
    Let $\rho$ and $\sigma$ denote any quantum states.
    Let $r \geq 0$, $\eta\in(0,1]$, and $\delta \in (0,1]$.
    The following statements are equivalent:
    \begin{enumerate}[label=(\roman*)]
    \item For all $\zeta>0$, there exist a $Q\in \Mr[r]$ and a $q \in [\eta,1]$ such that $\tr`([qQ]\rho) = \eta$ and $\tr`([qQ]\sigma) \leq \delta + \zeta$.
    \item It holds that
    \begin{align}
        \DHypr[r][\eta]{\rho}{\sigma} \geq -\log`*( \frac\delta\eta ) \ .
        \label{eq:appx-DHypr-interpretation-hypo-test}
    \end{align}
    \end{enumerate}
    
    Consequently, there exist a $Q\in \Mr[r]$ and a $q \in (0,1]$ such that $\tr`([qQ]\rho) = \eta$ and $\tr`([qQ]\sigma) < \delta$ if and only if
    \begin{align}
        \DHypr[r][\eta]{\rho}{\sigma} > -\log`*( \frac\delta\eta ) \ .
        \label{eq:appx-DHypr-interpretation-hypo-test-strict}
    \end{align}
\end{proposition}

\begin{corollary}
  \label{thm:DHypr-interpretation-hypo-test-Mr-closed}
  Let $r \geq 0$.
  If $\Mr[r]$ is compact, then statements (i) and (ii) in \cref{thm:DHypr-interpretation-hypo-test} are equivalent to
  \begin{enumerate}[label=(\roman*)]
    \item[(iii)] There exist a $Q\in \Mr[r]$ and a $q \in [\eta,1]$ such that $\tr`([qQ]\rho) = \eta$ and $\tr`([qQ]\sigma) \leq \delta$.
  \end{enumerate}
  In particular, suppose that $\Psimple$ is a set of simple POVM effects (\cref{defn:Psimple-general}) that is compact, and suppose that $C_{\mathcal{G}}$ is a circuit-superoperator-complexity measure (\cref{defn:circuit-complexity-measure}) associated with a compact gate set $\mathcal{G}$.
  Then $\Mr[r]`*(\Psimple, C_{\mathcal{G}})$ (\cref{defn:Mr-general-from-Psimple-and-superop-cplx-measure}) is compact.
\end{corollary}
\noindent
In particular, the sets $\Mr[r]$ defined in~\eqref{eq:setting-defn-Mr} are compact.

\begin{proof}[*thm:DHypr-interpretation-hypo-test]
    To prove the first part of the proposition, consider any $\zeta > 0$.
    Suppose there exist a $Q\in\Mr[r]$ and a $q\in [\eta,1]$ such that $\tr`([qQ]\rho) = \eta$ and $\tr`([qQ]\sigma) \leq \delta + \zeta$.
    Then $\tr`(Q\rho) \geq \tr`([qQ]\rho) = \eta$, and
    \begin{align}
        \ee^{-\DHypr[r][\eta]{\rho}{\sigma}}
        \leq \frac{\tr`*(Q\sigma)}{\tr`*(Q\rho)}
        = \frac{\tr`*([qQ]\sigma)}{\tr`*([qQ]\rho)}
        \leq \frac{\delta + \zeta}{\eta}\ .
    \end{align}
    Therefore, $\ee^{-\DHypr[r][\eta]{\rho}{\sigma}} \leq \delta/\eta$, since $\zeta$ is arbitrary.
    Conversely, suppose that $\exp \bm{(} -\DHypr[r][\eta]{\rho}{\sigma} \bm{)} \leq \delta /\eta$.
    Consider again any $\zeta>0$.
    There exists a $Q\in\Mr[r]$ such that $\tr`(Q\rho)\geq\eta$ and $\tr`(Q\sigma) / \tr`(Q\rho) \leq (\delta + \zeta)/\eta$.
    Let $q \coloneqq \eta/\tr`(Q\rho)$.
    $q \in [\eta,1]$, since $\eta\leq \tr`(Q\rho)\leq 1$.
    Then $\tr`([qQ]\rho) = \eta$, and
    \begin{align}
        \tr`([qQ]\sigma)
        = \tr`([qQ]\rho) \cdot \frac{\tr`*(Q\sigma)}{\tr`*(Q\rho)}
        \leq \tr`([qQ]\rho) \cdot \frac{\delta+\zeta}{\eta}
        = \delta + \zeta \ .
    \end{align}

    We now prove the second part of the proposition.
    Suppose there exist a $Q\in\Mr[r]$ and a $q \in (0,1]$ such that $\tr`([qQ]\rho) = \eta$ and $\tr`([qQ]\sigma) < \delta$.  
    Let $\delta' \coloneqq \tr`([qQ]\sigma) < \delta$.
    $Q$, $q$, and $\delta'$ satisfy the conditions equivalent to~\eqref{eq:appx-DHypr-interpretation-hypo-test}, so $\exp \bm{(} -\DHypr[r][\eta]{\rho}{\sigma} \bm{)} \leq \delta' / \eta < \delta / \eta$.
    Conversely, suppose that $\exp \bm{(} -\DHypr[r][\eta]{\rho}{\sigma} \bm{)} < \delta / \eta$.
    Let $\delta'' \coloneqq \eta \exp \bm{(} -\DHypr[r][\eta]{\rho}{\sigma} \bm{)} < \delta$.
    Then $\DHypr[r][\eta]{\rho}{\sigma} \leq -\log(\delta''/\eta)$, as in~\eqref{eq:appx-DHypr-interpretation-hypo-test}, so there exist $Q\in\Mr[r]$ and $q \in (0,1]$ such that $\tr`([qQ]\rho)=\eta$ and $\tr`([qQ]\sigma) \leq \delta''+\zeta$.
    By considering $\zeta$ sufficiently small such that $\delta''+\zeta < \delta$, one obtains $\tr`([qQ]\sigma) < \delta$.
\end{proof}

\begin{proof}[*thm:DHypr-interpretation-hypo-test-Mr-closed]
  The implication $\text{(iii)}\Rightarrow\text{(i)}$ is straightforward; we need to show that $\text{(i)}\Rightarrow\text{(iii)}$.
  Assume (i).
  Consider the continuous functions $f, g : [\eta,1]\times \Mr[r] \to [0,1]$ defined as $f(q, Q) \coloneqq \tr`([qQ]\rho)$ and $g(q, Q) \coloneqq \tr`([qQ]\sigma)$.
  Let $A$ denote the preimage of $\eta$ under $f$: $A = f^{-1}(`{\eta})$.
  $A$ is nonempty, since, by assumption, (i) holds.
  $A$ is closed, since $f$ is continuous and $`{\eta}$ is closed.
  Therefore, $A$ is compact, being a closed subspace of the compact set $[\eta,1] \times \Mr[r]$.
  Let
  \begin{align}
      \nu \coloneqq \inf `*{ g(A) } = \inf `*{ \tr`*([qQ]\sigma) \,:\; (q,Q)\in A } \ .
  \end{align}
  $g(A)$ is compact, since it is the image of a compact set under a continuous function.
  Therefore, the infimum is attained, and there exists a $(q, Q) \in A$ such that $g(q, Q) = \nu$.
  Since (i) holds, $\nu \leq \delta + \zeta$ for all $\zeta>0$; hence, $\nu \leq \delta$.
  Thus, there exists $(q,Q)\in A$ such that $g(q, Q) \leq \delta$, so (iii) holds.

  Finally, suppose that $\Psimple$ and $\mathcal{G}$ are compact.
  Consider the continuous function $h(P, \mathcal{E}_1, \ldots, \mathcal{E}_r) \coloneqq `( \mathcal{E}_1^\dagger \cdots \mathcal{E}_r^\dagger )`(P)$.
  $\Mr[r](\Psimple,C_\mathcal{G})$ is compact, since it is the image of a compact set under a continuous function.
\end{proof}

\subsection{Complexity (relative) entropy and state-complexity measures}
\label{appx-topic:complexity-entropy-relationship-to-state-complexity}

The complexity entropy relates to pure-state complexity as follows.
\begin{proposition}[Relation between complexity entropy and pure-state complexity]
  \label{thm:complexity-entropy-pure-state-complexity}
  Let $\Psimple$ denote any set of simple POVM effects (\cref{defn:Psimple-general}); and $C$, any adjoint-invariant unitary-complexity measure (\cref{defn:unitary-complexity}).
  Let $r \geq 0$, and let $\Mr[r] = \Mr[r](\Psimple,C)$ denote the POVM-effect-complexity set defined in~\eqref{eq:Mr-general-from-Psimple-and-superop-cplx-measure}.
  Let $\delta \geq 0$, and let $\ket{\psi_0}$ denote any pure reference state such that $\psi_0 \in \Psimple$.
  Let $C^{\psi_0,\delta}$ denote the $\delta$-approximate-pure-state-complexity measure with respect to $C$ and $\psi_0$ (\cref{defn:approx-state-complexity}).
  Let $\ket\psi$ denote any pure state.
  The following statements hold:
  \begin{enumerate}[label={\arabic*.}]
    \item Suppose that $\delta < 1$ and $r > C^{\psi_0,\,\delta}(\psi)$.
    Then
    \begin{align}
        \HHypr[r][1-\delta^2]{\psi} \leq \log`*( \frac1{1-\delta^2} ) \ .
    \end{align}

    \item Let $r_0 \geq 0$ and $\epsilon \in [0,1)$.
    Suppose that $\HHypr[r][1-\epsilon^2]{\rho} \leq \log`*( 1/`*[ 1-\epsilon^2 ] )$.
    Then
    \begin{align}
      C^{\psi_0,\,\delta}(\psi) &\leq r + r_0 \ ,
    \end{align}
    where
    \begin{align}
      \delta &= \epsilon + \sqrt{ \alpha \epsilon^2 + (1 - \alpha) }\ ;
      &
      \alpha\ &\coloneqq \ \lim_{\zeta\to0^+}\inf_{\substack{P \in \Psimple\\
      \frac{\tr`(P^2)}{[\tr`(P)]^2} \geq 1 - 8\epsilon^2 - \zeta }} \;
      `*{
      \sup_{ \substack{ U_0 :\\ \, C(U_0) \leq r_0} }\;
      `*{
      \dmatrixel{\psi_0}{ U_0 P U_0^\dagger } } } \ .
      \label{eq:complexity-entropy-pure-state-complexity--error-term}
    \end{align}
  \end{enumerate}
\end{proposition}

The error parameter $\alpha$ is necessary for the second statement of the proposition.
$\alpha$ assumes simple values for suitable choices of $\Psimple$, $C$, $\ket{\psi_0}$, and $r_0$.
For instance, consider an $n$-qubit system, and let $\ket{\psi_0} = \ket{0^n}$.
Suppose that $\Psimple$ equals the set of single-qubit projectors in~\eqref{eq:setting-defn-Mrzero}.
For every $P\in\Psimple$, the supremum in~\eqref{eq:complexity-entropy-pure-state-complexity--error-term} achieves its maximum value with $U_0 = \Ident$: $\dmatrixel{0^n}{U_0 P U_0^\dagger} = \braket{0^n}{0^n} = 1$.
[For any $C$ and any $r_0$, $C(\Ident) = 0 \leq r_0$.]
Having unit purity, $P = \proj{0^n}$ always satisfies the constraint on the infimum in~\eqref{eq:complexity-entropy-pure-state-complexity--error-term}, so $\alpha = 1$, and $\delta = 2\epsilon$.
Now, suppose that $\Psimple$ equals the set of tensor products of all single-qubit projectors.
Suppose further that $C$ assigns the value 1 (0) to every single-qubit unitary and that $r_0 = n$ ($r_0 = 0$).
Then, for all $P\in\Psimple$, there exists a tensor product $U_0$ of single-qubit unitaries such that $\dmatrixel{\psi_0}{U_0 P U_0^\dagger} = 1$.
$U_0$ satisfies $C(U_0) \leq r_0$.
Having unit purity, every rank-1 projector in $\Psimple$ always satisfies the constraint on the infimum in~\eqref{eq:complexity-entropy-pure-state-complexity--error-term}; so, again, $\alpha = 1$, and $\delta = 2\epsilon$.

We now relate the complexity relative entropy to the strong complexity of Ref.~\cite{Brandao2021PRXQ_models}.
Let
\begin{align}
  \beta^r`(\rho,\sigma) \coloneqq \max_{M\in \Mr[r]} `*{ \abs`*{ \tr`*( M `*[\rho-\sigma] ) } } \ .
\end{align}
Let $\delta \geq 0$, and let $\ket\psi$ denote a pure state in a Hilbert space of dimensionality $d$.
The \emph{strong complexity} of $\ket\psi$ is defined as
\begin{align}
  C_{\textup{strong}}^{\delta}(\ket\psi)
  \coloneqq \inf `*{ r\geq 0\,:\; \beta^r(\psi,\pi) \geq 1 - \frac1{d} - \delta } \ .
\end{align}
The $\pi \coloneqq \Ident/d$ denotes the maximally mixed state.
Our definition of $\beta^r(\psi,\pi)$ differs slightly from that in Ref.~\cite{Brandao2021PRXQ_models}, wherein $\Mr[r]$ can be defined using an auxiliary system.

The following proposition presents a bound relating the complexity entropy and the strong complexity.
\begin{proposition}[Complexity entropy and strong complexity]
  \label{thm:cplxrelentr-bound-strongcomplexity}
  Let $\rho$ and $\sigma$ denote any quantum states.
  Let $r \geq 0$ and $\eta \in (0,1]$.
  It holds that
  \begin{align}
    \DHypr[r][\eta]{\rho}{\sigma}
    \leq -\log`*( 1- \frac{\beta^r`(\rho,\sigma)}{\eta} )
    = \frac{\beta^r`(\rho,\sigma)}{\eta}
    + O`*[ `*( \frac{\beta^r`(\rho,\sigma)}{\eta} )^2 ] \ .
    \label{eq:cplx-rel-entropy-bound-beta-r}
  \end{align}
  Furthermore, let $\delta \geq 0$, and let $\ket\psi$ denote any pure state in a Hilbert space of dimensionality $d$.
  Suppose that $r < C_{\textup{strong}}^\delta(\ket\psi)$ and $\eta > 1 - d^{-1} - \delta$.
  Then
  \begin{align}
    \HHypr[r][\eta]{\psi}
    &> \log(d) - \log`*(\frac1{1-c})\ ;
    &
    c &\coloneqq \frac{1 - d^{-1} - \delta}{\eta} < 1 \ .  
    \label{eq:cplx-entropy-bound-strong-cplx}
  \end{align}
\end{proposition}

We now prove \cref{thm:complexity-entropy-pure-state-complexity,thm:cplxrelentr-bound-strongcomplexity}.

\begin{proof}[*thm:complexity-entropy-pure-state-complexity]
  We prove the first statement.
  Let $\zeta \coloneqq r - C^{\psi_0,\,\delta}(\psi) > 0$.
  There exists a unitary $U$ such that $\frac12\onenorm{U\psi_0 U^\dagger - \psi} \leq \delta$ and $C(U) \leq r = C^{\psi_0,\,\delta}(\psi) + \zeta$.
  Let $Q \coloneqq U \psi_0 U^\dagger$.
  $Q \in \Mr[r]$, since $\psi_0\in \Psimple$ and $C(U^\dagger) = C(U) \leq r$, by the adjoint invariance of $C$.
  Moreover,
  \begin{align}
    \tr`*( Q\psi )
    = \tr`*( U \psi_0 U^\dagger \psi )
    = \abs{ \matrixel{\psi}{U}{\psi_0} }^2 
    = 1 - `*( \frac12\onenorm{U\psi_0 U^\dagger - \psi} )^2
    \geq 1 - \delta^2
    \ .
  \end{align}
  The last equality expresses a general relation for the trace norm~\cite{BookNielsenChuang2000}:
  for any two pure states $\ket\phi$ and $\ket\varphi$, $\frac12\onenorm{\phi - \varphi} = \sqrt{1 - \abs{ \braket{\varphi}{\phi} }^2}$.
  Thus, $Q$ is a candidate for the $\HHypr[r][1-\delta^2]{\psi}$ optimization, so
  \begin{align}
    \HHypr[r][1-\delta^2]{\psi}
    \leq \log`*(\frac{\tr`*(Q)}{\tr`*(Q\psi)}) \leq \log`*(\frac1{1-\delta^2}) \ .
  \end{align}

  We now prove the second statement.
  Suppose that $\HHypr[r][1-\epsilon^2]{\psi} \leq \log`*( 1/ `*[ 1 - \epsilon^2 ] )$.
  Let $\zeta>0$.
  There exists a $Q\in \Mr[r]$ such that $\tr`(Q\psi) \geq 1-\epsilon^2$ and 
  \begin{align}
    \frac{\tr`(Q)}{\tr`(Q\psi)}
    \leq \ee^{\HHypr[r][1-\epsilon^2]{\psi}} + \zeta \leq \frac1{1-\epsilon^2} + \zeta \ .
    \label{eq:egkjbeio0389uiwefhopofjl}
  \end{align}
  $Q = U^\dagger P U$ for some effect $P \in \Psimple$ and for some unitary $U$ satisfying $C(U) \leq r$.
  Let $\ket{\psi'} \coloneqq U \ket\psi$.
  It holds that
  \begin{align}
    F^2`*(\psi', \frac{P}{\tr(P)}) = \dmatrixel{\psi'}{\frac{P}{\tr(P)}}
    = \frac{\tr(P\psi')}{\tr(P)}  = `*[ \frac{\tr`(Q)}{\tr`(Q\psi)} ]^{-1}
    \geq 1 - \epsilon^2 - \zeta_1 \ .
    \label{eq:fidelity1}
  \end{align}
  Here and in the following, each $\zeta_i \equiv \zeta_i(\zeta) > 0$ is an error parameter such that $\lim_{\zeta\to0^+}\zeta_i = 0$.
  By a Fuchs-van de Graaf inequality~\cite{Fuchs1999IEEETIT_distinguishability},
  \begin{align}
    \frac12\onenorm[\bigg]{ \psi' - \frac{P}{\tr(P)} }
    \leq \sqrt{1 - F^2`*( \psi', \frac{P}{\tr(P)} )}
    \leq \sqrt{ \epsilon^2 + \zeta_1 } \ .
    \label{eq:intermediate-step-trace-norm-1}
  \end{align}

  Let
  \begin{align}
    \alpha_P \coloneqq
    \sup_{ \substack{ U_0 :\\ \,
    C(U_0) \leq r_0} }
    `*{ \dmatrixel{\psi_0}{ U_0 P U_0^\dagger } } \ .
    \label{eq:kjefbhoeiqoajsdijkdfosa}
  \end{align}
  There exists a unitary $U_0$ such that $C(U_0) \leq r_0$ and
  \begin{align}
    \dmatrixel{\psi_0}{ U_0 P U_0^\dagger } \geq \alpha_P  - \zeta \ .
  \end{align}
  Consequently,
  \begin{align}
    F^2`*( U_0^\dagger \psi_0 U_0, \frac{P}{\tr(P)} )
    = \dmatrixel{\psi_0}{ U_0 \,\frac{P}{\tr(P)}\, U_0^\dagger }
    \geq \frac{\alpha_P - \zeta}{\tr`(P)}
    \geq `*(\alpha_P - \zeta) `*(1-\epsilon^2 - \zeta_1)
    = \alpha_P `*(1-\epsilon^2) - \zeta_2 \ .
    \label{eq:fidelity2}
  \end{align}
  The second inequality follows from~\eqref{eq:fidelity1}
  \begin{align}
      \frac1{\tr`(P)} 
      \geq \frac{\tr`(P\psi')}{\tr`(P)}
      \geq 1-\epsilon^2 - \zeta_1 \ .
  \end{align}
  The largest eigenvalue of $P$, $\opnorm{P}$, satisfies $\opnorm{P} \geq \dmatrixel{\psi'}{P} =\dmatrixel{\psi}{Q} \geq 1-\epsilon^2$, so $\tr`*(P^2) \geq \opnorm{P^2} = \opnorm{P}^2 \geq `*(1-\epsilon^2)^2$, because $P$ is positive-semidefinite.
  Hence,
  \begin{align}
    \frac{ \tr`\big(P^2) }{ `\big[\tr(P)]^2 }
    \geq `*(1-\epsilon^2)^2 `*(1-\epsilon^2 - \zeta_1)^2
    = `*(1-\epsilon^2)^4 - \zeta_3
    \geq 1 - 4 \epsilon^2 - 4 `*(\epsilon^2)^3 - \zeta_3
    \geq 1 - 8 \epsilon^2 - \zeta_3 \ ,
  \end{align}
  so
  \begin{align}
    \alpha_P \ \geq\  \alpha_{\zeta_3}
    \coloneqq \inf_{ \substack{P \in \Psimple\\
    \frac{\tr(P^2)}{[\tr(P)]^2} \geq 1 - 8 \epsilon^2 - \zeta_3
    } } `*{ \alpha_P } \ .
  \end{align}
  Thus,
  \begin{align}
    F^2`*( U_0^\dagger \psi_0 U_0 , \frac{P}{\tr(P)} )
    \geq \alpha_{\zeta_3} `*(1-\epsilon^2) - \zeta_2 \ .
  \end{align}
  By a Fuchs-van de Graaf inequality,
  \begin{align}
    \frac12\onenorm[\bigg]{ U_0^\dagger \psi_0 U_0 - \frac{P}{\tr(P)} }
    \leq \sqrt{1 - F^2`*( U_0^\dagger \psi_0 U_0, \frac{P}{\tr(P)} )}
    \leq \sqrt{1 - \alpha_{\zeta_3} `*(1 - \epsilon^2) - \zeta_2 } \ .
    \label{eq:intermediate-step-trace-norm-2}
  \end{align}
  
  Since the trace norm is unitarily invariant and obeys the triangle inequality,~\eqref{eq:intermediate-step-trace-norm-1} and~\eqref{eq:intermediate-step-trace-norm-2} imply that
  \begin{align}
    \frac12\onenorm[\big]{ \psi - U^\dagger U_0^\dagger \psi_0 U_0 U^\dagger }
    = \frac12\onenorm[\big]{ U \psi U^\dagger - U_0^\dagger \psi_0 U_0 }
    \leq \sqrt{\epsilon^2 + \zeta_1} + \sqrt{ 1 - \alpha_{\zeta_3} `*(1 - \epsilon^2) - \zeta_2 }
    \eqqcolon \delta_\zeta \ .
  \end{align}
  By the adjoint invariance of $C$, $C`\big(U^\dagger U_0^\dagger) \leq C`\big(U^\dagger) + C`\big(U_0^\dagger) = C`*(U) + C`*(U_0) \leq r +r_0$, so
  \begin{align}
    C^{\psi_0,\, \delta_\zeta}(\psi) \leq r + r_0 \ .
  \end{align}
  Since $\lim_{\zeta \to 0^+} \delta_\zeta = \delta$ and since $C$ is right-continuous [as per~\eqref{eq:approx-state-complexity-right-continuity-in-delta}],
  \begin{align}
      C^{\psi_0,\,\delta}(\psi)
      = \lim_{\zeta \to 0^+} C^{\psi_0,\,\delta_\zeta}(\psi)
      \leq r + r_0 \ .
  \end{align}
\end{proof}

\begin{proof}[*thm:cplxrelentr-bound-strongcomplexity]
  Consider any $\zeta>0$.
  There exists a $Q \in \Mr[r]$ such that $\tr`(Q\rho)\geq\eta$ and $\tr`(Q\sigma)/\tr`(Q\rho) \leq \exp \bm{(} -\DHypr[r][\eta]{\rho}{\sigma} \bm{)} + \zeta$.
  It holds that
  \begin{align}
    \frac{\tr`*(Q\sigma)}{\tr`*(Q\rho)}
    = 1 - \frac{\tr`*( Q [\rho-\sigma ] )}{\tr`*(Q\rho)}
    \geq 1 - \frac{\beta^r`(\rho,\sigma)}{\tr`*(Q\rho)}
    \geq 1 - \frac{\beta^r`(\rho,\sigma)}{\eta} \ ,
  \end{align}
  so
  \begin{align}
    \ee^{-\DHypr[r][\eta]{\rho}{\sigma}} \geq 1 - \frac{\beta^r`(\rho,\sigma)}{\eta} + \zeta \ ,
  \end{align}
  which implies~\eqref{eq:cplx-rel-entropy-bound-beta-r}, since $\zeta$ is arbitrary.

  We now prove~\eqref{eq:cplx-entropy-bound-strong-cplx}.
  Suppose that $r < C_{\textup{strong}}^{\delta}(\ket\psi)$ and $\eta > 1 - d^{-1} - \delta$.
  $r < C_{\textup{strong}}^{\delta}(\ket\psi)$ implies that $\beta^r(\psi,\pi) < 1 - d^{-1} - \delta < \eta$, so
  \begin{align}
    \HHypr[r][\eta]{\psi}
    = \log`(d) -\DHypr[r][\eta]{\psi}{\pi}
    &\geq \log`(d) + \log`*( 1 - \frac{\beta^r(\psi,\pi)}{\eta} )
    \nonumber\\
    &> \log`(d) + \log`*( 1 - \frac{1 - d^{-1} - \delta}{\eta} )
    = \log`(d) - \log`* (\frac1{1-c} ) \ .
  \end{align}
\end{proof}

\subsection{Reduced complexity (relative) entropy and related quantities}
\label{appx-topic:reduced-complexity-entropy}

\subsubsection{Definition of the reduced complexity (relative) entropy}

Some situations call for a variant of the complexity (relative) entropy---a variant that lacks the denominator in~\eqref{eq:defn-HHypr} [in~\eqref{eq:defn-DHypr}].
We call these variants the \emph{reduced complexity relative entropy} and the \emph{reduced complexity entropy}.
The variants lack a desirable property possessed by the original versions:
the reduced complexity (relative) entropy may diverge as $\eta\to0$, whereas the complexity (relative) entropy will not whenever $\Gamma$ is positive-definite (see \cref{thm:DHypr-trivial-bounds}).
Nevertheless, the variants bear simpler definitions than the originals and more readily facilitate certain technical proofs.
In any case, the reduced complexity (relative) entropy and the complexity (relative) entropy differ by at most $\log(1/\eta)$ (\cref{thm:relation-Dhypr-DHypr}).
Therefore, the quantities are interchangeable if the error tolerance is insignificant: $\eta \approx 1$.

\begin{definition}[Reduced complexity relative entropy]
  \label{defn:Dhypr}
  Let $\rho$ denote any subnormalized state; and $\Gamma$, any positive-semidefinite operator.
  Let $r \geq 0$ and $\eta \in (0,\tr`(\rho)]$.
  Here, $\rho$, $\Gamma$, and every $Q \in \Mr[r]$ act on the same Hilbert space.
  The \emph{reduced complexity-restricted hypothesis-testing relative entropy}, or simply the \emph{reduced complexity relative entropy}, is
  \begin{align}
    \Dhypr[r][\eta]{\rho}{\Gamma}
    \coloneqq -\log\; `*(
    \inf_{\substack{Q \in \Mr[r] \\  \tr`(Q\rho) \geq \eta}}
    `*{ \tr`*(Q\Gamma) }
    )\ .
    \label{eq:defn-Dhypr}
  \end{align}
\end{definition}

\begin{definition}[Reduced complexity entropy]
  \label{defn:Hhypr}
  Let $\rho$ denote any subnormalized state.
  Let $r \geq 0$ and $\eta \in (0, \tr(\rho) ]$.
  The \emph{reduced complexity-restricted hypothesis-testing entropy}, or simply the \emph{reduced complexity entropy}, is
  \begin{align}
    \Hhypr[r][\eta]{\rho}
    &\coloneqq - \Dhypr[r][\eta]{\rho}{\Ident}
    = \log\; `*(
    \inf_{\substack{Q \in \Mr[r] \\ \tr`(Q\rho) \geq \eta}}
    `*{ \tr`*(Q) }
    ) \ .
    \label{eq:defn-Hhypr}
  \end{align}
\end{definition}

By the constraint $\eta \leq \tr(Q\rho) \leq 1$, the reduced complexity (relative) entropy differs from the complexity (relative) entropy by at most $\log(1/\eta)$:
\begin{proposition}[Difference between $\DSym_{\mathrm{h}}^{r,\,\eta}$ and $\DSym_{\mathrm{H}}^{r,\,\eta}$]
  \label{thm:relation-Dhypr-DHypr}
    Let $\rho$ denote any subnormalized state; and $\Gamma$, any positive-semidefinite operator.
    Let $r \geq 0$ and $\eta \in (0, \tr(\rho) ]$.
    It holds that
    \begin{subequations}
    \begin{align}
      0
      \leq \Dhypr[r][\eta]{\rho}{\Gamma} - \DHypr[r][\eta]{\rho}{\Gamma}
      &\leq \log`*( \frac1\eta )
        \label{eq:relation-Dhypr-DHypr} \\
      \text{and} \; \; \;
      0
      \leq \HHypr[r][\eta]{\rho} - \Hhypr[r][\eta]{\rho}
      &\leq \log`*( \frac1\eta ) \ .
        \label{eq:relation-Dhypr-DHypr--Hhypr}
    \end{align}
  \end{subequations}
    If $\rho$ is normalized, $\Dhypr[r][\eta = 1]{\rho}{\Gamma} = \DHypr[r][\eta = 1]{\rho}{\Gamma}$, and $\Hhypr[r][\eta = 1]{\rho} = \HHypr[r][\eta = 1]{\rho}$.
\end{proposition}

\begin{proof}[**thm:relation-Dhypr-DHypr]
    The denominator $\tr(Q\rho)$ in~\cref{def:defn-DHypr} is constrained to $\eta\leq\tr(Q\rho)\leq 1$, leading to the claimed bounds.
    The second part of the proposition follows by setting $\eta=1$ in the bounds.
\end{proof}

The reduced complexity (relative) entropy also obeys bounds similar to those in \cref{thm:DHypr-trivial-bounds}.

\begin{proposition}[General bounds for $\DSym_{\mathrm{h}}^{r,\,\eta}$]
   \label{thm:Dhypr-trivial-bounds}
  Let $\rho$ denote any subnormalized state, and $\Gamma$ any positive-semidefinite operator, that act on a Hilbert space of dimensionality $d$.
  Let $r \geq 0$ and $\eta \in (0, \tr(\rho) ]$.
  It holds that
  \begin{subequations}
    \begin{align}
      \Dhypr[r][\eta]{\rho}{\Gamma} &\geq -\log \bm{(} \tr`(\Gamma) \bm{)} \ .
        \label{eq:thm-Dhypr-trivial-bounds--lower-bound}
    \end{align}
    Furthermore, if $\Gamma$ is positive-definite (has full rank),
    \begin{align}
      \Dhypr[r][\eta]{\rho}{\Gamma} &\leq \log`*( \opnorm{\Gamma^{-1}} ) + \log \bm{(} \tr`(\rho) \bm{)} + \log`*( \frac1\eta ) \ .
        \label{eq:thm-Dhypr-trivial-bounds--upper-bound}
    \end{align}
  \end{subequations}
  In particular,
  \begin{align}
      - \log \bm{(} \tr`(\rho) \bm{)} - \log`*( \frac1\eta ) 
      \leq \Hhypr[r][\eta]{\rho}
      \leq \log(d) \ .
       \label{eq:thm-Dhypr-trivial-bounds--HHypr}
  \end{align}
  Consequently, if $\rho$ is normalized and $\Gamma$ is positive-definite,
  \begin{subequations}
    \begin{align}
      -\log \bm{(} \tr`(\Gamma) \bm{)}
      \leq \Dhypr[r][\eta]{\rho}{\Gamma}
        &\leq \log `*( \opnorm{\Gamma^{-1}} ) + \log`*( \frac1\eta ) \ ,
        \label{eq:thm-Dhypr-trivial-bounds--special-case}
      \\
      \text{and} \; \; \;
      -\log`*( \frac1\eta ) \leq \Hhypr[r][\eta]{\rho}
      &\leq \log(d) \ .
      \label{eq:thm-Dhypr-trivial-bounds--special-case--HHypr}
    \end{align}
  \end{subequations}
\end{proposition}

\begin{proof}[**thm:Dhypr-trivial-bounds]
  \eqref{eq:thm-Dhypr-trivial-bounds--lower-bound} follows because $\Ident \in \Mr[r]$ is a candidate effect for the $\Dhypr[r][\eta]{\rho}{\Gamma}$ optimization.
  \eqref{eq:thm-Dhypr-trivial-bounds--upper-bound} follows from combining~\eqref{eq:thm-DHypr-trivial-bounds--upper-bound} of \cref{thm:DHypr-trivial-bounds} with~\eqref{eq:relation-Dhypr-DHypr} of \cref{thm:relation-Dhypr-DHypr}.
\end{proof}
\noindent
\eqref{eq:thm-Dhypr-trivial-bounds--lower-bound} implies that $\Dhypr[r][\eta]{\rho}{\sigma} \geq 0$ for normalized states $\rho$ and $\sigma$.

The reduced complexity (relative) entropy possesses several elementary properties of the complexity (relative) entropy.

\begin{proposition}[Properties of $\DSym_{\mathrm{h}}^{r,\,\eta}$]
  \label{thm:properties-Dhypr}
    The following propositions hold if, everywhere therein, one replaces the complexity (relative) entropy with the reduced complexity (relative) entropy:
    \cref{thm:DHypr-monotonous-in-r-and-in-eta}, \cref{thm:DHypr-scaling-2ndarg}, \cref{thm:DHypr-monotonic-under-ordering}, \cref{thm:DHypr-subadditivity}, \cref{thm:cplxrelentropy-monotonicity-partial-trace}, and \cref{thm:DHypr-arg-U-rp}.
\end{proposition}
\begin{proof}[**thm:properties-Dhypr]
    Adapt the propositions' proofs as appropriate.
\end{proof}
\noindent
In general, the reduced complexity (relative) entropy lacks the properties stated in \cref{thm:DHypr-zero-same-args} and in \cref{thm:DHypr-bounded-by-DHyp}.
\cref{thm:DHypr-zero-same-args} does not hold because
$\Dhypr[r][\eta]{\rho}{\rho} > 0$ if $\Mr[r] = \{\Ident\}$ and $\tr`(\Ident) > 1$.
\cref{thm:DHypr-bounded-by-DHyp} does not hold because $\Hhypr[r][\eta]{\rho}$, but not $\HHyp[\eta]{\rho}$, can be negative.

\subsubsection{Complexity success probability}

It is convenient to introduce a quantity closely related to the reduced complexity entropy.
We call this quantity the \emph{complexity success probability}.
Given any subnormalized state $\rho$ and any $r \geq 0$, we have defined the reduced complexity entropy as the infimum of $\log \bm{(} \tr`(Q) \bm{)}$ over all $Q \in \Mr[r]$, subject to a constraint (a lower bound) on $\tr`(Q\rho)$, the probability of successfully identifying $\rho$.
Reciprocally, we define the complexity success probability as the supremum of $\tr`(Q\rho)$ over all $Q \in \Mr[r]$, subject to a constraint (an upper bound) on $\log \bm{(} \tr`(Q) \bm{)}$.

\begin{definition}[Complexity success probability]
  \label{defn:cplx-entropy-success-probability}
  Let $\rho$ denote any subnormalized state.
  Let $r\geq 0$.
  Let $m_- \coloneqq \inf_\eta `*{ \Hhypr[r][\eta]{\rho} }$ and $m_+ \coloneqq \sup_\eta `*{ \Hhypr[r][\eta]{\rho} } = \Hhypr[r][\eta = \tr`(\rho)]{\rho}$, with $\eta$ ranging over $(0,\tr`(\rho)]$.
  Let $m \in [m_-, m_+]$ if $m_-$ is the minimum of $`*{ \Hhypr[r][\eta]{\rho} }_\eta$; otherwise, let $m \in (m_-, m_+]$.
  The \emph{reduced-complexity-entropy success probability}, or simply the \emph{complexity success probability}, is
  \begin{align}
    \etahypr[r][m]{\rho}
    \coloneqq \sup `*{
      \eta\,:\; \Hhypr[r][\eta]{\rho} \leq m
      }
    = \sup_{ \substack{Q \in \Mr[r] \\ 
    \log\bm{(} \tr`(Q) \bm{)}
    \leq m} }
    `*{ \tr`(Q\rho) } \ .
  \label{eq:defn-cplx-entropy-success-probability}
\end{align}
\end{definition}
\noindent
$\etahypr[r][m]{\rho}$ appears in the proof of \cref{thm:ComplexityEntropyLowerBoundApproximateKDesignSimple} (\cref{appx-topic:complexity-entropy-random-circuits}).

$\etahypr[r][m]{\rho}$ takes a simple form in the special case where $\Mr[r]$ is defined for $n$ qubits, as in~\eqref{eq:setting-defn-Mr}.
Consider any candidate $Q \in \Mr[r]$.
$Q = U^\dagger P U$ for some unitary $U$ and for some effect $P$ that, up to a permutation of qubits, equals $\Ident_2^{\otimes k_Q} \otimes \proj{0^{n-k_Q}}$, wherein $k_Q \coloneqq \log_2 \bm{(} \tr`(Q) \bm{)}$.
$\log \bm{(} \tr`(Q) \bm{)} \leq m$ precisely if $k_Q \leq m / \log(2)$, which holds precisely if $k_Q \leq \lfloor{m / \log(2)}\rfloor$.
Moreover, if $k_Q < \lfloor{m / \log(2)}\rfloor$, then there exists a $\tilde{Q} \in \Mr[r]$ such that $k_{\tilde{Q}} = \lfloor{m / \log(2)}\rfloor$ and $\tr`(\tilde{Q}\rho) \geq \tr`(Q\rho)$.
$\big($Indeed, let $\tilde{Q} \coloneqq U^\dagger \tilde{P} U$, wherein $\tilde{P}$, up to the aforementioned permutation, equals $\Ident_2^{ \otimes k_{\tilde{Q}} } \otimes \proj{ 0^{n-k_{\tilde{Q}}} }$.$\big)$
Thus, we can replace the constraint $k_Q \leq \lfloor{m / \log(2)}\rfloor$ with $k_Q = \lfloor{m / \log(2)}\rfloor$:
\begin{align}
    \etahypr[r][m]{\rho}
    = \sup_{ \substack{Q \in \Mr[r] \\ 
    \log_2\bm{(} \tr`(Q) \bm{)} = \lfloor{m / \log(2)}\rfloor } }
    `*{ \tr`(Q\rho) } \ .
  \label{eq:cplx-success-probability--special-case}
\end{align}

The complexity success probability monotonically increases as $r$ increases and monotonically increases as $m$ increases.
\begin{proposition}[Monotonicity of $\etahypr{}$ in $r$ and $m$]
  \label{thm:success-probability-monotonic-in-r-and-in-eta}
  Let $\rho$ denote any subnormalized state.
  Let $r \geq 0$.
  Let $m$ and $m_+$ be as in \cref{defn:cplx-entropy-success-probability}.
  \begin{enumerate}
    \item For all $r' \geq r$, $\etahypr[r'][m]{\rho} \geq \etahypr[r][m]{\rho}$.
    \item For all $m' \in [m, m_+]$, $\etahypr[r][m']{\rho} \geq \etahypr[r][m]{\rho}$.
  \end{enumerate}
\end{proposition}
\begin{proof}[**thm:success-probability-monotonic-in-r-and-in-eta]
  Consider the definition~\eqref{eq:defn-cplx-entropy-success-probability} of $\etahypr[r][m]{\rho}$.
  Consider any $\eta \in (0, \tr`(\rho)]$.
  The monotonicity in $r$ follows because $\Hhypr[r][\eta]{\rho} \leq m$ implies that $\Hhypr[r'][\eta]{\rho} \leq m$, by \cref{thm:DHypr-monotonous-in-r-and-in-eta} (via \cref{thm:properties-Dhypr}).
  The monotonicity in $m$ follows because $\Hhypr[r][\eta]{\rho} \leq m$ implies that $\Hhypr[r][\eta]{\rho} \leq m'$.
\end{proof}

Moreover, the complexity success probability is convex.

\begin{proposition}[Convexity of $\etahypr{}$]
  \label{thm:etahypr-convexity}
  Let $`{ p_k }$ denote any probability distribution; and $`{ \rho_k }$, any associated collection of 
  subnormalized states.
  Let $r \geq 0$.
  Let $m$ be as in \cref{defn:cplx-entropy-success-probability}.
  It holds that
  \begin{align}
    \etahypr[r][m]`*{ \sum_k p_k \rho_k } \leq \sum_k p_k \, \etahypr[r][m]{ \rho_k }\ .
    \label{eq:etahypr-convexity}
  \end{align}
\end{proposition}

\begin{proof}[**thm:etahypr-convexity]
  Consider any $\zeta > 0$.
  There exists a $Q \in \Mr[r]$ such that $\log\bm{(} \tr`(Q) \bm{)} \leq m$ and
  \begin{align}
      \etahypr[r][m]`*{ \sum_k p_k \rho_k } - \zeta
      \leq \tr`*( Q `*[ \sum_k p_k \rho_k ] )
      = \sum_k p_k \, \tr`(Q \rho_k)
      \leq \sum_k p_k \, \etahypr[r][m]{\rho_k} \ .
  \end{align}
  The second inequality holds because, for each $k$, $Q$ is a candidate for the $\etahypr[r][m]{\rho_k}$ optimization and thus satisfies $\tr`(Q \rho_k) \leq \etahypr[r][m]{\rho_k}$.
  \eqref{eq:etahypr-convexity} now follows because $\zeta$ is arbitrary.
\end{proof}

\subsubsection{Interrelating reduced complexity (relative) entropies under two computational models}

Under certain conditions, a simple bound relates reduced complexity (relative) entropies defined with respect to different computational models.

\begin{proposition}[Reduced complexity relative entropies under two computational models]
  \label{thm:cplx-rel-entr-bridging-complexity-measures}
  Let $\rho$ denote any subnormalized state; and $\Gamma$, any positive-semidefinite operator.
  Let $\Psimple$ denote any set of simple POVM effects (\cref{defn:Psimple-general}), and let $C^{(1)}$ and $C^{(2)}$ denote superoperator-complexity measures (\cref{defn:superoperator-complexity-general}).
  For each $i = 1,2$, let $\big\{ \Mr[r ,\, (i)] \coloneqq \Mr[r] \big( \Psimple,C^{(i)} \big) \big\}$ denote the family of POVM-effect-complexity sets defined by~\eqref{eq:Mr-general-from-Psimple-and-superop-cplx-measure}.
  Let $r \geq 0$, $\alpha > 0$, $\epsilon \in [0,1]$, and $\eta \in (\epsilon, 1]$.
  Assume there exists a monotonically increasing function $f : \mathbb{R}_+ \to \mathbb{R}_+$ with the following property:
  for every operation $\mathcal{E}$ satisfying $C^{(1)}(\mathcal{E}) < \infty$, there exists an operation $\mathcal{F}$ such that
  \begin{align}
      \frac12\dianorm{\mathcal{E} - \mathcal{F}} \leq \epsilon \ , \quad
      \mathcal{F}(\Gamma) \leq \alpha \mathcal{E}(\Gamma) \ ,
      \quad \text{and} \quad
      C^{(2)}(\mathcal{F}) \leq f \bm{(} C^{(1)}(\mathcal{E}) \bm{)} \ .
      \label{eq:cplx-rel-entr-bridging-complexity-measures--assmpt-E-F}
  \end{align}
  It holds that
  \begin{align}
      \Dhypr[f(r)][\eta-\epsilon][(2)]{\rho}{\Gamma}
      \geq \Dhypr[r][\eta][(1)]{\rho}{\Gamma} - \log`(\alpha) \ .
      \label{eq:cplx-rel-entr-bridging-complexity-measures}
  \end{align}
  The reduced complexity relative entropies $\Dhypr[r][\eta][(1)]{}{}$ and $\Dhypr[f(r)][\eta-\epsilon][(2)]{}{}$ are defined with respect to $\Mr[r ,\, (1)]$ and $\Mr[f(r) ,\, (2)]$, respectively.
\end{proposition}

\begin{corollary}[Reduced complexity entropies under two unitary computational models]
  \label{thm:cplx-entr-bridging-unitary-complexity-measures}
  Let $\rho$ denote any subnormalized state.
  Let $\Psimple$ denote any set of simple POVM effects (\cref{defn:Psimple-general}), and let $C^{(1)}$ and $C^{(2)}$ denote unitary-complexity measures (\cref{defn:unitary-complexity}).
  For each $i =1,2$, let $\big\{ \Mr[r ,\, (i)] \coloneqq \Mr[r] \big( \Psimple,C^{(i)} \big) \big\}$ denote the family of POVM-effect-complexity sets defined by~\eqref{eq:Mr-general-from-Psimple-and-superop-cplx-measure}.
  Let $r \geq 0$, $\epsilon \in [0,1]$, and $\eta \in (\epsilon, 1]$.
  Assume there exists a monotonically increasing function $f : \mathbb{R}_+ \to \mathbb{R}_+$ with the following property:
  for every unitary $U$ satisfying $C^{(1)}(U) < \infty$, there exists a unitary $V$ such that
  \begin{align}
      \norm{U - V}
      \leq \epsilon
      \quad \text{and} \quad
      C^{(2)}(V)
      \leq f \bm{(} C^{(1)}(U) \bm{)} \ .
  \end{align}
  It holds that
  \begin{align}
      \Hhypr[f(r)][\eta-\epsilon][(2)]{\rho}
      \leq \Hhypr[r][\eta][(1)]{\rho}\ .
      \label{eq:cplx-entr-bridging-unitary-complexity-measures}
  \end{align}
  The reduced complexity entropies $\Hhypr[r][\eta][(1)]{}{}$ and $\Hhypr[f(r)][\eta-\epsilon][(2)]{}{}$ are defined with respect to $\Mr[r ,\, (1)]$ and $\Mr[f(r) ,\, (2)]$, respectively.
\end{corollary}

\begin{proof}[*thm:cplx-rel-entr-bridging-complexity-measures]
  Consider any $\zeta>0$.
  There exists a $Q \in \Mr[r ,\, (1)]$ such that $\Dhypr[r][\eta][(1)]{\rho}{\Gamma} \leq -\log \bm{(} \tr`(Q\Gamma) \bm{)} + \zeta$.
  $Q = \mathcal{E}^\dagger(P)$ for some effect $P \in \Psimple$ and for some operation $\mathcal{E}$ satisfying $C^{(1)}(\mathcal{E})\leq r$.
  By assumption, there exists an operation $\mathcal{F}$ satisfying the conditions~\eqref{eq:cplx-rel-entr-bridging-complexity-measures--assmpt-E-F}.
  Hence, $C^{(2)}(\mathcal{F}) \leq f \bm{(} C^{(1)}(\mathcal{E}) \bm{)} \leq f(r)$.
  The second inequality follows because $f$ is monotonically increasing.
  Let $Q' \coloneqq \mathcal{F}^\dagger(P) \in \Mr[f(r) ,\, (2)]$.
  It holds that
  \begin{align}
    \tr`(Q'\rho)
    = \tr \bm{(} P \mathcal{F}(\rho) \bm{)}
    \geq \tr \bm{(} P \mathcal{E}(\rho) \bm{)} - \epsilon
    = \tr`(Q\rho) - \epsilon
    \geq \eta - \epsilon \ .
  \end{align}
  The first inequality follows because $\tr \bm{(} P \mathcal{F}(\rho) \bm{)} \geq \tr \bm{(} P \mathcal{E}(\rho) \bm{)} - \frac12\onenorm{\mathcal{E}(\rho) - \mathcal{F}(\rho)}$, by Corollary 9.1.1 of Ref.~\cite{BookWilde2013QIT}, and because $\frac12\dianorm{\mathcal{E} - \mathcal{F}} \leq \epsilon$ implies that $\frac12\onenorm{\mathcal{E}(\rho) - \mathcal{F}(\rho)} \leq \epsilon$, by~\eqref{eq:dianorm-upper-bounds-onenorm}.
  Therefore, $Q'$ is a candidate for the $\Dhypr[f(r)][\eta-\epsilon][(2)]{\rho}{\Gamma}$ optimization.
  Moreover, $\tr`(Q'\Gamma) = \tr \bm{(} P \mathcal{F}(\Gamma) \bm{)} \leq \tr \bm{(} P `*[\alpha \mathcal{E}(\Gamma)] \bm{)} = \alpha \tr`( Q\Gamma )$, so
  \begin{align}
    \Dhypr[f(r)][\eta-\epsilon][(2)]{\rho}{\Gamma}
    \geq -\log \bm{(} \tr`( Q'\Gamma ) \bm{)}
    \geq - \log \bm{(} \tr`(Q\Gamma) \bm{)} -\log`(\alpha)
    \geq \Dhypr[r][\eta][(1)]{\rho}{\Gamma} -\log`(\alpha) - \zeta \ ,
  \end{align}
  which implies~\eqref{eq:cplx-rel-entr-bridging-complexity-measures}, since $\zeta$ is arbitrary.
\end{proof}

\begin{proof}[*thm:cplx-entr-bridging-unitary-complexity-measures]
  \eqref{eq:cplx-entr-bridging-unitary-complexity-measures} follows from applying \cref{thm:cplx-rel-entr-bridging-complexity-measures} when $\Gamma = \Ident$ and $\alpha = 1$.
  We need to show only that the conditions~\eqref{eq:cplx-rel-entr-bridging-complexity-measures--assmpt-E-F} hold in this case.
  Let $\Gamma = \Ident$ and $\alpha = 1$.
  Consider any operation $\mathcal{E}$ satisfying $C^{(1)}(\mathcal{E}) < \infty$.
  $\mathcal{E}$ is a unitary operation, since $C^{(1)}$ is a unitary-complexity measure.
  Hence, $\mathcal{E}(\cdot) = U`(\cdot) U^\dagger$ for some unitary operator $U$ satisfying $C^{(1)}(U) < \infty$.
  By assumption, there exists a unitary operator $V$ such that $\opnorm{U - V}\leq \epsilon$ and $C^{(2)}(V) \leq f \bm{(} C^{(1)}(U) \bm{)}$.
  Let $\mathcal{F}(\cdot) \coloneqq V`(\cdot) V^\dagger$.
  By \cref{thm:diamond-norm-unitary-to-ident}, $\frac12 \dianorm{\mathcal{E} - \mathcal{F}} \leq \opnorm{U - V} \leq \epsilon$.
  Also, $\mathcal{F}(\Gamma) \leq \alpha \mathcal{E}(\Gamma)$, since $\mathcal{F}(\Ident) = \Ident = \mathcal{E}(\Ident)$.
  Last, $C^{(2)}(\mathcal{F}) \leq f \bm{(} C^{(1)}(\mathcal{E}) \bm{)}$ is a restatement of $C^{(2)}(V) \leq f \bm{(} C^{(1)}(U) \bm{)}$.
\end{proof}

\subsection{Complexity conditional entropy}
\label{appx-topic:complexity-conditional-entropy}

We here define the \emph{complexity conditional entropy} in terms of the complexity relative entropy, just as one defines one-shot conditional entropies in terms of one-shot relative entropies~\cite{BookTomamichel2016_Finite}.
In \cref{sec:decoupling}, the complexity conditional entropy quantifies the maximally mixed qubits one can decouple from a reference system.

\begin{definition}[Complexity conditional entropy]
  \label{defn:complexity-conditional-entropy}
  Let $A$ and $B$ denote distinct quantum systems.
  Let $\rho_{AB}$ denote any subnormalized state of $AB$.
  Let $r \geq 0$ and $\eta \in (0,\tr(\rho)]$.
  The \emph{complexity conditional entropy} of $A$ conditioned on $B$ is
  \begin{align}
    \HHyprc[r][\eta][\rho]{A}[B] \coloneqq -\DHypr[r][\eta]{\rho_{AB}}{\Ident_A\otimes\rho_B}\ .
  \end{align}
\end{definition}

\begin{proposition}[General bounds for the complexity conditional entropy]
  \label{thm:genbounds-cplx-cond-entropy}
  Let $A$ and $B$ denote distinct quantum systems.
  Let $\rho_{AB}$ denote any subnormalized state of $AB$.
  Let $r \geq 0$ and $\eta \in (0,\tr(\rho)]$.
  It holds that
  \begin{align}
    -\log(d_A) \leq \HHyprc[r][\eta][\rho]{A}[B] \leq \log(d_A) \ .
    \label{eq:genbounds-cplx-cond-entropy--HHyprc}
  \end{align}
  Equivalently,
  \begin{align}
    0 \leq \DHypr[r][\eta]{\rho_{AB}}{\pi_A \otimes \rho_B} \leq 2 \log(d_A) \ .
    \label{eq:genbounds-cplx-cond-entropy--DHypr}
  \end{align}
\end{proposition}
\begin{proof}[**thm:genbounds-cplx-cond-entropy]
  $\HHyprc[r][\eta][\rho]{A}[B] \leq \log(d_A)$ follows from applying~\eqref{eq:thm-DHypr-trivial-bounds--lower-bound} in \cref{thm:DHypr-trivial-bounds} with $\Gamma = \Ident_A\otimes\rho_B$:
  \begin{align}
    -\DHypr[r][\eta]{\rho_{AB}}{\Ident_A\otimes\rho_B}
    \leq \log \bm{(} \tr`(\Ident_A\otimes\rho_B) \bm{)} - \log \bm{(} \tr`(\rho_{AB}) \bm{)}
    = \log \bm{(} d_A \tr`(\rho_B) \bm{)} - \log \bm{(} \tr`(\rho_B) \bm{)}
    = \log(d_A) \ .
  \end{align}
  The first equality follows because $\tr(\rho_{AB}) = \tr_B \bm{(} \tr_A(\rho_{AB}) \bm{)} = \tr(\rho_B)$.
  $\HHyprc[r][\eta][\rho]{A}[B] \geq -\log(d_A)$ follows from applying \cref{thm:DHypr-monotonic-under-ordering} to the inequality $d_A^{-1}\rho_{AB} \leq \Ident_A\otimes\rho_B$ (\cref{thm:state-partial-order-lemma}):
  \begin{align}
    -\DHypr[r][\eta]{\rho_{AB}}{\Ident_A\otimes\rho_B}
    \geq -\DHypr[r][\eta]{\rho_{AB}}{d_A^{-1}\rho_{AB}}
    = -\DHypr[r][\eta]{\rho_{AB}}{\rho_{AB}} - \log(d_A)
    = -\log(d_A)\ .
  \end{align}
  The first equality follows from \cref{thm:DHypr-scaling-2ndarg} and the second equality from \cref{thm:DHypr-zero-same-args}.

  Last,~\eqref{eq:genbounds-cplx-cond-entropy--HHyprc} and~\eqref{eq:genbounds-cplx-cond-entropy--DHypr} are equivalent because \cref{thm:DHypr-scaling-2ndarg} implies that $\HHyprc[r][\eta][\rho]{A}[B] = \log(d_A) - \DHypr[r][\eta]{\rho_{AB}}{\pi_A\otimes\rho_B}$:
  \begin{align}
      \HHyprc[r][\eta][\rho]{A}[B]
      = -\DHypr[r][\eta]{\rho_{AB}}{\Ident_A\otimes\rho_B}
      = -\DHypr[r][\eta]{\rho_{AB}}{d_A\pi_A\otimes\rho_B}
      = \log(d_A) - \DHypr[r][\eta]{\rho_{AB}}{\pi_A\otimes\rho_B} \ .
      \label{eq:alternative-form-of-HHyprc}
  \end{align}
\end{proof}

\begin{proposition}[Strong subadditivity of the complexity conditional entropy]
  \label{thm:cond-cplx-entropy-strong-subadditivity}
  Let $A$, $B$, and $C$ denote distinct quantum systems.
  Let $\rho_{ABC}$ denote any subnormalized state of $ABC$.
  Let $r \geq 0$ and $\eta \in (0,\tr(\rho)]$.
  It holds that
  \begin{align}
    \HHyprc[r][\eta][\rho]{A}[BC]
    \leq \HHyprc[r][\eta][\rho]{A}[B]\ .
  \end{align}
\end{proposition}
\begin{proof}[**thm:cond-cplx-entropy-strong-subadditivity]
  By \cref{thm:cplxrelentropy-monotonicity-partial-trace}, the complexity relative entropy never increases under a partial trace, so
  \begin{align}
    \HHyprc[r][\eta][\rho]{A}[BC]
    = -\DHypr[r][\eta]{\rho_{ABC}}{\Ident_A\otimes\rho_{BC}}
    \leq -\DHypr[r][\eta]{\rho_{AB}}{\Ident_A\otimes\rho_{B}}
    = \HHyprc[r][\eta][\rho]{A}[B] \ .
  \end{align}
\end{proof}

\section{Thermodynamic erasure of qubits governed by a product Hamiltonian}
\label{appx-topic:erasure-nontrivial-Ham}
\label{appx-topic:erasure-exact-expression-reduced-Hhypr} 

Here, we generalize the arguments in \cref{sec:main-erasure-Wcost-degenerate-H} to prove \cref{thm:not-necessarily-degenerate-Hamiltonians} (\cref{sec:erasure-nontrivial-H}).
Recall the theorem's setting.
We consider a system of $n$ noninteracting qubits.
Qubit $i$ evolves under a Hamiltonian $H_i$ that has zero ground-state energy, $H_i\ket{0}_i =0$.
We fix an inverse temperature $\beta>0$.
Qubit $i$ has the Gibbs-weight operator $\Gamma_i = \ee^{-\beta H_i}$; and the $n$-qubit system, the Gibbs-weight operator $\Gamma = \bigotimes_{i=1}^n \Gamma_i = \ee^{-\beta \sum_i H_i}$.

Let $\mathcal{T}$ denote a set of elementary computations on the $n$-qubit system.
Each computation in $\mathcal{T}$ is a completely positive, trace-preserving map that sends $\Gamma$ to itself and costs no work to implement.
After implementing elementary computations, one can apply \textsc{reset} operations to a selection $\mathcal{W} \subset `{1, 2, \ldots, n}$ of qubits.
Each \textsc{reset} operation $\mathcal{E}_{\textsc{reset},i}$ initializes qubit $i$ in the state $\ket0_i$ and costs an amount of work $W_{\textsc{reset},i} = \beta^{-1} \log \bm{(} \tr`(\Gamma_i) \bm{)} \geq 0$.
Consider a general protocol $\mathcal{E}$ consisting of $r \geq 0$ elementary computations followed by \textsc{reset} operations on the qubits of $\mathcal{W}$.
$\mathcal{E}$ has the form $\mathcal{E} = `*( \textstyle\prod_{i\in\mathcal{W}} \mathcal{E}_{\textsc{RESET},i} ) \mathcal{E}_r\cdots\mathcal{E}_2\mathcal{E}_1$.
Each $\mathcal{E}_i$ is in $\mathcal{T}$, and the work cost $W(\mathcal{E}) = \sum_{i \in \mathcal{W}} W_{\textsc{reset},i} = \beta^{-1} \log \bm{(} \prod_{i \in \mathcal{W}}   \tr`(\Gamma_i) \bm{)}$.

\begin{theorem}
  \label{thm:erasure-nontrivial-Ham-Wrstar}
  Let $\rho$ denote any quantum state of $n$ qubits.
  Let $\Psimple$ denote the set of simple POVM effects (\cref{defn:Psimple-general}) defined in~\eqref{eq:setting-defn-Mrzero}.
  Let $C_{\mathcal{T}}$ denote the circuit-complexity measure associated with $\mathcal{T}$ (\cref{defn:circuit-complexity-measure}).
  Let $r \geq 0$ denote an integer, and let $\eta \in (0,1]$.
  Let $\Mr[r] = \Mr[r]`*(\Psimple,C_{\mathcal{T}})$ denote the POVM-effect-complexity set defined in~\eqref{eq:Mr-general-from-Psimple-and-superop-cplx-measure}.
  Let
  \begin{align}
      W_r^* \coloneqq 
      \min `*{
        W(\mathcal{E}) \,:\ 
        \mathcal{E} =
        `*( \textstyle\prod_{i\in\mathcal{W}}\mathcal{E}_{\textsc{RESET},i} ) \mathcal{E}_r\cdots\mathcal{E}_2\mathcal{E}_1 \,,\; 
      \mathcal{W} \subset `{1, 2, \ldots, n} \,,\; 
        F^2 \bm{(} \mathcal{E}(\rho),\proj{0^n} \bm{)} \geq \eta
        } \ .
    \label{eq:erasure-nontrivial-Ham-Wrstar---Wrstar}
  \end{align}
  Each $\mathcal{E}_i$ is in $\mathcal{T}$, and $F$ denotes the fidelity between quantum states.
  It holds that
  \begin{align}
    \beta W_r^* = -\Dhypr[r][\eta][\mathcal{T}]{\rho}{\Gamma}\ ,
    \label{eq:erasure-nontrivial-Ham-Wrstar---Dhypr}
  \end{align}
  wherein the reduced complexity relative entropy $\Dhypr[r][\eta][\mathcal{T}]{}{}$ (\cref{defn:Dhypr}) is defined with respect to $\Mr[r]$.
  Consequently,
  \begin{align}
    - \DHypr[r][\eta][\mathcal{T}]{\rho}{\Gamma} -\log `*( \frac1\eta )
    \leq \beta W_r^*
    \leq -\DHypr[r][\eta][\mathcal{T}]{\rho}{\Gamma} \ .
    \label{eq:erasure-nontrivial-Ham-Wrstar---DHypr}
  \end{align}
\end{theorem}

Consider the special case where each Hamiltonian $H_i$ is degenerate: $H_i = 0$.
In this case, $\Gamma = \Ident$.
By~\eqref{eq:erasure-nontrivial-Ham-Wrstar---Dhypr}, $\beta W_r^*$ equals the reduced complexity entropy (\cref{defn:Hhypr}): $\beta W_r^* = -\Dhypr[r][\eta][\mathcal{T}]{\rho}{\Ident} = \Hhypr[r][\eta][\mathcal{T}]{\rho}$. 
If, furthermore, $\mathcal{T}$ contains only unitary operations, then~\eqref{eq:erasure-nontrivial-Ham-Wrstar---DHypr} yields the bounds in \cref{mainthm:erasure-work-cost-complexity-entropy-simple} (\cref{sec:main-erasure-Wcost-degenerate-H}).

\begin{proof}[*thm:erasure-nontrivial-Ham-Wrstar]
  We first prove that $\beta W_r^* \geq -\Dhypr[r][\eta][\mathcal{T}]{\rho}{\Gamma}$.
  Consider any operation $\mathcal{E}$ that achieves the minimum in~\eqref{eq:erasure-nontrivial-Ham-Wrstar---Wrstar} and thus satisfies $ W(\mathcal{E}) = W_r^*$.
  $\mathcal{E} = `*(\textstyle\prod_{i\in\mathcal{W}} \mathcal{E}_{\textsc{RESET},i}) \mathcal{E}_r\cdots\mathcal{E}_2\mathcal{E}_1$ for some subset $\mathcal{W} \subset `{1, 2, \ldots, n}$ and some operations $\mathcal{E}_i \in \mathcal{T}$.
  Moreover, $F^2 \bm{(} \mathcal{E}(\rho),\proj{0^n} \bm{)} \geq \eta$.
  Let $P \coloneqq \Ident_{\mathcal{W}} \otimes \proj{ 0^{ n - \abs{\mathcal{W}} } }_{\mathcal{W}^{\rm c}} \in \Psimple$, wherein $\mathcal{W}^{\rm c}$ denotes the set complement of $\mathcal{W}$.
  Let $Q \coloneqq \mathcal{E}_1^\dagger \mathcal{E}_2^\dagger \cdots \mathcal{E}_r^{\dagger} `(P) \in \Mr[r]$.
  It holds that
  \begin{align}
    \tr`(Q\rho)
    &= \tr \bm{(} P\,\mathcal{E}_r\cdots\mathcal{E}_2\mathcal{E}_1(\rho) \bm{)}
    = \tr`*{ \proj{0^{n-\abs{\mathcal{W}}}}_{\mathcal{W}^{\rm c} } \tr_{\mathcal{W}} \bm{(} \mathcal{E}_r\cdots\mathcal{E}_2\mathcal{E}_1(\rho) \bm{)} }
    \nonumber \\
    &= \tr \bm{(} \proj{0^n}\,\mathcal{E}(\rho) \bm{)}
    = F^2 \bm{(} \mathcal{E}(\rho),\proj{0^n} \bm{)} \ .
    \label{eq:erasure-nontrivial-Ham-Wrstar---candidate-constraint}
  \end{align}
  The third equality follows because $\mathcal{E}(\rho) = \proj{ 0^{ \abs{\mathcal{W}} } }_{\mathcal{W}} \otimes \tr_{\mathcal{W}} \bm{(} \mathcal{E}_r\cdots\mathcal{E}_2\mathcal{E}_1 (\rho) \bm{)}$.
  Hence, $\tr(Q\rho) \geq \eta$, so $Q$ is a candidate for the $\DHypr[r][\eta][\mathcal{T}]{\rho}{\Gamma}$ optimization.
  Furthermore,
  \begin{align}
    \tr`(Q\Gamma)
    = \tr \bm{(} P\,\mathcal{E}_r\cdots\mathcal{E}_1(\Gamma) \bm{)}
    = \tr`(P \Gamma)
    =\tr`*( `\Bigg[ \prod_{i \in \mathcal{W}} \Gamma_i ] \otimes \proj{0^{n-\abs{\mathcal{W}}}}_{\mathcal{W}^{\rm c} } )
    = \prod_{i \in \mathcal{W}} \tr`(\Gamma_i)
    = \ee^{ \beta W(\mathcal{E}) } \ .
    \label{eq:erasure-nontrivial-Ham-Wrstar---candidate-value}
  \end{align}
  The second equality follows because $\mathcal{E}_r\cdots\mathcal{E}_2\mathcal{E}_1(\Gamma) = \Gamma$, since each $\mathcal{E}_i$ satisfies $\mathcal{E}_i(\Gamma) = \Gamma$.
  Thus, $\beta W_r^* = \beta W(\mathcal{E}) = \log \bm{(} \tr`(Q\Gamma) \bm{)} \geq -\Dhypr[r][\eta][\mathcal{T}]{\rho}{\Gamma}$.

  We now prove that $\beta W_r^* \leq -\Dhypr[r][\eta][\mathcal{T}]{\rho}{\Gamma}$.
  Consider any $\zeta>0$.
  There exists a $Q \in \Mr[r]$ such that $\tr`(Q\rho)\geq \eta$ and $\log \bm{(} \tr`(Q\Gamma) \bm{)} \leq -\Dhypr[r][\eta][\mathcal{T}]{\rho}{\Gamma} + \zeta$.
  Since $\Mr[r] = \Mr[r]`*(\Psimple,C_{\mathcal{T}})$, $Q = \mathcal{E}_1^\dagger \mathcal{E}_2^\dagger \cdots \mathcal{E}_r^\dagger ( P )$ for some operations $\mathcal{E}_i \in \mathcal{T}$ and for some effect $P \in \Psimple$, with $P = \Ident_{\mathcal{W}} \otimes \proj{ 0^{ n-\abs{\mathcal{W}} } }_{\mathcal{W}^{\rm c} }$ for some subset $\mathcal{W} \subset `{ 1, 2, \ldots, n }$.
  Let $\mathcal{E} \coloneqq `*(\textstyle\prod_{i\in\mathcal{W}}\mathcal{E}_{\textsc{RESET},i}) \mathcal{E}_r\cdots\mathcal{E}_2\mathcal{E}_1$.
  As per~\eqref{eq:erasure-nontrivial-Ham-Wrstar---candidate-constraint}, $F^2 \bm{(} \mathcal{E}(\rho),\proj{0^n} \bm{)} = \tr`(Q\rho) \geq \eta$, so $\mathcal{E}$ is a candidate for the $W_r^*$ optimization in~\eqref{eq:erasure-nontrivial-Ham-Wrstar---Wrstar}.
  As per~\eqref{eq:erasure-nontrivial-Ham-Wrstar---candidate-value}, $\tr`(Q\Gamma) = \ee^{\beta W(\mathcal{E})}$, so $\beta W_r^* \leq \beta W(\mathcal{E}) = \log \bm{(} \tr`(Q\Gamma) \bm{)} \leq -\Dhypr[r][\eta][\mathcal{T}]{\rho}{\Gamma} + \zeta$.
  Thus, $\beta W_r^* \leq -\Dhypr[r][\eta][\mathcal{T}]{\rho}{\Gamma}$, since $\zeta$ is arbitrary.
  Finally,~\eqref{eq:erasure-nontrivial-Ham-Wrstar---DHypr} follows from substituting~\eqref{eq:erasure-nontrivial-Ham-Wrstar---Dhypr} into~\eqref{eq:relation-Dhypr-DHypr}.
\end{proof}

\section{Evolution of the complexity entropy under random circuits}
\label{appx-topic:complexity-entropy-random-circuits}

Here, we prove the bound~\eqref{eq:mainthm-complexity-entropy-random-circuits-lower-bound} to complete the proof of \cref{mainthm:ComplexityEntropyInRandomCircuits} (\cref{sec-topic:main-random-circuits}), which describes the complexity entropy's evolution under random circuits.
Random circuits are often used as proxies for Hamiltonian quantum chaotic dynamics.
We adapt the proof of Theorem 8 in Ref.~\cite{Brandao2021PRXQ_models}; the theorem describes the strong complexity's growth under random circuits.

We consider a system of $n \geq 2$ qubits.
Let $G$ denote any set of two-qubit unitary gates: $G \subset \SU(4)$.
Let $`*{ \Mr[r][G] }$ denote the family of POVM-effect-complexity sets (\cref{defn:Mr-sets-general}) defined by~\eqref{eq:setting-defn-Mr}, wherein the computational gate set $\mathcal{G}$ equals the set of gates on $n$ qubits, in $G$, with arbitrary (fixed) connectivity.

Let $\mathtt{E} = `{ p_j, U_j }$ denote a random ensemble of unitary operators.
Each unitary $U_j \in \mathrm{U}(2^n)$ is chosen with the probability $p_j$.
For each integer $k>0$, we define the \emph{$k$-twirling superoperator} associated with $\mathtt{E}$ as
\begin{align}
  \mathcal{M}_{\mathtt{E}}^{(k)}(\cdot) \coloneqq \sum_j p_j \, U_j^{\otimes k} \, (\cdot) \, U_j^{\dagger\,\otimes k} \ .
\end{align}
The $k$-twirling superoperator
associated with the Haar measure is
\begin{align}
  \mathcal{M}_{\mathrm{Haar}}^{(k)}(\cdot) \coloneqq
  \int dU\, U^{\otimes k} \, (\cdot) \, U^{\dagger\,\otimes k} \ .
\end{align}
$dU$ denotes the Haar measure on the unitary group $\mathrm{U}(2^n)$.
If $\mathcal{M}_{\mathtt{E}}^{(k)}$ and $\mathcal{M}_{\mathrm{Haar}}^{(k)}$ are sufficiently close in diamond distance, $\mathtt{E}$ is an approximate $k$-design.

\begin{definition}[Approximate unitary $k$-design]
  \label{defn:approx-k-design}
  Let $\epsilon\geq 0$, and let $k>0$ denote an integer.
  An ensemble $\mathtt{E} = `{ p_j, U_j }$ is an \emph{$\epsilon$-approximate unitary $k$-design} if
  \begin{align}
    \dianorm[\big]{ \mathcal{M}_{\mathtt{E}}^{(k)} - \mathcal{M}_{\mathrm{Haar}}^{(k)} }
    \leq \frac{\epsilon}{2^{nk}} \ .
    \label{eq:defn-approx-k-design}
  \end{align}
\end{definition}

We employ the definition of an approximate $k$-design used in Ref.~\cite{Haferkamp2022Q_random}; this definition differs from that used in Ref.~\cite{Brandao2021PRXQ_models}.
In Ref.~\cite{Brandao2021PRXQ_models}, $`*( k!/2^{2nk} ) \, \epsilon$ replaces the upper bound in~\eqref{eq:defn-approx-k-design}.
Therefore, an $\epsilon$-approximate $k$-design, according to \cref{defn:approx-k-design}, is a $`*(2^{nk}\epsilon /k!)$-approximate $k$-design, according to Ref.~\cite{Brandao2021PRXQ_models}.

Following Ref.~\cite{Brandao2021PRXQ_models}, we rely on the fact that random circuits generate approximate $k$-designs~\cite{Brandao2016CMP_local,Haferkamp2022Q_random,%
  Haferkamp2021PRA_improved,chen2024incompressibility}.
In particular, we utilize the following result, which shows that certain random circuits of a depth $t>0$ are approximate $k$-designs, with $k \sim t$.

\begin{theorem}[\protect\cite{chen2024incompressibility}]
  \label{thm:Haferkamp2022Q_random}
  \noproofref
  Consider an $n$-qubit system.
  Let $\tilde\epsilon>0$ and $c\geq 0$ be independent of $n$.
  Let $V$ denote a circuit of a depth $t > 0$, with staggered layers of nearest-neighbor two-qubit gates (the ``brickwork'' layout).
  Suppose each of the circuit's gates is chosen at random from the Haar measure on $\SU(4)$.
  Then $V$ is a $`\big(2^{-cnk}\tilde\epsilon)$-approximate $k$-design, with
  \begin{align}
    k = \min`*{ 2^{n/2 - O`*(\sqrt{n})} \ ,\ \frac{ t }{\poly(n)}  } \ .
    \label{eq:thm-Haferkamp2022Q_random-k-expression}
  \end{align}
\end{theorem}
\noindent
The dependence of $k$ on $\tilde\epsilon$ is hidden in the coefficient implied by the $\poly(n)$ notation.

We briefly show how condition~\eqref{eq:thm-Haferkamp2022Q_random-k-expression} arises from the results in Ref.~\cite{chen2024incompressibility}.
Theorem~1.5 in Ref.~\cite{Haferkamp2022Q_random} implies the following statement.
Let $0 < k \leq 2^{n/2-\sqrt{n}-2}$ and let $\epsilon>0$.
Then there exists $C>0$ such that depth-$t$ brickwork random circuits form an $\epsilon$-approximate $k$-design whenever
\begin{align}
    t \geq C n^3
  [ 2nk + \log_2(1/\epsilon) ]  \eqqcolon t_0(k) \ .
\end{align}
Setting $\epsilon=2^{-cnk}\tilde\epsilon$, with $\tilde\epsilon\in(0,1)$ a constant of $n$, and using $k\leq 2^{n/2}$, we find that
\begin{align}
    t_0(k) \leq \poly(n)\, k \ .
    \label{eq:oipwfuhjsajknlsfamdnlfosaji}
\end{align}
Consider now the setting of \cref{thm:Haferkamp2022Q_random}.
The $k$ specified in~\eqref{eq:thm-Haferkamp2022Q_random-k-expression} implies $t \geq t_0(k)$, ensuring that Corollary~1 in Ref.~\cite{Haferkamp2022Q_random} applies.

The bound~\eqref{eq:mainthm-complexity-entropy-random-circuits-lower-bound}---and hence \cref{mainthm:ComplexityEntropyInRandomCircuits}---follows from \cref{thm:Haferkamp2022Q_random} and the following proposition.

\begin{proposition}[Complexity-entropy lower bound with an approximate $k$-design (simple)]
  \label{thm:ComplexityEntropyLowerBoundApproximateKDesignSimple}
  Let $\ket{\psi_0}$ denote any pure $n$-qubit state.
  Let $\tilde\epsilon\in(0,1)$.
  Assume that $k$ is even.
  Let $V$ denote a circuit sampled from a $`\big(2^{-nk}\tilde\epsilon)$-approximate $k$-design.
  Let $r \geq 0$ and $\eta \in (0,1]$.
  Suppose that
  \begin{align}
    r \leq \frac{ k `*[ \frac{n}{2} - \log_2(k) - \log_2(4/\eta) ]
    - `*{ 3n + \log_2(n) `*[ 1 + \log_2(4/\eta) ] } }
    { c_1 `*[ 2n + \log_2`(2/\eta) ]^4 + \log_2`(n+1) + 1 } \ ,
    \label{eq:thm-ComplexityEntropyLowerBoundApproximateKDesignSimple-condition-r}
  \end{align}
  wherein $c_1>0$ denotes a constant independent of $n$.
  Then
  \begin{align}
    \HHypr[r][\eta] `*{ V \psi_0 V^\dagger }
    \geq  \Hhypr[r][\eta] `*{ V \psi_0 V^\dagger }
    \geq n\log(2) - \log`*( \frac{4}{\eta} ) \ ,
    \label{eq:thm-ComplexityEntropyLowerBoundApproximateKDesignSimple-condition-main-bound}
  \end{align}
  except with a probability $\leq \ee^{-\Omega(n)}$ over the sampling of $V$.
  Here, $\HHypr[r][\eta]{}$ is defined with respect to $\Mr[r][\SU(4)]$.
\end{proposition}

We specify the constant $c_1$ in the proof of \cref{thm:ComplexityEntropyLowerBoundApproximateKDesignSimple}.
To prove the proposition, we first use the Solovay-Kitaev theorem~\cite{%
    Dawson2005arXiv_SolovayKitaev,%
    BookNielsenChuang2000,%
    BookKitaevShenVyalyi2002,%
    Oszmaniec2021arXiv_epsilonnets}
to approximate $\SU(4)$ with a finite, universal gate set.
We then apply the following lemma, whose proof closely follows that of Theorem 8 of Ref.~\cite{Brandao2021PRXQ_models}.

\begin{lemma}[Gate-set complexity-entropy lower bound for approximate $k$-designs]
  \label{thm:ComplexityEntropyLowerBoundApproximateKDesignPrecise}
  Let $\ket{\psi_0}$ denote any pure $n$-qubit state.
  Let $\tilde\epsilon\in(0,1)$.
  Assume that $k$ is even.
  Let $V$ denote a circuit sampled from a $`\big(2^{-nk}\tilde\epsilon)$-approximate $k$-design.
  Let $r \geq 0$, $\bar\eta \in (0,1]$, $c \in (0,1)$, and $g > 0$.
  Suppose that
  \begin{align}
    r   &\leq \frac{1}{\log_2`*( [n+1] \abs{G} ) }
        `*(
            k `\bigg[ \frac{n}{2} - \log_2(k) - \log_2`*( \frac1{\bar\eta\,(1-c)} ) ]
            - `*{ (g+1)n + 2 + \log_2(n) \left[ 1 + \log_2`*(\frac1{c\bar\eta}) \right] }
        ) \ .
      \label{eq:thm-ComplexityEntropyLowerBoundApproximateKDesignPrecise-condition-r}
  \end{align}
  Then
  \begin{align}
    \Hhypr[r][\bar\eta][G] `*{ V\psi_0 V^\dagger }
    &\geq n\log(2) - \log`*(\frac1{c \bar\eta}) \ ,
    \label{eq:lower-bound--approx-k-design---precise}
  \end{align}
  except with a probability $\leq 2^{-gn}$ over the sampling of $V$.
  Here, the reduced complexity entropy $\Hhypr[r][\bar\eta][G]{}$ (\cref{defn:Hhypr}) is defined with respect to $\Mr[r][G]$, wherein $G$ is finite.
\end{lemma}

\begin{proof}[**thm:ComplexityEntropyLowerBoundApproximateKDesignPrecise]
  Let
  \begin{align}
    m \coloneqq \biggl\lfloor n - \log_2`*(\frac1{c\bar\eta}) \biggr\rfloor
    = n - \biggl\lceil \log_2`*(\frac1{c\bar\eta}) \biggr\rceil \ .
    \label{eq:m-floor-of-lower-bound}
  \end{align}
  \eqref{eq:lower-bound--approx-k-design---precise} holds, except with the probability
  \begin{align}
      p_<
      &\coloneqq
        \Pr_V `*[ \Hhypr[r][\bar\eta][G] `*{ V\psi_0 V^\dagger } < n\log(2) - \log`*(\frac1{c \bar\eta}) ] 
        \nonumber \\
      &\leq
        \Pr_V `*[ \Hhypr[r][\bar\eta][G] `*{ V\psi_0 V^\dagger } \leq n\log(2) - \log`*(\frac1{c \bar\eta}) ]
        \nonumber \\
      &=
        \Pr_V `*[ \Hhypr[r][\bar\eta][G] `*{ V\psi_0 V^\dagger } \leq m \log (2) ]
        \nonumber \\
      &\leq
        \Pr_V `*[ \etahypr[r][m \log(2)][G] `*{ V\psi_0 V^\dagger } \geq \bar\eta ]
        \nonumber \\
      &=
        \Pr_V `*[ \max_{ \substack{ Q \in \Mr[r] \\ \log_2\bm{(} \tr`(Q) \bm{)} = m } }
        `*{ \tr`*(Q V\psi_0 V^\dagger) } \geq \bar\eta ]
        \nonumber \\
      &\leq
        \sum_{ \substack{ Q \in \Mr[r] \\ \log_2\bm{(} \tr`(Q) \bm{)} = m } }
        \Pr_V `*[ \tr`*(Q V\psi_0 V^\dagger) \geq \bar\eta ]
        \nonumber \\
      &=
        \sum_{ \substack{ Q \in \Mr[r] \\ \log_2\bm{(} \tr`(Q) \bm{)} = m } }
        \Pr_V `*[ \tr`*( `*[Q - 2^{m-n} \Ident] V\psi_0 V^\dagger) \geq \bar\eta - 2^{m-n} ]
        \nonumber \\
      &\leq
        \sum_{ \substack{ Q \in \Mr[r] \\ \log_2\bm{(} \tr`(Q) \bm{)} = m } }
        \Pr_V `*[ \abs`*{ \tr`*( `*[Q - 2^{m-n} \Ident] V\psi_0 V^\dagger) } \geq \abs`*{ (1-c)\bar\eta } ]
        \nonumber \\
      &=
        \sum_{ \substack{ Q \in \Mr[r] \\ \log_2\bm{(} \tr`(Q) \bm{)} = m } }
        \Pr_V `*[ `*{ \tr`*( `*[Q - 2^{m-n} \Ident] V\psi_0 V^\dagger) }^k \geq `*{ (1-c)\bar\eta }^k ]
        \nonumber \\
      &\leq
        \sum_{ \substack{ Q \in \Mr[r] \\ \log_2\bm{(} \tr`(Q) \bm{)} = m } }
        \frac{ \mathbb{E}_V `*[ `*{ \tr`*( [Q - 2^{m-n} \Ident] V \psi_0 V^\dagger ) }^k ] } { `*[ (1-c)\bar\eta ]^k }
        \nonumber \\
      &\leq
        \sum_{ \substack{ Q \in \Mr[r] \\ \log_2\bm{(} \tr`(Q) \bm{)} = m } } 2 `*( \frac{k}{ 2^{n/2} (1-c) \bar\eta } )^k
        \nonumber \\
      &=
        2 `*( \frac{k}{ 2^{n/2} (1-c) \bar\eta } )^k \abs`*{ `*{ Q \in \Mr[r] \,:\; \log_2\bm{(} \tr`(Q) \bm{)} = m } } \ .
  \end{align}
  The second equality follows because, by the form of $\Mr[r]$, $\Hhypr[r][\bar\eta] `*{ V\,\psi_0 V^\dagger }$ is an integer multiple of $\log(2)$.
  The second inequality follows from \cref{defn:cplx-entropy-success-probability}: $\Hhypr[r][\bar\eta] `*{ V\,\psi_0 V^\dagger } \leq m \log (2)$ implies that $\etahypr[r][m \log(2)]`*{ V\psi_0 V^\dagger } \geq \bar\eta$.
  The third equality follows from~\eqref{eq:cplx-success-probability--special-case}; here, the supremum is a maximum because $G$---and hence $\Mr[r][G]$---is finite.
  The third inequality is
  a crude union bound.
  The fourth equality introduces the traceless part of each $Q$: $Q - 2^{m-n} \Ident = Q - \tr`(Q) \cdot 2^{-n} \Ident$. 
  The fourth inequality follows because 
  \begin{equation}
  \abs`*{ \tr`*( `*[Q - 2^{m-n} \Ident] V\psi_0 V^\dagger) } \geq \tr`*( `*[Q - 2^{m-n} \Ident] V\psi_0 V^\dagger) 
  \end{equation}
  and because $\bar\eta - 2^{m-n} \geq (1-c)\bar\eta > 0$, since $2^{m-n} = 2^{-\lceil \log_2`*(1/c\bar\eta) \rceil} \leq 2^{-\log_2`*(1/c\bar\eta)} = c\bar\eta$.
  The fifth equality follows because, by assumption, $k$ is even.
  The fifth inequality follows from Markov's inequality: for any non-negative random variable $X$ and any $\tau > 0$, $\Pr[X \geq \tau] = \mathbb{E}[X]/\tau$.
  The last inequality follows from Corollary~24 of Ref.~\cite{Brandao2021PRXQ_models}: $\mathbb{E}_V `*[ `*{ \tr( [Q - 2^{m-n} \Ident] V \psi_0 V^\dagger) }^k ] \leq `*(1+\epsilon) `*( k / 2^{n/2})^k \leq 2 `*( k / 2^{n/2})^k$, wherein $\epsilon=\frac{2^{nk}}{k!}2^{-nk}\tilde\epsilon \leq 1$ (recall the alternative convention used for $\epsilon$-approximate designs in Ref.~\cite{Brandao2021PRXQ_models}).

  By the form of $\Mr[r]$, $\abs`*{ `*{ Q \in \Mr[r] \,:\; \log_2\bm{(} \tr`(Q) \bm{)} = m } } \leq N_r \binom{n}{m}$.
  The $N_r$ denotes the number of circuits that one can compose from $\leq r$ gates in $\mathcal{G}$.
  $\binom{n}{m}$ is the number of projectors that, up to a permutation of qubits, equal $\Ident_2^{\otimes m} \otimes \proj{0^{n-m}}$.
  By Eq.~(B21) of Ref.~\cite{Brandao2021PRXQ_models}, $N_r \leq 2^{n+1} (n+1)^r \abs{ G }^r$.
  Moreover,
  $\binom{n}{m}
    = \binom{n}{n-m} \leq n^{n-m}
    = n^{\lceil \log_2`*( 1/c\bar\eta ) \rceil}
    \leq n^{1 + \log_2`*( 1/c\bar\eta ) }$.
  Hence,
  \begin{align}
      \abs`*{ `*{ Q \in \Mr[r] \,:\; \log_2\bm{(} \tr`(Q) \bm{)} = m } } \leq 2^{n+1} (n+1)^r \abs{ G }^r \cdot n^{1 + \log_2`*( 1/c\bar\eta ) } \ ,
  \end{align}
  so $p_< \leq 2^{a + rb - ky}$,
  where
  \begin{align}
    a \coloneqq n + 2 + \log_2(n) `*[ 1 + \log_2 `*( \frac1{c\bar\eta} ) ] \ , \; \; \;
      b \coloneqq \log_2 `*( [n+1] \abs{G} ) \ , \; \; \;
    \text{and} \; \; \;
      y \coloneqq \frac{n}{2} - \log_2(k) - \log_2 `*( \frac1{(1-c)\bar\eta} ) \ .
  \end{align}
  By~\eqref{eq:thm-ComplexityEntropyLowerBoundApproximateKDesignPrecise-condition-r}, $r \leq \frac1{b} `*[ ky - ( gn + a ) ]$, so $a + rb - ky \leq -gn$.
  Thus, $p_< \leq 2^{-gn}$.
\end{proof}

\begin{proof}[*thm:ComplexityEntropyLowerBoundApproximateKDesignSimple]
  Let $\epsilon' \coloneqq \eta/2$.
  Let $G \subset \SU(4)$ denote any finite gate set that forms an $(\epsilon'/4^n)$-net on $\SU(4)$.
  One can build such a gate set using the Solovay-Kitaev theorem~\cite{%
    Dawson2005arXiv_SolovayKitaev,%
    BookNielsenChuang2000,%
    BookKitaevShenVyalyi2002,%
    Oszmaniec2021arXiv_epsilonnets}
  as follows.
  Consider any finite, universal gate set $G_0 \subset \SU(4)$, such as the Clifford + $T$ gate set (see \cref{sec:examples-unitary-complexity measures}).
  By the Solovay-Kitaev theorem, one can approximate any unitary $U \in \SU(4)$ to the error $\epsilon'$ with at most $N \coloneqq c_{G_0} `*[ \log_2(4^n/\epsilon') ]^4 = c_{G_0} `*[ 2n + \log_2`(2/\eta) ]^4$ gates in $G_0$.
  The $c_{G_0} > 0$ is a universal constant depending only on $G_0$.
  Now, let $G$ consist of all the circuits that one can compose from $\leq N$ gates in $G_0$.
  One can approximate any unitary $U \in \SU(4)$ to the error $\epsilon'$ by some circuit in $G$.
  The number of circuits that one can compose from exactly $t > 0$ gates in $G_0$ is $\leq \abs{G_0}^t$, so
  \begin{align}
      \abs{G}
      \leq \sum_{t=1}^N \abs{G_0}^t
      = \frac{ \abs{G_0} } { \abs{G_0} - 1 } `*( \abs{G_0}^N - 1  )
      \leq 2 \abs{G_0}^N \ .
  \end{align}
  Hence,
  \begin{align}
    \log_2`*( \abs{G} )
    \leq N \log_2`*( \abs`*{G_0} ) + 1 
    = c_1 `*[ 2n + \log_2`*(\frac{2}{\eta}) ]^4 + 1 \ ,
  \end{align}
  wherein $c_1 \coloneqq c_{G_0}\,\log_2`*( \abs`*{ G_0 } )$.
  
  We now apply \cref{thm:ComplexityEntropyLowerBoundApproximateKDesignPrecise}.
  Set
  \begin{align}
      g = 1
      \ , \; \; \; \; \; \; \; \;
      c = \frac12 \ , \; \; \; \; \; \; \; \;
    &
    \text{and} \; \; \; \; \; \; \; \;
      \bar\eta = \eta - \epsilon' = \frac\eta2 \ .
  \end{align}
  $r$ satisfies condition~\eqref{eq:thm-ComplexityEntropyLowerBoundApproximateKDesignPrecise-condition-r} in the lemma, since the right-hand side of~\eqref{eq:thm-ComplexityEntropyLowerBoundApproximateKDesignSimple-condition-r} lower-bounds the right-hand side of~\eqref{eq:thm-ComplexityEntropyLowerBoundApproximateKDesignPrecise-condition-r}, as one can verify using the following inequalities: 
  \begin{gather}
    \log_2`*( [n+1] \abs`*{ G } )
    \leq c_1 `*[ 2n + \log_2`*( \frac{2}{\eta} ) ]^4 + \log_2`(n+1) + 1 \ ,
    \\
    \log_2`*( \frac1{(1-c)\bar\eta} ) = \log_2`*( \frac4{\eta} ) \ ,
    \\
    \text{and} \; \; \; \; \;
    2 \leq n
    \quad\Rightarrow\quad
    (g+1)n + 2 + \log_2`*( \frac1{c\bar\eta} )
    \leq 3n + \log_2`*( \frac4\eta ) \ .
  \end{gather}
  Thus, by \cref{thm:ComplexityEntropyLowerBoundApproximateKDesignPrecise},
  \begin{align}
    \Hhypr[r][\eta-\epsilon'][G] `*{ V \psi_0 V^\dagger }
    = \Hhypr[r][\bar\eta][G] `*{ V \psi_0 V^\dagger }
    \geq n\log(2) - \log`*( \frac1{c\bar\eta} ) 
    = n\log(2) - \log`*( \frac4\eta ) \ ,
    \label{eq:cplx-entropy-lower-bound-final-ineqs-1}
  \end{align}
  except with a probability $\leq 2^{-n}$.
  Furthermore,
  \begin{align}
    \HHypr[r][\eta] `*{ V \psi_0 V^\dagger }
    \geq \Hhypr[r][\eta] `*{ V \psi_0 V^\dagger }
    \geq \Hhypr[r][\eta][G] `*{ V \psi_0 V^\dagger }
    \geq \Hhypr[r][\eta-\epsilon'][G] `*{ V \psi_0 V^\dagger } \ ,
    \label{eq:cplx-entropy-lower-bound-final-ineqs-2}
  \end{align}
  wherein $\Hhypr[r][\eta]{}$ is defined with respect to $\Mr[r][\SU(4)]$.
  The first inequality follows from \cref{thm:relation-Dhypr-DHypr}.
  The second inequality follows because $\Mr[r][G] \subset \Mr[r][\SU(4)]$: every candidate $Q \in \Mr[r][G]$ for the $\Hhypr[r][\eta][G]`*{ V \psi_0 V^\dagger }$ optimization belongs to $\Mr[r][\SU(4)]$ and is therefore a candidate for the $\Hhypr[r][\eta]`*{ V \psi_0 V^\dagger }$ optimization.
  The third inequality follows from \cref{thm:DHypr-monotonous-in-r-and-in-eta} (via \cref{thm:properties-Dhypr}).
  Therefore,
  \begin{align}
      \HHypr[r][\eta] `*{ V \psi_0 V^\dagger }
      \geq \Hhypr[r][\eta] `*{ V \psi_0 V^\dagger }
      \geq n\log(2) - \log`*( \frac4\eta ) \ ,
  \end{align}
  except with a probability at most exponentially small in $n$.
\end{proof}

\section{Entanglement bounds on the complexity entropy}
\label{appx-topic:entanglement-bounds}

Here, we use entanglement measures to bound the complexity entropy, proving the results in \cref{sec:quantum-info--entgl-bounds}.
We restrict our attention to a spatially one-dimensional chain $S$ of $n \geq 2$ qubits: $S = S_1 S_2 \ldots S_n$.
Let $`*{ \Mr[r] }$ denote the family of POVM-effect-complexity sets (\cref{defn:Mr-sets-general}) defined by~\eqref{eq:setting-defn-Mr}, wherein the computational gate set $\mathcal{G}$ contains only unitary operations that can act nontrivially only on two neighboring qubits in $S$.
Furthermore, we define the \textit{potential entangling power} of $\mathcal{G}$ (cf.\@ Ref.~\cite{Eisert2021PRL_entangling}) as
\begin{align}
  e(\mathcal{G})
  \coloneqq \sup_{U \in \mathcal{G}} `*{ e(U) } \ ,
  \quad \text{wherein} \quad
    e(U)
    \coloneqq
      \min`*{
      \opnorm{H} \,:\; H = H^\dagger , \; \exists\chi\in\mathbb{R} \text{ such that } \ee^{-iH} = \ee^{-i\chi}U } \ .
\end{align}
In general, $e(\mathcal{G})$ can be as large as $\pi$.
However, our entanglement bounds are most interesting when $\mathcal{G}$ contains only operations that act weakly and remain near the identity operator (up to an overall phase).
In this case, $e(\mathcal{G}) \approx e(\Ident) = 0$.

In the following, we use the \emph{quantum mutual information}.
Let $A$ and $B$ denote distinct quantum systems, and let $\sigma_{AB}$ denote any state of $AB$.
The quantum mutual information of $\sigma_{AB}$ is defined in terms of the von Neumann entropy as
\begin{align}
  \II[\sigma]{A}{B}
  \coloneqq \HH[\sigma]{A} + \HH[\sigma]{B} - \HH[\sigma]{AB}
  = \HH[\sigma]{A} - \HH[\sigma]{A}[B] \ .
  \label{eq:def-quantum-mutual-info}
\end{align}
For any state $\rho$ of the qubit chain $S$, we define the entanglement measure
\begin{align}
  E(\rho) &\coloneqq \frac1{n-1} \sum_{j=1}^{n-1} \II[\rho]{S_1 \ldots S_j}{S_{j+1}\ldots S_n} \ .
    \label{eq:entgl-measure-1d-chain-mutual-information}
\end{align}
$E(\rho)$ quantifies the average amount of correlation, as quantified by the mutual information, between any bipartition of $S$ into subsystems of neighboring qubits.
We now prove the bounds in \cref{mainprop:change-in-entanglement-measure} and in~\eqref{eq:maintext-entanglement-bound-on-cplx-entropy}.

\begin{theorem}[Continuity bound on $E(\rho)$ under one operation]
  \label{thm:continuity-bound-Erho-single-gate}
  Let $\rho$ denote any quantum state of $S$.
  Let $U$ denote any unitary that can act nontrivially only on two neighboring qubits in $S$.
  It holds that
  \begin{align}
    \abs`*{ E`*( U \rho U^\dagger ) - E`*(\rho) }
    \leq \frac1{n-1} \min`*{ 8 \log(2) \ , \ 8 \nu \log(2) + 3`(1+\nu) \, h`*( \frac\nu{1+\nu} ) } \ .
    \label{eq:thm-continuity-bound-Erho-single-gate--bound}
  \end{align}
  The $h(x) \coloneqq -x\log(x) - `(1-x) \log `(1-x)$ denotes the binary entropy function, and
  \begin{align}
    \nu \coloneqq \sin `*( \min`*{ e(U), \frac\pi2 } ) \ .
    \label{eq:thm-continuity-bound-Erho-single-gate--def-nu}
  \end{align}
\end{theorem}

\begin{proof}[**thm:continuity-bound-Erho-single-gate]
  Let $\rho'$ denote the state that results from evolving $\rho$ under $U$: $\rho' \coloneqq U\rho U^\dagger$.
  For all $j \in `{ 1, 2, \ldots, n-1 }$, we write, for short, $A_j \coloneqq S_1 S_2 \ldots S_j$ and $B_{j+1} \coloneqq S_{j+1} S_{j+2} \ldots S_n$.
  By~\eqref{eq:entgl-measure-1d-chain-mutual-information},
  \begin{align}
    E(\rho') - E(\rho) = \frac1{n-1} \sum_{j=1}^{n-1} 
      `*[ \II[\rho']{A_j}{B_{j+1}} - \II[\rho]{A_j}{B_{j+1}} ] \ .
  \end{align}

  Let $k, k+1 \in `{1, 2, \ldots, n}$ denote two qubits on which $U$ can act nontrivially.
  For all $j\neq k$, $\II[\rho']{A_j}{B_{j+1}} = \II[\rho]{A_j}{B_{j+1}}$: 
  $U$ can act nontrivially on qubits only in $A_j$ or only in $B_{j+1}$, so
  \begin{align}
      \II[\rho]{A_j}{B_{j+1}}
      &= \HH[\rho]{A_j} + \HH[\rho]{B_{j+1}} - \HH[\rho]{A_jB_{j+1}}
      \nonumber\\
      &= \HH[\rho']{A_j} + \HH[\rho']{B_{j+1}} - \HH[\rho']{A_jB_{j+1}}
      = \II[\rho']{A_j}{B_{j+1}} \ .
  \end{align}
  The first and last equalities follow from~\eqref{eq:def-quantum-mutual-info}.
  The middle equality follows because the von Neumann entropy is unitarily invariant.
  Thus,
  \begin{align}
    E(\rho') - E(\rho) = \frac1{n-1} `*[ \II[\rho']{A_k}{B_{k+1}} - \II[\rho]{A_k}{B_{k+1}} ] \ .
    \label{eq:diff-mutual-info}
  \end{align}
  Using the chain rule for the von Neumann entropy, we simultaneously re-express $\II[\rho]{A_k}{B_{k+1}}$ and $\II[\rho']{A_k}{B_{k+1}}$:
  \begin{align}
    &\II{A_k}{B_{k+1}}
      \nonumber\\
    &= \HH{A_k} + \HH{B_{k+1}} - \HH{A_kB_{k+1}}
      \nonumber\\
    &= \HH{A_{k-1}} + \HH{S_k}[A_{k-1}]
    + \HH{B_{k+2}} + \HH{S_{k+1}}[B_{k+2}]
     - \HH{A_{k-1}B_{k+2}} - \HH{S_kS_{k+1}}[A_{k-1}B_{k+2}]
      \nonumber\\
    &=
      \underbrace{
        \HH{S_k}[A_{k-1}] + \HH{S_{k+1}}[B_{k+2}] - \HH{S_k S_{k+1}}[A_{k-1}B_{k+2}]
      }_{\textup{(I)}}
      + \underbrace{
        \HH{A_{k-1}} + \HH{B_{k+2}} - \HH{A_{k-1}B_{k+2}}
      }_{\textup{(II)}}
      \ .
  \end{align}
  $A_{k-1} `*(B_{k+2})$ is trivial if $k=1$ ($k=n-1$).
  Since $U$ can act nontrivially only on $S_k S_{k+1}$, the reduced states of $\rho$ and $\rho'$ on $A_{k-1} B_{k+2}$ coincide.
  Hence, the terms in \textup{(II)} have the same value whether evaluated on $\rho$ or on $\rho'$, so~\eqref{eq:diff-mutual-info} equals the difference between the terms in \textup{(I)} evaluated on $\rho$ and on $\rho'$.
  We now upper-bound the absolute value of this difference.
  Let $\mathcal{U}(\cdot) \coloneqq U`(\cdot) U^\dagger$.
  By the diamond norm's definition~\eqref{eq:def-dianorm} and by \cref{thm:diamond-norm-unitary-to-ident},
  \begin{align}
    \frac12 \onenorm{\rho' - \rho}
    = \frac12 \onenorm{ `*(\mathcal{U}-\IdentProc{}) \rho}
    \leq \frac12 \dianorm{ \mathcal{U} - \IdentProc{} }
    = \sin `*( \min`*{ e(U), \frac\pi2} )
    = \nu \ .
    \label{eq:one-norm-dist-for-entgl-bound}
  \end{align}
  By the Alicki-Fannes-Winter continuity bound on the conditional entropy~\cite{Alicki2004JPA_continuity,Winter2016CMP_tight} (see Lemma~2 of Ref.~\cite{Winter2016CMP_tight}),~\eqref{eq:one-norm-dist-for-entgl-bound} implies that the terms in \textup{(I)} obey the following bounds:
  \begin{subequations}
    \begin{gather}
      \abs`*{ \HH[\rho']{S_k}[A_{k-1}] - \HH[\rho]{S_k}[A_{k-1}] }
      \leq 2\nu\log`\big(d_{S_k}) + `(1+\nu) \, h`*(\frac\nu{1+\nu}) \ ,
      \\
      \abs`*{ \HH[\rho']{S_{k+1}}[B_{k+2}] - \HH[\rho]{S_{k+1}}[B_{k+2}] }
      \leq 2\nu\log`\big(d_{S_{k+1}}) + `(1+\nu) \, h`*(\frac\nu{1+\nu}) \ ,
      \\
      \text{and} \; \; \;
      \abs`*{ \HH[\rho']{S_k S_{k+1}}[A_{k-1}B_{k+2}]
      - \HH[\rho]{S_k S_{k+1}}[A_{k-1}B_{k+2}] }
    \leq 2\nu\log`\big( d_{S_k} d_{S_{k+1}} ) + `(1+\nu) \, h`*(\frac\nu{1+\nu}) \ .
    \end{gather}
  \end{subequations}
  Therefore,
  \begin{align}
    \abs`*{ \II[\rho']{A_k}{B_{k+1}} - \II[\rho]{A_k}{B_{k+1}} }
    &\leq 2\nu\log`\big(d_{S_k}^2 d_{S_{k+1}}^2) + 3`(1+\nu)\, h`*( \frac\nu{1+\nu} )
    \nonumber\\
    &= 8\nu\log(2) + 3`(1+\nu)\, h`*( \frac\nu{1+\nu} ) \ .
    \label{eq:abs-diff-mutual-info---AFW-bound}
  \end{align}
  The equality follows because $S_k$ and $S_{k+1}$ are qubits: $d_{S_k} = d_{S_{k+1}} = 2$.
  Moreover, each term in $(I)$ can vary by at most the range of the term's corresponding conditional entropy, so
  \begin{align}
    \abs`*{ \II[\rho']{A_k}{B_{k+1}} - \II[\rho]{A_k}{B_{k+1}} }
    \leq 2\log`\big(d_{S_k}) + 2\log`\big(d_{S_{k+1}}) + 2\log`\big(d_{S_k} d_{S_{k+1}})
    = 2\log`\big(d_{S_k}^2 d_{S_{k+1}}^2) 
    = 8\log(2) \ .
    \label{eq:abs-diff-mutual-info---range-bound}
  \end{align}
  Combining~\eqref{eq:diff-mutual-info},~\eqref{eq:abs-diff-mutual-info---AFW-bound}, and~\eqref{eq:abs-diff-mutual-info---range-bound} yields
  \begin{align}
    \abs`*{ E(\rho') - E(\rho) } \leq \frac1{n-1} \min`*{ 8\log(2) \ ,\ 8\nu\log(2) + 3`(1+\nu)\, h`*( \frac\nu{1+\nu} ) } \ .
  \end{align}
\end{proof}

It is convenient to define the worst-case upper bound in~\eqref{eq:thm-continuity-bound-Erho-single-gate--bound} among all the gates in $\mathcal{G}$.
Namely, let
\begin{align}
  \mu(\mathcal{G})
  \coloneqq \frac1{n-1}  \min`*{
    8\log(2) \ ,\ 
    8 \nu_\mathcal{G} \log(2) + 3`(1+\nu_\mathcal{G})\, h`*( \frac{ \nu_\mathcal{G} }{ 1 + \nu_\mathcal{G} } ) } \ ,
    \; \; \text{where} \; \;
    \nu_\mathcal{G} \coloneqq \sin `*( \min`*{ e(\mathcal{G}), \frac\pi2 } ) \ .
\end{align}
By definition, $\mu(\mathcal{G})$ always satisfies $\mu(\mathcal{G}) \leq 8 \log(2) / (n-1)$.

\begin{theorem}[Entanglement bound on the complexity entropy]
  \label{thm:Hhypr-bound-entgl-Erho}
  Let $\rho$ denote any quantum state of $S$.
  Let $r\geq 0$ and $\eta\in(0,1]$.
  It holds that
  \begin{align}
    \HHypr[r][\eta]{\rho} \geq \frac1{\eta} `*[ E(\rho) - r \, \mu(\mathcal{G}) + H(\rho) - 2h(\eta) - (1-\eta) \, n\log(2) ] -2\log(\eta) \ .
    \label{eq:thm-lower-bound-on-cplx-entropy}
  \end{align}
  The $h(x) \coloneqq -x\log(x) - `(1-x) \log `(1-x)$ denotes the binary entropy function, and $\HHypr[r][\eta]{}$ is defined with respect to $\Mr[r]$.
\end{theorem}

\begin{proof}[**thm:Hhypr-bound-entgl-Erho]
  Consider any $\zeta>0$.
  There exists a $Q\in \Mr[r]$ such that $\tr`(Q\rho) \geq \eta$ and $\log \bm{(} \tr`(Q) / \tr`(Q\rho) \bm{)} \leq \HHypr[r][\eta]{\rho} + \zeta$.
  $Q = U^\dagger P U$ for some effect $P \in \Mr[r=0]$ and for some unitary $U$ satisfying $C_{\mathcal{G}}(U_0) \leq r$.
  Let $\rho' \coloneqq U \rho U^\dagger$.
  $\tr`(P\rho') = \tr`(Q\rho) \geq \eta$, so $P$ is a candidate for the $\HHypr[r=0][\eta]{\rho'}$ optimization:
  \begin{align}
    \HHypr[r=0][\eta]{\rho'}
    \leq \log `*( \frac{\tr`*(P)}{\tr`*(P\rho')} )
    = \log `*( \frac{\tr`*(Q)}{\tr`*(Q\rho)} )
    \leq \HHypr[r][\eta]{\rho} + \zeta \ .
    \label{eq:jdsa90hieubhfjsdviuakls}
  \end{align}
  For all $j \in `{ 1, 2, \ldots, n-1 }$, we
  write, for short, $A_j \coloneqq S_1 S_2 \ldots S_j$ and $B_{j+1} \coloneqq S_{j+1} S_{j+2} \ldots S_n$.
  Since every POVM effect in $\Mr[r=0]$ is tensor-product, \cref{thm:bound-DHypr-DHyp-tensor-products} implies that, for all $j$,
  \begin{align}
    \HHypr[r=0][\eta]{\rho'} 
    = -\DHypr[r=0][\eta]{\rho'}{\Ident_S}
    \geq - \DHyp[\eta]{ \rho'_{A_j} }{ \Ident_{A_j} } - \DHyp[\eta]{ \rho'_{B_{j+1}} }{ \Ident_{B_{j+1}} } \ .
  \end{align}
  Consequently,
  \begin{align}
    \HHypr[r=0][\eta]{\rho'} 
    \geq \frac1{n-1} \sum_{j=1}^{n-1} `*[ -\DHyp[\eta]{ \rho'_{A_j} }{ \Ident_{A_j} } - \DHyp[\eta]{ \rho'_{B_{j+1}} }{ \Ident_{B_{j+1}} } ] \ .
    \label{eq:averaged-hyp-entropy-lower-bound}
  \end{align}
  By Proposition~4.67 of Ref.~\cite{BookKhatri_communication}, one can bound each hypothesis-testing relative entropy in~\eqref{eq:averaged-hyp-entropy-lower-bound} in terms of a standard (Umegaki) quantum relative entropy: for each $A_j$,
  \begin{align}
    -\DHyp[\eta]{ \rho'_{A_j} }{ \Ident_{A_j} }
    &\geq - \frac1\eta `*[ \DD{ \rho'_{A_j} }{ \Ident_{A_j} } + h(\eta) + `(1-\eta) \log \bm{(} \tr`( \Ident_{A_j} ) \bm{)} ] - \log(\eta)
    \nonumber\\
    &=  \frac1\eta `*[ H`\big( \rho'_{A_j} ) - h(\eta) - j (1-\eta) \log(2) ] -\log(\eta) \ .
  \end{align}
  Similarly, for each $B_{j+1}$,
  \begin{align}
    -\DHyp[\eta]{ \rho'_{B_{j+1}} }{ \Ident_{B_{j+1}} }
    &\geq \frac1\eta `*[ H`\big( \rho'_{B_{j+1}} ) - h(\eta) - (n-j) `(1-\eta) \log(2) ] -\log(\eta)  \ .
  \end{align}
  Thus,
  \begin{align}
    \HHypr[r][\eta]{\rho} + \zeta
    &\geq \HHypr[r=0][\eta]{\rho'}
    \nonumber\\
    &\geq \frac1{n-1} \sum_{j=1}^{n-1} `*{ \frac1\eta `*[ H`\big(\rho'_{A_j}) + H`\big(\rho'_{B_{j+1}})  - 2h(\eta) - `(1-\eta)\,n\log(2) ] -2\log(\eta) }
      \nonumber\\
    &= \frac1\eta `*{ \frac1{n-1} \sum_{j=1}^{n-1}  `*[ \II[\rho']{A_j}{B_{j+1}} + H`*(\rho') ] -  2h(\eta) - `(1-\eta)\,n\log(2) } -2\log(\eta)
      \nonumber\\
    &= \frac1\eta `*[ E(\rho') + H(\rho') - 2h(\eta) - `(1-\eta)\,n\log(2) ] -2\log(\eta)
      \nonumber\\
    &\geq \frac1\eta `*[ E(\rho) - r \mu(\mathcal{G}) + H(\rho) - 2h(\eta) - `(1-\eta)\,n\log(2) ] -2\log(\eta) \ .
    \label{eq:thm-lower-bound-on-cplx-entropy--with-zeta}
  \end{align}
  The first and second equalities follow from the definitions~\eqref{eq:def-quantum-mutual-info} and~\eqref{eq:entgl-measure-1d-chain-mutual-information}, respectively.
  The last inequality follows because $E(\rho') - E(\rho) \geq - r \mu(\mathcal{G})$, by \cref{thm:continuity-bound-Erho-single-gate}, and because $H(\rho') = H(\rho)$, since the von Neumann entropy is unitarily invariant.
  Finally,~\eqref{eq:thm-lower-bound-on-cplx-entropy--with-zeta} implies~\eqref{eq:thm-lower-bound-on-cplx-entropy}, since $\zeta$ is arbitrary.
\end{proof}

\section{Bounds from short-time Hamiltonian dynamics}
\label{appx-sec:Hamiltonian-dynamics}

Here, we apply the bounds in \cref{appx-topic:entanglement-bounds} to short-time Hamiltonian dynamics.
We consider a nearest-neighbor Hamiltonian $H$ acting on a 1D chain $S$ of $n \geq 2$ qubits.
The chain is initially in a state $\rho$.
We show that the entanglement measure $E(\rho)$ can grow only slowly in time.

\begin{corollary}[Continuity bound on $E(\rho)$ under local Hamiltonian evolution]
\noproofref
Let $\rho$ denote any quantum state of $S$.
Let $H$ denote a nearest-neighbor, translationally invariant Hamiltonian governed by local-interaction terms $h$ that can act on only two qubits at a time. 
It holds that
\begin{equation}
  \frac{d}{dt} \abs`*{ E`*( \ee^{-iHt} \rho \ee^{iHt} ) - \rho } \leq C (n-1) \opnorm{h} \ ,
  \label{eq:Hamiltonianbound}
\end{equation}
for some universal constant $C>0$.
\end{corollary}

\begin{proof}
  The Hamiltonian has the form
  \begin{equation}
    H = \sum_{j=1}^{n-1} \tau_j(h) \ ,
  \end{equation}
  wherein $h$ denotes a fixed interaction term that acts nontrivially only on two neighboring qubits in $S$.
  For all $j \in `*{ 1, 2, \ldots, n-1 }$, $\tau_j$ is the shift operator that places $\tau_j(h)$ onto qubits $j$ and $j+1$. 
  By the Trotter formula, it is manifest that for suitably small times $t>0$, $\ee^{-iHt} \rho \ee^{iHt}$ can be arbitrarily well approximated by a brickwork circuit involving nearest-neighbor quantum gates, even for noncommuting $\tau_j(h)$ for different $j$. 
  For each cut specified by $j$, the small incremental entangling bound for the von-Neumann entropy \cite{Marien2016, Eisert2021PRL_entangling} gives
  \begin{equation} 
    \frac{d}{dt} I(S_1\ldots S_j : S_{j+1}\ldots S_n)_\rho \leq C \opnorm{h} \ .
  \end{equation}
  Summing this inequality over all cuts yields~\eqref{eq:Hamiltonianbound}.
  The constant $C$ is proven in Ref.~\cite{Marien2016} to be $C = 22 \log(2)$.
  The $2$ originates from the qubits' local dimension.
  In numerical studies, it is found to be rather tightly given by $C = 22 \log(2)$. 
\end{proof}

The linear growth of $E(\rho)$ in time is tight, in the sense that there exist local nearest-neighbor Hamiltonians that exhibit such growth for suitably short times.
This claim largely follows from the results of Ref.~\cite{SchuchQuench}, which build on those in Ref.~\cite{AnalyticalQuench}.
Intuitively speaking, the linear growth originates from a light-cone-like entanglement dynamics governed by Lieb-Robinson bounds.
For free bosons and fermions (called noninteracting bosons and fermions in this context), such bounds are readily computable.

\begin{corollary}[Lower bound on $E(\rho)$ under local Hamiltonian evolution]
\noproofref
Let $\rho$ denote any quantum state of $S$.
There exist local Hamiltonians and product initial conditions $\rho$, with $E(\rho)=0$, such that
  \begin{equation}
    E`*( \ee^{-iHt }\rho \ee^{iHt} ) \geq (n-1) `*[ \frac{4}{3\pi} - \frac{1}{2}\log(t)-1 ] \ . 
  \end{equation}
\end{corollary}

\begin{proof}
  Let $S$ consist of an odd number $n\geq 21$ of qubits.
  Let $H$ be the Ising Hamiltonian with periodic boundary conditions.
  Let the initial state be $\rho = \proj{1,1,\ldots, 1}$, so that $E(\rho)=0$.
  The time-evolved state is $\rho(t) = e^{-iHt} \rho e^{iHt}$.
  For this state, bounds on entanglement entropies for sufficiently large subsystems are known.
  Since we consider only bipartite cuts of $S$, the subsystem of length $L$ considered in Ref.~\cite{SchuchQuench} always satisfies $L=(n-1)/2$.
  Hence, the constraint $L\geq 10$ is always satisfied.
  It follows from Theorem 1 in Ref.~\cite{SchuchQuench}, applied to the $n-1$ cuts of $S$, that
  \begin{equation}
    E \bm{(} \rho(t) \bm{)} \geq (n-1) `*[\frac{4}{3\pi} - \frac{1}{2}\log(t) - 1 ] \ . 
  \end{equation}	 
\end{proof}

\section{Recovering the hypothesis-testing (relative) entropy at high complexities}
\label{appx-sec:Dhypr-asymptotic-properties}

Consider an ``unrestricted'' agent who can render POVM effects of arbitrarily high complexities.
In this Appendix, we formulate conditions under which the complexity entropy approximates the hypothesis-testing entropy.
We first identify the set of measurements that an unrestricted agent can render with arbitrary precision.

\begin{definition}[POVM effects accessible with unbounded complexity]
  Let $\{\Mr[r]\}$ denote any family of POVM-effect-complexity sets (\cref{defn:Mr-sets-general}).
  The set of \emph{POVM effects accessible with unbounded complexity} is
  \begin{align}
    \Mr[\infty] \coloneqq \overline{ \bigcup_{r\geq 0} \Mr[r] } \ .
      \label{eq:POVM-effects-accessible-with-unbounded-complexity}
  \end{align}
  The $\overline{A}$ denotes the topological closure of a set $A$.
\end{definition}

\begin{definition}[``Complexity-unrestricted'' entropy and relative entropy]
  Let $\Mr[\infty]$ denote any set of POVM effects defined by~\eqref{eq:POVM-effects-accessible-with-unbounded-complexity}.
  Let $\rho$ denote any subnormalized state; and $\Gamma$, any positive-semidefinite operator.
  Let $\eta \in (0, \tr(\rho) ]$.
  We define the \textit{``complexity-unrestricted'' relative entropy} $\DHypr[\infty][\eta]{\rho}{\Gamma}$ by replacing $\Mr[r]$ with $\Mr[\infty]$ in the definition of $\DHypr[r][\eta]{\rho}{\Gamma}$ (\cref{def:defn-DHypr}).
  The \textit{``complexity-unrestricted'' entropy} is $\HHypr[\infty][\eta]{\rho} \coloneqq -\DHypr[\infty][\eta]{\rho}{\Ident}$.
\end{definition}

If $\Mr[\infty]$ contains all POVM effects, we recover the hypothesis-testing (relative) entropy exactly:
$\HHypr[\infty][\eta]{\rho} = \HHyp[\eta]{\rho}$ $`[ \DHypr[\infty][\eta]{\rho}{\Gamma} = \DHyp[\eta]{\rho}{\Gamma} ]$.
Yet, for many natural families $\{\Mr[r]\}$ of POVM-effect-complexity sets, $\Mr[\infty]$ contains only a limited portion of all POVM effects.
For instance, let $C$ denote a unitary-complexity measure (\cref{defn:unitary-complexity}), and let $\Psimple$ denote a set of simple POVM effects (\cref{defn:Psimple-general}) consisting only of projectors.
Suppose that $`*{ \Mr[r] } = `*{ \Mr[r](\Psimple,C) }$, as per~\eqref{eq:Mr-general-from-Psimple-and-superop-cplx-measure}.
Since every unitary transformation maps projectors to projectors, $\Mr[\infty]$ is a set of projectors and therefore excludes almost all POVM effects.

However, even when limited in form, $\Mr[\infty]$ may contain enough effects to ensure that $\HHypr[r][\eta]{\rho}$ converges to a value \emph{near} $\HHyp[\eta]{\rho}$ as $r \to \infty$.
We now formulate conditions under which $\HHypr[\infty][\eta]{\rho}$ $`[ \DHypr[\infty][\eta]{\rho}{\Gamma} ]$ lies close to the hypothesis-testing (relative) entropy, $\HHyp[\eta]{\rho}$ $`[ \DHyp[\eta]{\rho}{\Gamma} ]$.

\begin{definition}[$q$-quasiuniversal sets of POVM effects]
  \label{def:quasi-universal-set}
  Let $q \geq 0$, and let $\Gamma$ denote any positive-semidefinite operator.
  Let $\Mr[\infty]$ denote any set of POVM effects defined by~\eqref{eq:POVM-effects-accessible-with-unbounded-complexity}.
  $\Mr[\infty]$ is \emph{$q$-quasiuniversal with respect to $\Gamma$} if, for all POVM effects $Q$ with $\opnorm{Q}=1$ and for all probability weights $g \in (0, 1)$, there exist accessible effects $\tilde{Q}, \tilde{Q}' \in \Mr[\infty]$ such that
  \begin{subequations}
    \label{eq:Mr-quasi-universality-conditions}
    \begin{gather}
      g \tilde{Q} \leq Q \leq (1-g)\tilde{Q}' + g \Ident\ 
      \label{eq:Mr-quasi-universality-conditions--sandwich-POVM-effects}
      \\
      \text{and }\hspace{0.05cm}
      \log \bm{(} \tr`(\tilde{Q}'\Gamma) \bm{)} - \log \bm{(} \tr`(\tilde{Q}\Gamma) \bm{)} \leq q\ .
      \label{eq:Mr-quasi-universality-conditions--entropy-gap}
    \end{gather}
  \end{subequations}
\end{definition}

Consider any $Q$ with $\opnorm{Q}=1$.
Whenever $\Mr[\infty]$ is quasiuniversal (for any $q$), one can ``underestimate'' $Q$ by forming a convex mixture from an accessible effect $\tilde{Q} \in \Mr[\infty]$ and the trivial effect $0$.
Similarly, one can ``overestimate'' $Q$ by mixing an accessible $\tilde{Q}'$ and the trivial effect $\Ident$.
One can estimate $Q$ with such mixtures, no matter the mixing probability $g$.
The variable $g$ continuously parametrizes inequalities that approach $0 \leq Q \leq \tilde{Q}'$ as $g \to 0$ and $\tilde{Q} \leq Q \leq \Ident$ as $g \to 1$.
These boundaries provide intuition about the parametrization: as one requirement (that $\tilde{Q}$ lower-bound $Q$) tightens, the other requirement (that $\tilde{Q}'$ upper-bound $Q$) relaxes and vice versa.
We exclude the probability weight $g=0$ ($g=1$) to avoid trivializing the role of $\tilde{Q}$ ($\tilde{Q}')$.

The estimations $\tilde{Q}$ and $\tilde{Q}'$ lie close to $Q$ because they lie close to one another, by condition~\eqref{eq:Mr-quasi-universality-conditions--entropy-gap}:
their candidate hypothesis-testing-entropy values $\log \bm{(} \tr`( \tilde{Q} \Gamma ) / \eta \bm{)}$ and $\log \bm{(} \tr`( \tilde{Q}' \Gamma ) / \eta \bm{)}$ differ by $\leq q$ for every error intolerance $\eta \in (0 , 1]$.
Without~\eqref{eq:Mr-quasi-universality-conditions--entropy-gap}, one might trivially satisfy~\eqref{eq:Mr-quasi-universality-conditions--sandwich-POVM-effects} with $\tilde{Q} = 0$ and $\tilde{Q}' = \Ident$.
Finally, as expected, $\Mr[\infty]$ is $q$-quasiuniversal for every $q$ and with respect to every $\Gamma$, whenever the set contains all POVM effects.
In this case, one can choose $\tilde{Q} = \tilde{Q}' = Q \in \Mr[\infty]$.

The hypothesis-testing (relative) entropy bounds the complexity-unrestricted (relative) entropy from above and below.

\begin{theorem}[Recoverability of the hypothesis-testing entropy]
  \label{thm:DHypr-recover-DHyp-r-to-infty}
  Let $\rho$ denote any subnormalized state; and $\Gamma$, any positive-semidefinite operator.
  Let $\Mr[\infty]$ denote any set of POVM effects defined by~\eqref{eq:POVM-effects-accessible-with-unbounded-complexity}.
  Let $q \geq 0$, $\eta \in (0, \tr(\rho) ]$, and $g \in (0 , 1)$.
  Suppose that $\Mr[\infty]$ is $q$-quasiuniversal with respect to $\Gamma$ (\cref{def:quasi-universal-set}).
  Then
  \begin{align}
      \DHyp[\eta]{\rho}{\Gamma}
      \geq \DHypr[\infty][\eta]{\rho}{\Gamma}
      \geq \DHyp[\eta_g]{\rho}{\Gamma} - q - \log`*( \frac{\tr`(\rho)}{g\eta} ) \ ,
        \label{eq:thm-DHypr-recover-DHyp-r-to-infty--bounds}
  \end{align}
  wherein $\eta_g \coloneqq (1-g) \eta + g \tr`(\rho) \geq \eta$.
  If $\Mr[\infty]$ is $q$-quasiuniversal with respect to $\Ident$, then
  \begin{align}
      \HHyp[\eta]{\rho}
      \leq \HHypr[\infty][\eta]{\rho}
      \leq \HHyp[\eta_g]{\rho} + q + \log`*( \frac{\tr`(\rho)}{g\eta} ) \ .
  \end{align}
\end{theorem}

\begin{proof}[**thm:DHypr-recover-DHyp-r-to-infty]
  The $\DHyp[\eta]{\rho}{\Gamma}$ optimization ranges over all POVM effects, while the $\DHypr[\infty][\eta]{\rho}{\Gamma}$ optimization ranges over only effects in $\Mr[\infty]$.
  Therefore, $\DHyp[\eta]{\rho}{\Gamma} \geq \DHypr[\infty][\eta]{\rho}{\Gamma}$.
  Now, we prove the second inequality in~\eqref{eq:thm-DHypr-recover-DHyp-r-to-infty--bounds}.
  Let $Q$ denote an optimal POVM effect for $\DHyp[\eta_g]{\rho}{\Gamma}$ in~\eqref{eq:DHyp-alternative-expressions--ratio-trQGamma-trQrho}.
  Without loss of generality, $\opnorm{Q} = 1$, since we could otherwise replace $Q$ by another candidate $Q' \coloneqq Q/\opnorm{Q} \geq Q$ that achieves the same objective value as $Q$.
  Moreover, $\tr`(Q\rho) \geq \eta_g$ and $\tr`(Q\Gamma)/\tr`(Q\rho) = \exp\bm{(} -\DHyp[\eta_g]{\rho}{\Gamma} \bm{)}$.
  Since $\Mr[\infty]$ is $q$-quasiuniversal, some $\tilde{Q},\tilde{Q}' \in \Mr[\infty]$ satisfy the conditions~\eqref{eq:Mr-quasi-universality-conditions}.
  Condition~\eqref{eq:Mr-quasi-universality-conditions--sandwich-POVM-effects} implies that
  \begin{align}
    (1-g)\tr`(\tilde{Q}'\rho)
    \geq \tr`(Q\rho) - g \tr`(\rho)
    \geq \eta_g - g \tr`(\rho) 
    = (1-g)\eta \ .
  \end{align}
  Equivalently, $\tr`(\tilde{Q}'\rho) \geq \eta$.
  Therefore, $\tilde{Q}'$ is a candidate for the $\DHypr[\infty][\eta]{\rho}{\Gamma}$ optimization in~\eqref{eq:defn-DHypr}, so
  \begin{align}
    \DHypr[\infty][\eta]{\rho}{\Gamma}
    &\geq - \log`*( \frac{\tr`*(\tilde{Q}'\Gamma)}{\tr`*(\tilde{Q}'\rho)} )
    \geq - \log`*( \frac{\tr`*(\tilde{Q}'\Gamma)}{\eta} )
    \geq - \log`*( e^q \frac{ \tr`*(\tilde{Q}\Gamma) }{\eta} )
    \geq  - \log`*( \frac{\tr`*(Q\Gamma)}{g\eta} ) - q
      \nonumber\\
    &\geq - \log`*( \frac{\tr`*(Q\Gamma)}{\tr`*(Q\rho)} \; \frac{\tr`*(\rho)}{g\eta} ) - q
    = \DHyp[\eta_g]{\rho}{\Gamma} - q - \log`*( \frac{\tr`*(\rho)}{g\eta} )\ .
  \end{align}
  The third inequality holds because $\tr`(\tilde{Q}'\Gamma) \leq e^q \tr`(\tilde{Q}\Gamma)$, by condition~\eqref{eq:Mr-quasi-universality-conditions--entropy-gap}.
  The fourth inequality follows from condition~\eqref{eq:Mr-quasi-universality-conditions--sandwich-POVM-effects};
  the fifth from $\tr`(Q\rho)\leq \tr`(\rho)$.
\end{proof}

We now provide an example of quasiuniversal set that arises naturally in quantum computation.
\begin{proposition}[Example of $q$-quasiuniversal set]
  \label{thm:example-quasi-universal-set}
  Let $\mathcal{G}$ denote any set of unitary gates that is universal for quantum computation.
  Let $C_{\mathcal{G}}$ be the unitary-circuit-complexity measure associated with $\mathcal{G}$ (\cref{defn:circuit-complexity-measure}).
  Let $\Psimple$ denote a set of simple POVM effects (\cref{defn:Psimple-general}) that consists only of projectors and contains at least one rank-1 projector.
  Let $`*{ \Mr[r] = \Mr[r]`*( \Psimple,C_{\mathcal{G}} ) }$ denote the family of POVM-effect-complexity sets defined by~\eqref{eq:Mr-general-from-Psimple-and-superop-cplx-measure}.
  Let $\Mr[\infty]$ denote the set of POVM effects defined in~\eqref{eq:POVM-effects-accessible-with-unbounded-complexity}, and assume that the effects in $\Mr[\infty]$ act on a Hilbert space of dimensionality $d$.
  $\Mr[\infty]$ is $q$-quasiuniversal with respect to $\Ident$ (\cref{def:quasi-universal-set}).
  The $q$ denotes the largest gap within the set $\mathcal{X} \coloneqq `*{ \log \bm{(} \tr(P) \bm{)} : P \in \Psimple } \subset [0,\log(d)]$:
  \begin{align}
    q \coloneqq \sup`*{ \abs{b-a}: (a,b) \subset [0,\log(d)] \setminus \mathcal{X} } \ .
    \label{eq:thm-Psimple-CU-is-q1-quasi-universal-wrt-Ident--entropy-gap}
  \end{align}
  $q$ is the greatest length of any interval $(a,b) \subset [0,\log(d)]$ disjoint from $\mathcal{X}$.
\end{proposition}

\begin{proof}[**thm:example-quasi-universal-set]
  Let $Q$ denote any POVM effect with $\opnorm{Q}=1$.
  Let us fix a matrix representation of $Q$ that is diagonal and has decreasing diagonal elements: $Q = \diag`(\lambda_1, \lambda_2, \ldots, \lambda_d)$, with $1 = \lambda_1 \geq \lambda_2 \geq \ldots \geq \lambda_d \geq 0$.
  Consider any $g \in (0, 1)$, and let $\ell_g \coloneqq \max`{ i: \lambda_i \geq g}$ denote the number of $Q$ eigenvalues not less than $g$.
  $\ell_g$ satisfies $1 \leq \ell_g \leq d$, since $\lambda_1 = 1 \geq g$.
  Let
  \begin{align}
    \ell_g^- \coloneqq \max`*{ \rank(P)\,:\; P \in \Psimple,\ \rank(P) \leq \ell_g }
    \; \; \; \; \text{and} \; \; \; \;
    \ell_g^+ \coloneqq \min`*{ \rank(P)\,:\; P \in \Psimple,\ \rank(P) \geq \ell_g } \ .
  \end{align}
  As per the definitions of $\ell_g^-$ and $\ell_g^+$, there exist projectors $\tilde{P}, \tilde{P}' \in \Psimple$ such that $\rank(\tilde{P}) = \ell_g^-$ and $\rank(\tilde{P}') = \ell_g^+$.
  Moreover, $1 \leq \ell_g^- \leq \ell_g^+ \leq d$, since $\Psimple$ contains both a rank-1 projector and $\Ident$.
  Let $U$ denote a unitary such that $\tilde{Q} = U^\dagger \, \tilde{P} \, U$ is diagonal: $\tilde{Q} = \diag(1,\ldots, 1, 0, \ldots, 0)$, with $\rank(\tilde{Q}) = \ell_g^-$.
  Likewise, let $U'$ denote a unitary such that $\tilde{Q} = U'^\dagger \, \tilde{P}' \, U'$ is diagonal, with $\rank(\tilde{Q}') = \ell_g^+$.
  By assumption, $U$ and $U'$ are arbitrarily well-approximated by unitaries of finite complexity, so $\tilde{Q},\tilde{Q}' \in \Mr[\infty]$.
  We now show that $\tilde{Q}$ and $\tilde{Q}'$ satisfy the conditions~\eqref{eq:Mr-quasi-universality-conditions} in the definition of a $q$-quasiuniversal set.

  Consider the operator $g \tilde{Q}$.
  The first $\ell_g \geq \ell_g^-$ diagonal entries of $g \tilde{Q}$ do not exceed $g$ and, hence, do not exceed the first $\ell_g$ diagonal entries of $Q$.
  Moreover, the last $\ell_g$ diagonal entries of $g \tilde{Q}$ equal 0 and, hence, do not exceed the last $\ell_g$ diagonal entries of $Q$.
  Thus, $g \tilde{Q} \leq Q$.
  Now, consider the operator $(1-g) \tilde{Q}' + g \Ident$.
  The first $\ell_g \leq \ell_g^+$ diagonal entries of $(1-g) \tilde{Q}' + g \Ident$ equal 1 and, hence, are not less than the first $\ell_g$ diagonal entries of $Q$.
  Moreover, the last $\ell_g$ diagonal entries of $g \tilde{Q}$ are not less than $g$ and, hence, are not less than the last $\ell_g$ diagonal entries of $Q$.
  Thus, $Q \leq (1-g) \tilde{Q}' + g \Ident$.
  As such, $\tilde{Q}$ and $\tilde{Q}'$ satisfy condition~\eqref{eq:Mr-quasi-universality-conditions--sandwich-POVM-effects}.

  Compare the definitions of $\ell_g^-$ and $\ell_g^+$.
  There exists no projector $P \in \Psimple$ such that $\ell_g^- < \rank(P) < \ell_g^+$.
  Accordingly, the interval $\bm{(} \log`(\ell_g^-) , \log`(\ell_g^+) \bm{)}$ is disjoint from $\mathcal{X}$.
  $[$If $\ell_g^- = \ell_g^+$, then $\bm{(} \log`(\ell_g^-) , \log`(\ell_g^+) \bm{)} = \emptyset$ by convention.$]$
  Thus, $\log \bm{(} \tr`(\tilde{Q}') \bm{)} - \log \bm{(} \tr`(\tilde{Q}) \bm{)} = \log`(\ell_g^+) - \log`(\ell_g^-) \leq q$, with $q$ defined as in~\eqref{eq:thm-Psimple-CU-is-q1-quasi-universal-wrt-Ident--entropy-gap}, so $\tilde{Q}$ and $\tilde{Q}'$ satisfy condition~\eqref{eq:Mr-quasi-universality-conditions--entropy-gap}.
\end{proof}

\section{Data compression under complexity limitations}
\label{appx-topic:data-compression-complexity-limitations}

Here, we establish optimal protocols for data compression under complexity limitations.
We thereby prove the bounds~\eqref{eq:optimal-qubit-compression-number} in \cref{sec:data-compression-under-computational-limitations}.

Let $\rho$ denote any state, and $C$ any unitary-complexity measure, of $n$ qubits.
Let $r \geq 0$.
Let $m \leq n$ denote a number of qubits; and $\epsilon \in [0, 1)$, an error parameter.
We define an \emph{$(m, r,\epsilon)$--data-compression protocol} as a unitary $U$ that satisfies $C(U) \leq r$ and that can compress $\rho$ into $m$ qubits.
The compression succeeds if the other $n-m$ qubits end in a state $\epsilon$-close to $\ket{0^{n-m}}$ in fidelity:
there exists a subset $\mathcal{W} \subset `{ 1, 2, \ldots, n}$ of qubits such that $\abs{\mathcal{W}} = m$ and
\begin{align}
  \tr \bm{(} \tr_\mathcal{W}(U \rho U^\dagger) \,
  \proj{0^{n-m}} \bm{)} \geq 1-\epsilon \ .
  \label{defn-data-compression-protocol--compression-quality}
\end{align}

We first prove that some protocol compresses $\rho$ into a number of qubits proportional to the reduced complexity entropy (\cref{defn:Hhypr}).

\begin{theorem}[Achievability of data compression]
  \label{thm:complexity-entropy-data-compression-achievability}
  Let $\rho$ denote any state, and $C$ any unitary-complexity measure (\cref{defn:unitary-complexity}), of $n$ qubits.
  Let $\Psimple$ denote the set of simple POVM effects (\cref{defn:Psimple-general}) defined in~\eqref{eq:setting-defn-Mrzero}.
  Let $r \geq 0$ and $\epsilon \in [0, 1)$.
  Let $\Mr[r] = \Mr[r](\Psimple,C)$ denote the POVM-effect-complexity set defined in~\eqref{eq:Mr-general-from-Psimple-and-superop-cplx-measure}.
  There exists an $(m, r,\epsilon)$--data-compression protocol for $\rho$ such that
  \begin{align}
      \label{eq:thm-complexity-entropy-data-compression-achievability--limit}
    m = \frac1{\log(2)} \, \Hhypr[r][1-\epsilon]{\rho} \ .
  \end{align}
  The reduced complexity entropy $\Hhypr[r][1-\epsilon]{}$ is defined with respect to $\Mr[r]$.
\end{theorem}

The above theorem can be reformulated in terms of the complexity entropy $\HHypr[r][1-\epsilon]{\rho}$.
\cref{thm:relation-Dhypr-DHypr} shows that the reduced complexity entropy $\Hhypr[r][1-\epsilon]{\rho}$ differs from the complexity entropy $\HHypr[r][1-\epsilon]{\rho}$ by at most $\log(1/[1-\epsilon])$.
Combining this fact with \cref{thm:complexity-entropy-data-compression-achievability} yields the bounds~\eqref{eq:optimal-qubit-compression-number} in \cref{sec:data-compression-under-computational-limitations}.
If $\epsilon \approx 0$, then the optimal number of qubits in~\eqref{eq:thm-complexity-entropy-data-compression-achievability--limit} is approximately proportional to the complexity entropy: $m \approx \HHypr[r][1-\epsilon]{\rho} / \log(2)$.

We further prove that the protocol in \cref{thm:complexity-entropy-data-compression-achievability} is optimal.

\begin{theorem}[Optimality of data compression]
  \label{thm:complexity-entropy-data-compression-optimality}
  Let $\rho$ denote any state of $n$ qubits.
  Let $r \geq 0$ and $\epsilon \in [0, 1)$.
  Let $\Mr[r]$ denote a POVM-effect-complexity set as in \cref{thm:complexity-entropy-data-compression-achievability}.
  Every $(m, r,\epsilon)$--data-compression protocol satisfies
  \begin{align}
    m \geq \frac1{\log(2)} \, \Hhypr[r][1-\epsilon] {\rho} \ .
  \end{align}
  The reduced complexity entropy $\Hhypr[r][1-\epsilon]{}$ is defined with respect to $\Mr[r]$.
\end{theorem}

The above data-compression problem maps directly onto the thermodynamic-erasure problem in \cref{sec:main-erasure-Wcost-degenerate-H}.
The proofs of \cref{thm:complexity-entropy-data-compression-achievability,thm:complexity-entropy-data-compression-optimality} are closely related to the arguments in \cref{sec:main-erasure-Wcost-degenerate-H} and to the proof of \cref{thm:erasure-nontrivial-Ham-Wrstar}.

\begin{proof}[*thm:complexity-entropy-data-compression-achievability]
  Consider any $\zeta > 0$.
  There exists a $Q \in \Mr[r]$ such that $\tr`(Q\rho) \geq 1-\epsilon$ and $\log \bm{(}  \tr`(Q) \bm{)} \leq \Hhypr[r][1-\epsilon]{\rho} + \zeta$.
  $Q = U^\dagger P U$ for some unitary $U$ satisfying $C(U) \leq r$ and for some simple POVM effect $P$ acting as $\Ident_2$ on every qubit in a subset $\mathcal{W} \subset `{ 1, 2, \ldots, n}$ and as $\proj{0}$ on every other qubit.
  Without loss of generality, assume that $\mathcal{W}$ consists of the first $m \coloneqq \abs{\mathcal{W}}$ of the $n$ qubits.
  Hence $P = \Ident_2^{\otimes m} \otimes \proj{0^{n-m}}$.

  $U$ is an $(m, r,\epsilon)$--data-compression protocol for $\rho$, since $C(U) \leq r$ and condition~\eqref{defn-data-compression-protocol--compression-quality} holds:
  \begin{align}
    \tr \bm{(} \tr_{\mathcal{W}}(U \rho U^\dagger) \, \proj{0^{n-m}} \bm{)}
    = \tr ( [U \rho U^\dagger] \, [ \Ident_2^{\otimes m} \otimes \proj{0^{n-m}} ] )
    =  \tr ( [U \rho U^\dagger]\, P )
    =  \tr`( Q \rho )
    &\geq 1 - \epsilon \ .
  \end{align}
  Moreover,
  \begin{align}
    m 
    = \log_2 \bm{(}  \tr`(P) \bm{)}
    = \log_2 \bm{(}  \tr`(Q) \bm{)}
    = \frac1{\log(2)} \, \log \bm{(}  \tr`(Q) \bm{)}
    \leq \frac1{\log(2)} \left( \Hhypr[r][1-\epsilon]{\rho} + \zeta \right) \ ,
  \end{align}
  which implies~\eqref{eq:thm-complexity-entropy-data-compression-achievability--limit}, since $\zeta$ is arbitrary.
\end{proof}

\begin{proof}[*thm:complexity-entropy-data-compression-optimality]
  Let $U$ denote any $(m,r,\epsilon)$--data-compression protocol for $\rho$.
  By definition, there exists a subset $\mathcal{W} \subset `{1, 2, \ldots, n}$ such that $\abs{\mathcal{W}} = m$ and condition~\eqref{defn-data-compression-protocol--compression-quality} holds.
  Without loss of generality, assume that $\mathcal{W}$ consists of the first $m$ of $n$ qubits.
  Let $P \coloneqq \Ident_2^{\otimes m} \otimes \proj{0^{n-m}}$ and $Q \coloneqq U^\dagger P U$.
  By construction, $Q$ is in $\Mr[r]$ and satisfies
  \begin{align}
    \tr`(Q \rho)
    = \tr`(P\, U \rho U^\dagger)
    = \tr \bm{(} \tr_{\mathcal{W}}`( U \rho U^\dagger ) \, \proj{0^{n-m}} \bm{)}
    \geq 1 - \epsilon \ .
  \end{align}
  $Q$ is therefore a candidate for the $\Hhypr[r][1-\epsilon]{\rho}$ optimization, so
  \begin{align}
    m = \abs{\mathcal{W}} = \log_2 \bm{(} \tr`(P) \bm{)} = \log_2 \bm{(} \tr`(Q) \bm{)}
    \geq \frac1{\log(2)} \Hhypr[r][1-\epsilon] {\rho}\ .
  \end{align}
\end{proof}

\section{Decoupling from a reference system}
\label{sec:appendix-DHypr-randomness-extraction}

Here, we prove the decoupling bound of \cref{mainthm:DecouplingLowerBound}.
Our proof relies on the following conjectured property of the complexity conditional entropy.

\begin{conjecture}
  \noproofref
  \label{conj:HHyprc-ptrace-upper-bound}
  Let $A$, $B$, and $R$ denote distinct quantum systems.
  Let $\rho_{ABR}$ denote any state of $ABR$.
  Let $r\geq 0$ and $\eta\in(0,1]$.
  It holds that
  \begin{align}
    \HHyprc[r][\eta][\rho]{B}[R]
    \leq
    \HHyprc[r][\eta][\rho]{AB}[R]
    + \log(d_A)\ .
    \label{eq:HHyprc-ptrace-upper-bound--HHyprc}
  \end{align}
  Equivalently,
  \begin{align}
    \DHypr[r][\eta]{\rho_{BR}}{\pi_B\otimes\rho_R}
    \geq
    \DHypr[r][\eta]{\rho_{ABR}}{\pi_{AB}\otimes\rho_R} - 2\log(d_A)\ .
    \label{eq:HHyprc-ptrace-upper-bound--DHypr}
  \end{align}
\end{conjecture}
The equivalence between~\eqref{eq:HHyprc-ptrace-upper-bound--HHyprc} and~\eqref{eq:HHyprc-ptrace-upper-bound--DHypr} holds because $\HHyprc[r][\eta][\rho]{B}[R] = \log(d_B) - \DHypr[r][\eta]{\rho_{BR}}{\pi_B\otimes\rho_R}$ and because $\HHyprc[r][\eta][\rho]{AB}[R] = \log(d_Ad_B) - \DHypr[r][\eta]{\rho_{ABR}}{\pi_{AB}\otimes\rho_R}$, by the argument in~\eqref{eq:alternative-form-of-HHyprc}.
\eqref{eq:HHyprc-ptrace-upper-bound--HHyprc} mirrors a known inequality for the von Neumann conditional entropy $\HH[\rho]{AB}[R]$.
The latter inequality follows from a chain rule and the lower bound $\HH[\rho]{A}[BR]\geq-\log(d_A)$:
\begin{align}
  \HH[\rho]{B}[R]
  = \HH[\rho]{AB}[R] - \HH[\rho]{A}[BR]
  \leq \HH[\rho]{AB}[R] + \log(d_A) \ .
\end{align}

\begin{theorem}[Upper bound on the number of qubits Alice can decouple under complexity limitations]
    \label{thm:DHypr-randomness-extraction--lower-bound-k}
    Let $A$ denote an $n$-qubit system, and let $R$ denote a quantum system distinct from $A$.
    Let $\rho_{AR}$ denote any state of $AR$.
    Let $\Psimple = \Psimple[AR]$ denote any set of simple POVM effects (\cref{defn:Psimple-general}); and $C = C_{AR}$, any adjoint-invariant unitary-complexity measure (\cref{defn:unitary-complexity}).
    Let $`*{ \Mr[r] = \Mr[r](\Psimple,C) }$ denote the family of POVM-effect-complexity sets defined by~\eqref{eq:Mr-general-from-Psimple-and-superop-cplx-measure}.
    Let $r_0 \geq 0$ and $r_1 \geq r_0$.
    Let $\eta\in(0,1]$ and $\delta \in (0,1]$.
    Let $U_0$ denote any unitary on $A$ satisfying $C(U_0) \leq r_0$.
    Let $A_1$ denote a subsystem of $k \geq 0$ qubits in $A$; and $A_2$, the subsystem of the other $n-k$ qubits.
    Let
    \begin{align}
        \rho'_{AR} &\coloneqq (U_0 \otimes \Ident_R) \rho_{AR} (U_0 \otimes \Ident_R)^\dagger \ .
    \end{align}
    Suppose that
    \begin{align}
        \DHypr[r_1][\eta]{\rho'_{A_2R}}{\pi_{A_2}\otimes\rho_R}
        \leq -\log \left( \frac\delta\eta \right) \ ,
        \label{eq:thm-DHypr-randomness-extraction--lower-bound-k-assumption-DHypr}
    \end{align}
    wherein $\DHypr[r_1][\eta]{}{}$ is defined with respect to $\Mr[r_1]$.
    Assume \cref{conj:HHyprc-ptrace-upper-bound} holds.
    It holds that
    \begin{align}
        k \geq \frac12 \left[ n -\frac1{\log(2)}
            \HHyprc[r_1-r_0][\eta][\rho]{A}[R] +\log_2 \left( \frac\delta\eta \right)
        \right] \ ,
      \label{eq:thm-DHypr-randomness-extraction-lower-bound-k}
    \end{align}
    wherein $\HHyprc[r_1-r_0][\eta][]{A}[R]$ is defined with respect to $\Mr[r_1-r_0]$.
\end{theorem}

\begin{proof}[**thm:DHypr-randomness-extraction--lower-bound-k]
  The following chain of inequalities implies~\eqref{eq:thm-DHypr-randomness-extraction-lower-bound-k}:
  \begin{align}
    -\log \left( \frac\delta\eta \right)
    &\geq \DHypr[r_1][\eta]{\rho'_{A_2R}}{\pi_{A_2}\otimes\rho_R}
    \nonumber\\
    &= \log \left( d_{A_2} \right) - \HHyprc[r_1][\eta][\rho']{A_2}[R]
    \nonumber\\
    &\geq \log \left( d_{A_2} \right) - \HHyprc[r_1][\eta][\rho']{A}[R] - \log \left( d_{A_1} \right)
    \nonumber\\
    &= (n-2k)\log(2) - \HHyprc[r_1][\eta][\rho']{A}[R]
    \nonumber\\
    &\geq (n-2k)\log(2) - \HHyprc[r_1-r_0][\eta][\rho]{A}[R] \ .
  \end{align}
  The first equality follows by the argument in~\eqref{eq:alternative-form-of-HHyprc}; the second inequality, by the assumption that \cref{conj:HHyprc-ptrace-upper-bound} holds; and the last inequality, by \cref{thm:DHypr-arg-U-rp}.
\end{proof}



\def\dotbibolamazi{.bibolamazi}
\def\selectlanguage#1{}
\bibliography{Bibliography\dotbibolamazi}


\onecolumngrid

\appendix

\renewcommand{\thesection}{\Alph{section}}
\renewcommand{\thesubsection}{\Alph{section} \arabic{subsection}}
\renewcommand{\thesubsubsection}{\Alph{section} \arabic{subsection} \roman{subsubsection}}

\makeatletter\@addtoreset{equation}{section}
\def\theequation{\thesection\arabic{equation}}

\end{document}